\documentclass[11pt]{article}
\usepackage{enumerate}
\usepackage{url}
\usepackage{pdflscape} 
\usepackage{rotating} 
\usepackage{booktabs}
\usepackage{threeparttable} 

\usepackage{geometry}
\RequirePackage{amsthm,amsmath,amsfonts,amssymb}
\usepackage[authoryear, semicolon, sort&compress]{natbib}
\usepackage{xcolor}
\usepackage{xr-hyper}
\usepackage[pdftex,bookmarks,colorlinks,breaklinks]{hyperref} 
\usepackage{pifont}
\newcommand{\cmark}{\ding{51}}
\newcommand{\xmark}{\ding{55}}
\definecolor{dullmagenta}{rgb}{0.4,0,0.4}
\definecolor{darkblue}{rgb}{0,0,0.4}
\definecolor{navyblue}{rgb}{0,0,0.5}
\hypersetup{linkcolor=blue,citecolor=blue,filecolor=blue,urlcolor=blue} 
\definecolor{ao(english)}{rgb}{0.0, 0.5, 0.0}
\definecolor{coquelicot}{rgb}{0.90, 0.42, 0.72}
\definecolor{burntorange}{rgb}{0.8, 0.33, 0.0}
\newcommand{\green}{\color{ao(english)}}
\newcommand{\orange}{\color{orange}}
\definecolor{burntblue}{RGB}{0, 114, 206}

\newcommand\independent{\protect\mathpalette{\protect\independenT}{\perp}}
\def\independenT#1#2{\mathrel{\rlap{$#1#2$}\mkern2mu{#1#2}}}
\usepackage{epstopdf}
\usepackage{tablefootnote}
\usepackage{multirow}
\usepackage{mathtools}
\usepackage{bbm}
\usepackage{graphicx} 
\usepackage{flafter}
\usepackage{subcaption}
\usepackage{graphicx}
\usepackage{tikz}
\usetikzlibrary{positioning}
\usetikzlibrary{arrows.meta,arrows}
\usetikzlibrary{arrows,decorations.markings,shapes,arrows,fit}
\tikzset{box/.style={draw, minimum size=2em, text width=4.5em, text centered},
	bigbox/.style={draw, inner sep=20pt,label={[shift={(-3ex,3ex)}]south east:#1}}
}
\usepackage{bm}
\usepackage{colortbl}
\usepackage{arydshln}
\usepackage{caption}
\usepackage{algorithm}
\usepackage{algpseudocode}
\usepackage{program}
\numberwithin{table}{section}
\numberwithin{equation}{section}
\usepackage[nodisplayskipstretch]{setspace}
\usepackage{enumitem}

\numberwithin{table}{section}
\numberwithin{equation}{section}

\def\bX{\mathbf{X}}

\def\bz{\mathbf{z}}
\def\bZ{\mathbf{Z}}

\def\bU{\mathbf{U}}
\def\bv{\mathbf{v}}

\def\ba{\mathbf{a}}
\def\K{\mathbb{K}}
\def\E{\mathbb{E}}

\def\P{\mathbb{P}}

\def\balpha{\boldsymbol{\alpha}}
\def\bbeta{\boldsymbol{\beta}}

\def\bnabla{\boldsymbol{\nabla}}

\def\etabf{\boldsymbol{\eta}}
\def\etabftil{\widetilde{\etabf}}
\def\Ktil{\widetilde{\K}}
\def\Ctil{\widetilde{\C}}
\def\balphahat{\widehat{\balpha}}
\def\bbetahat{\widehat{\bbeta}}

\def\etabfhat{\widehat{\etabf}}
\def\varepsilonhat{\widehat{\varepsilon}}
\def\zetahat{\widehat{\zeta}}
\def\bzero{\boldsymbol{0}}

\def\independenT#1#2{\mathrel{\rlap{$#1#2$}\mkern4mu{#1#2}}}

\def\Cov{\mbox{Cov}}

\def\bzero{\mathbf{0}}

\def\bS{\mathbf{S}}
\def\bs{\mathbf{s}}

\def\S{\mathbb{S}}

\def\bDelta{\boldsymbol{\Delta}}

\def\bW{\mathbf{W}}
\def\bQ{\mathbf{Q}}

\def\bgamma{\boldsymbol{\gamma}}
\def\bdelta{\boldsymbol{\delta}}

\def\C{\mathbb{C}}

\def\R{\mathbb{R}}

\def\S{\mathbb{S}}

\def\Ytil{\widetilde{Y}}

\def\bbetatil{\widetilde{\bbeta}}
\def\bgammatil{\widetilde{\bgamma}}
\def\bdeltatil{\widetilde{\bdelta}}
\def\bDeltatil{\widetilde{\bDelta}}
\def\balphatil{\widetilde{\balpha}}
\def\zetatil{\widetilde{\zeta}}
\def\varepsilontil{\widetilde{\varepsilon}}

\def \hs2{\hspace{2mm}}

\definecolor{jcolor}{RGB}{041,122,000}
\definecolor{darkred}{RGB}{100,000,000}
\definecolor{purple}{RGB}{200,000,200}

\def\red{\color{red}}

\def\boxit#1{\vbox{\hrule\hbox{\vrule\kern6pt \vbox{\kern6pt#1\kern6pt}\kern6pt\vrule}\hrule}}

\def\be{\mathbf{e}}

\def\bS{\mathbf{S}}

\def\idf{\mathbbm1}

\def\bgammahat{\widehat{\bgamma}}
\def\bdeltahat{\widehat{\bdelta}}

\def\Ytil{\widetilde{Y}}

\def\bomega{\boldsymbol{\omega}}

\def\sigmahat{\widehat{\sigma}}

\def\thetahat{\widehat{\theta}}

\theoremstyle{plain}
\newtheorem{theorem}{Theorem}[section]

\newtheorem{lemma}[theorem]{Lemma}

\theoremstyle{remark}
\newtheorem{assumption}{Assumption}
\newtheorem{remark}{Remark}

\usepackage{xr}

\begin{document}

\date{}

\title{\bf Dynamic treatment effects: high-dimensional doubly robust inference under model misspecification}

\author{Yuqian Zhang\thanks{Institute of Statistics and Big Data, Renmin University of China} \and Weijie Ji\thanks{School of Statistics and Management, Shanghai University of Finance and Economics} \and Jelena Bradic\thanks{Department of Mathematics and Halicioglu Data Science Institute, University of California, San Diego, E-mail: \href{mailto:jbradic@ucsd.edu}{jbradic@ucsd.edu}}}

\maketitle

\begin{abstract}
Estimating dynamic treatment effects is crucial across various disciplines, providing insights into the time-dependent causal impact of interventions. However, this estimation poses challenges due to time-varying confounding, leading to potentially biased estimates. Furthermore, accurately specifying the growing number of treatment assignments and outcome models with multiple exposures appears increasingly challenging to accomplish. Double robustness, which permits model misspecification, holds great value in addressing these challenges. This paper introduces a novel “sequential model doubly robust” estimator. We develop novel moment-targeting estimates to account for confounding effects and establish that root-$N$ inference can be achieved as long as at least one nuisance model is correctly specified at each exposure time, despite the presence of high-dimensional covariates. Although the nuisance estimates themselves do not achieve root-\(N\) rates, the carefully designed loss functions in our framework ensure final root-\(N\) inference for the causal parameter of interest. Unlike off-the-shelf high-dimensional methods, which fail to deliver robust inference under model misspecification even within the doubly robust framework, our newly developed loss functions address this limitation effectively.
\end{abstract}


\section{Introduction}\label{sec:intro}
Statistical inference and estimation of causal relationships have a long-standing tradition. In various applications, data is collected dynamically over time, and individuals undergo treatments at multiple stages. Examples include mobile health datasets, electronic health records, and a broad range of studies from biomedicine and public health to political science. Conducting randomized controlled trials, especially those with multiple treatment stages, is often time-consuming, and results are not immediately available. Additionally, financial and ethical constraints frequently result in small sample sizes with unrepresentative populations. In contrast, observational studies provide a more accessible alternative, generating large-scale dynamic datasets with rich information. While observational studies have advantages in terms of economic efficiency, sample representativeness, and timely results, statistical analysis based on them is more challenging. Over time, confounding variables at each time point simultaneously affect future treatments and final outcomes. As a result, methods that are effective in randomized controlled trials, such as the two-sample t-test, generally suffer from non-negligible bias in observational studies.

Bias from dynamic confounders is a major challenge in causal inference. In large-scale dynamic studies, confounders often outnumber treatment-specific samples due to exponentially decreasing sizes across stages, leading to high-dimensional settings. Additionally, model misspecification complicates inference, particularly when earlier counterfactual models depend on later ones \citep{babino2019multiple}. We demonstrate that this challenge can be overcome by introducing new estimates of the underlying causal effects. 

\subsection{Estimation of causal effects under dynamic setups}

Consider a dynamic setting with binary treatments at two exposure times, $A_1$ and $A_2$, although our results extend to any finite number of exposures. We observe independent and identically distributed samples $\mathcal{S} = \{\mathbf{W}_i\}_{i=1}^N =\{ (Y_i, A_{1i}, A_{2i}, \mathbf{S}_{1i}, \mathbf{S}_{2i})\}_{i=1}^N$. Here, $Y \in \mathbb{R}$ denotes the observed outcome at the final stage. We assume the existence of potential outcome variables $Y(a_1, a_2)$ for each $(a_1, a_2) \in \{0, 1\}^2$, representing the outcome an individual would have experienced if exposed to a treatment path $(a_1, a_2)$. Before each exposure, we also collect confounders (or covariates), denoted by $\mathbf{S}_1 \in \mathbb{R}^{d_1}$ and $\mathbf{S}_2 \in \mathbb{R}^{d_2}$, respectively. Covariate history up to the second exposure is denoted by $\bar{\mathbf{S}}_2 := (\mathbf{S}_1^\top, \mathbf{S}_2^\top)^\top \in \mathbb{R}^d$, where the dimensions $d_1$ and $d := d_1 + d_2$ are potentially much larger than $N$. We consider observational studies that allow all variables evaluated at previous stages to potentially influence later ones, without relying on any Markov assumptions, as illustrated in Figure \ref{fig:diag}.

\begin{figure}[h]
	\scriptsize
	\centering
	\begin{tikzpicture}
		[
		roundnode/.style={circle, draw=red!20, fill=red!5, very thick, minimum size=7mm},
		circlenode/.style={circle, draw=blue!20, fill=blue!5, very thick, minimum size=7mm},
		];
		\node[circlenode] (s1) at (0,0) {$\bS_1$};
		\node[roundnode] (a1) at (1.5,2) {$A_1$};
		\node[circlenode] (s2) at (3,0.5) {$\bS_2$};
		\node[roundnode] (a2) at (4.5,2) {$A_2$};
		\node[circlenode] (y) at (6,0) {$Y$};
		\draw[-{Stealth[length=2mm, width=1mm]}] (s1) -- (a1);
		\draw[-{Stealth[length=2mm, width=1mm]}] (a1) -- (s2);
		\draw[-{Stealth[length=2mm, width=1mm]}] (s1) -- (s2);
		\draw[-{Stealth[length=2mm, width=1mm]}] (s2) -- (a2);
		\draw[-{Stealth[length=2mm, width=1mm]}] (s1) -- (a2);
		\draw[-{Stealth[length=2mm, width=1mm]}] (a1) -- (a2);
		\draw[-{Stealth[length=2mm, width=1mm]}] (s1) -- (y);
		\draw[-{Stealth[length=2mm, width=1mm]}] (a1) -- (y);
		\draw[-{Stealth[length=2mm, width=1mm]}] (s2) -- (y);
		\draw[-{Stealth[length=2mm, width=1mm]}] (a2) -- (y);
	\end{tikzpicture}
	\caption{\centering Causal diagrams for dynamic settings with two exposure occasions.}
	\label{fig:diag}
\end{figure}

In this work, we concentrate on estimating a causal effect known as the dynamic treatment effect (DTE), defined as $\mathrm{DTE}:=\E\{Y(a_1,a_2)-Y(a_1’,a_2’)\}$. We focus on counterfactual mean \(\theta_{1,1} := \E\{Y(1,1)\}\) because the same method extends to any \(\E\{Y(a_1,a_2)\}\) and thus to the DTE. Because \(Y_i(1,1)\) is observed only when \(A_{1i} = A_{2i} = 1\), we cannot simply average across all potential outcomes \(Y_i(1,1)\) since many remain unobserved. Additionally, due to the presence of confounding, we generally have
$$\theta_{1,1} \neq \E\{Y(1,1) \mid A_1 = A_2 = 1\}.$$
Marginal Structural Mean (MSM) models are widely used in causal inference to assess the impact of time-dependent treatments on outcomes, allowing for time-dependent covariates affected by previous treatments \citep{robins2000marginal}. MSMs are determined by a score function that identifies $\theta_{1,1}$ and a number of nuisance parameters. Our goal is to ensure correct optimal inference — even if nuisance models are misspecified and cannot be estimated at the usual \(\sqrt{N}\)-rate. We begin with the minimal set of assumptions outlined in Assumption 1. 
\begin{assumption} \label{cond:basic}
(a) Sequential ignorability: 
$
Y(a_1,a_2)\independent A_1\mid\bS_1$, $Y(a_1,a_2)\independent A_2\mid(\bS_1,\bS_2,A_1=a_1).
$
(b) Consistency:
$Y=Y(A_1,A_2).$
(c) Overlap: 
let $\P(c_0<\pi(\bS_1)<1-c_0)=1,\ \P(c_0<\rho(\bar\bS_2)<1-c_0)=1$ with some constant $c_0\in(0,1)$. Additionally, let $\pi^*$ and $\rho^*$ be some functions satisfying $\P(c_0<\pi^*(\bS_1)<1-c_0)=1,\ \P(c_0<\rho^*(\bar\bS_2)<1-c_0)=1.$
\end{assumption}

In the standard conditions above (see, e.g., \cite{robins1987addendum,robins2000marginal,murphy2003optimal}), we include two propensity score (PS) models — \(\pi(\mathbf{s}_1) := \P(A_1 = 1 \mid \mathbf{S}_1 = \mathbf{s}_1)\) and \(\rho(\bar{\mathbf{s}}_2) := \P(A_2 = 1 \mid \bar{\mathbf{S}}_2 = \bar{\mathbf{s}}_2, A_1 = 1)\) — and two conditional, counterfactual, outcome regression (OR) models — \(\mu(\mathbf{s}_1) := \E\{Y(1,1) \mid \mathbf{S}_1 = \mathbf{s}_1\}\) and \(\nu(\bar{\mathbf{s}}_2) := \E\{Y(1,1) \mid \bar{\mathbf{S}}_2 = \bar{\mathbf{s}}_2, A_1 = 1\}\). These four models are the true, yet unknown population processes that will factor into the identification of $\theta_{1,1}$. 
 We introduce “working” models \(\pi^*\), \(\rho^*\), \(\mu^*\), and \(\nu^*\), which need not coincide with the true processes but are used to guide estimation the necessary nuisances.
Given this framework, \(\theta_{1,1}\) can be identified in multiple ways using only observable variables. We focus on the doubly robust identification approach, as outlined in \cite{murphy2001marginal, bang2005doubly, yu2006double}, where
\begin{equation}\label{def:score}
\theta_{1,1}=\E\{\psi(\bW;\etabf^*)\}, \ \psi(\bW;\etabf^*):=\mu^*(\bS_1)+\frac{A_1\{\nu^*(\bar\bS_2)-\mu^*(\bS_1)\}}{\pi^*(\bS_1)}+\frac{A_1A_2\{Y-\nu^*(\bar\bS_2)\}}{\pi^*(\bS_1)\rho^*(\bar\bS_2)},
\end{equation}
as long as Assumption \ref{cond:mis} holds. 
\begin{assumption}[Sequential model double robustness]\label{cond:mis}
Let (a) either $\pi=\pi^*$ or $\mu=\mu^*$ hold, but not necessarily both; and (b) either $\rho=\rho^*$ or $\nu=\nu^*$, but not necessarily both.
\end{assumption} 

However, modeling the counterfactual mean given covariates and treatment history up to a certain time point inherently imposes constraints on the counterfactual mean given earlier histories. For instance, in the representation
$
\mu(\mathbf{s}_1) \;=\; \E\{\nu(\bar{\mathbf{S}}_2) \mid \mathbf{S}_1 = \mathbf{s}_1, A_1 = 1\}
$
from \cite{murphy2001marginal}, the correctness of \(\mu\) hinges on the correctness of \(\nu\). We show that the above double-robust representation is not sufficient to guarantee \emph{model double robustness}. A method is deemed \emph{doubly robust} if it yields consistent estimates under Assumption \ref{cond:mis}. 
In contrast, a method is considered \emph{model doubly robust} (or said to provide \emph{doubly robust inference}) if its inference remains valid under the same Assumption \ref{cond:mis}, as long as the (possibly misspecified) working models can be estimated with $o(N^{-1/4})$ rates \citep{smucler2019unifying}.
New estimates of the nuisances are required to accommodate non-\(\sqrt{N}\) convergence rate while still guaranteeing \(\sqrt{N}\)-rate inference for $\theta_{1,1}$ under the same Assumption \ref{cond:mis}. 

\subsection{Existing results}

We illustrate the existing results under the following four settings all encompassed within Assumption \ref{cond:mis}:
\vskip -10pt
\begingroup
\setlength{\abovedisplayskip}{0pt}
\setlength{\belowdisplayskip}{0pt}
\setlength{\abovedisplayshortskip}{0pt}
\setlength{\belowdisplayshortskip}{0pt}
\begin{align}
&\text{Only the OR models are correctly parametrized;}\label{CAN_a}
\end{align}
\begin{align}
&\text{Only the PS models are correctly parametrized;}\label{CAN_b}
\end{align}
\begin{align}
&\text{Only the first OR and the second PS models are correctly parametrized;}\label{CAN_c}
\end{align}
\begin{align}
&\text{Only the second OR and first the PS models are correctly parametrized.}\label{CAN_d}
\end{align}
\endgroup

\noindent In low dimensions, inverse probability weighting (IPW) \citep{robins1986new,hernan2001marginal,robins2004optimal} provide valid inference allowing only \eqref{CAN_b}. On the other hand, covariate balancing methods \citep{kallus2018optimal,yiu2018covariate} only allow for \eqref{CAN_a}. Doubly robust methods \citep{bang2005doubly, yu2006double,orellana2010dynamic} allow \eqref{CAN_a} or \eqref{CAN_b}, but do not accommodate \eqref{CAN_c} and \eqref{CAN_d}. The multiple robust estimator proposed by \cite{babino2019multiple} allows for \eqref{CAN_a}, \eqref{CAN_b}, or \eqref{CAN_c}, but does not accommodate \eqref{CAN_d}. Using the following double-robust imputation step,
\begin{equation}\label{rep:DR-mu}
\mu(\mathbf{s}_1) \;=\; \E\!\Biggl[\nu^*(\bar{\mathbf{S}}_2) \;+\; \frac{A_2 \bigl\{Y - \nu^*(\bar{\mathbf{S}}_2)\bigr\}}{\rho^*(\bar{\mathbf{S}}_2)} \;\mid\; \mathbf{S}_1 = \mathbf{s}_1, A_1 = 1\Biggr],
\end{equation}
\citet{luedtke2017sequential} propose a consistent estimator and \citet{rotnitzky2017multiply} develop an asymptotically normal estimator of the DTE under conditions \eqref{CAN_a}--\eqref{CAN_d}, requiring that the nuisance estimates belong to a Donsker class \citep{van2000asymptotic}. However, Donsker conditions place bounded complexity on the functional class, ensuring \(\sqrt{N}\)-rate convergence of nuisance estimates. These constraints are unsuitable for many nonparametric or high-dimensional models. \cite{rotnitzky2017multiply,bodory2022evaluating,diaz2023nonparametric} adopt cross-fitting techniques and double machine learning \citep{chernozhukov2017double} to relax the need for Donsker conditions, thereby allowing more flexible, nonparametric methods. However, valid inference still requires \emph{all} nuisance estimates to converge sufficiently quickly to their true underlying models, effectively ruling out model misspecification.

Similar picture persists with high-dimensional working models -- in order to guarantee \(\sqrt{N}\)-inference, all working nuisance models must be correctly specified. This is achieved in a sequence of papers across different models. For structural nested mean models \citep{Robins1997causal}, \citet{lewis2021double} achieve inferential guarantees only when the “blip functions”—differences in outcome regression across treatment paths—are low-dimensional and correctly specified. \citet{viviano2021dynamic} require both OR models to be correct, covering only \eqref{CAN_a}. Dynamic Treatment Lasso (DTL) \citep{bodory2022evaluating} achieves \emph{consistency} under \eqref{CAN_a}, \eqref{CAN_b}, and \eqref{CAN_d}, and Sequential Doubly Robust Lasso (S-DRL) \citep{bradic2024high} extends this to \eqref{CAN_c} via \eqref{rep:DR-mu}. However, both DTL and S-DRL require all nuisance models to be correctly specified for valid \emph{inference}, thus excluding misspecification in \eqref{CAN_a}--\eqref{CAN_d}. 

\subsection{Addressing misspecification: direct bias control}

All the doubly robust methods discussed above share a common limitation: while estimators based on the doubly robust score \eqref{def:score} reduce bias to product (quadratic) terms of the nuisance estimation errors, this reduction critically depends on correctly specifying \emph{all} nuisance models. Once any model is misspecified, the bias reduction fails. To address this, we propose new moment conditions that directly control bias under Assumption~\ref{cond:mis}.
To be concrete, we employ linear OR and logistic PS working models: 
\[
\nu^*(\bar\bs_2) = \bar\bs_2^\top\balpha^*,
\quad
\mu^*(\bs_1) = \bs_1^\top\bbeta^*, 
\quad 
\pi^*(\bs_1)=g(\bs_1^\top\bgamma^*), 
\quad
\rho^*(\bar\bs_2)=g(\bar\bs_2^\top\bdelta^*),
\]
where \(g(u)=\exp(u)/\{\exp(u)+1\}\). These represent the ``best'' linear or logistic models approximating the true underlying processes \citep{buja2019models}. With
\(\boldsymbol{\eta}^* := (\balpha^{*\top}, \bbeta^{*\top}, \bgamma^{*\top}, \bdelta^{*\top})^\top\),
 \(\thetahat_{1,1}=N^{-1}\sum_{i=1}^N\psi(\bW_i;\etabfhat)\) for $\etabfhat=(\balphahat^\top,\bbetahat^\top,\bgammahat^\top,\bdeltahat^\top)^\top$. Then, 
\begin{equation}\label{eq:taylor}
\thetahat_{1,1}-\theta_{1,1}=\Delta_1+\Delta_2+\Delta_3.
\end{equation}
In the above, \(\Delta_1 := N^{-1} \sum_{i=1}^N \psi(\bW_i;\etabf^*) - \theta_{1,1}\) is \(O_p(N^{-1/2})\) and asymptotically normal under Assumption~\ref{cond:mis} and standard moment conditions. \(\Delta_3\) depends quadratically on \(\etabfhat - \etabf^*\) and is typically negligible.
 There are two main strategies to control the main bias term $
\Delta_2 :=N^{-1}\sum_{i=1}^N \nabla_{\etabf} \psi(\bW_i;\etabf^*)^\top (\etabfhat - \etabf^*)$. One is to constrain the nuisance models to be either Donsker or low-dimensional, ensuring \(\sqrt{N}\)-rate convergence of $\hat \etabf -\etabf^*$ -- even if those models are misspecified -- and yielding \(\Delta_2 = o_p(N^{-1/2})\). The other is to apply cross-fitting, which induces independence between the summands in \(N^{-1}\sum_{i=1}^N \nabla_{\etabf} \psi(\bW_i;\etabf^*)\) and thereby allows a weak law of large numbers argument to deliver \(\Delta_2 = o_p(N^{-1/2})\). However, cross-fitting requires \(\E\{\nabla_{\etabf}\psi(\bW;\etabf^*)\}=\bzero\), a “Neyman orthogonality” property \citep{chernozhukov2017double} that holds only when the nuisance models are correctly specified—thus excluding the misspecification scenario. Whenever misspecification occurs, \(\E\{\nabla_{\etabf}\psi(\bW;\etabf^*)\}\neq\bzero\) and
\[
\Delta_2
=
\E\{\nabla_{\etabf}\psi(\bW;\etabf^*)\}^\top(\etabfhat - \etabf^*)
\;+\;
o_p(N^{-1/2}),
\]
meaning \(\Delta_2\) depends linearly on the estimation error \(\etabfhat - \etabf^*\) even if cross-fitting is used and \(\etabfhat\) is estimated on a separate dataset. Instead, we design new nuisance estimates, named \emph{moment-targeted estimates} that directly guarantee the following moment condition
\begin{equation}
\E\{\nabla_{\etabf}\psi(\bW;\etabf^*)\}=\bzero, \;\; \text{even under model misspecification}. \label{eq:SMDR}
\end{equation}
This reduction remains effective even when $\etabfhat - \etabf^*$ does not converge at the $\sqrt{N}$ rate.
We leverage two key components: (i) the standard doubly robust score \eqref{def:score}, and (ii) new moment-targeted nuisance estimates and new loss functions. Both components are needed. The former addresses bias when all models are correctly specified. The latter specifically targets the bias of model misspecification. As shown in Table \ref{table:rate}, this approach transforms linear bias into quadratic, achieving faster convergence rates and robust inference even under model misspecification; see Table \ref{table:sparsity}. 
Our approach applies to all cases \eqref{CAN_a}-\eqref{CAN_d}.

\section{The sequential model doubly robust estimator}\label{sec:DR-DTE}

In this section, we introduce he sequential model doubly robust (SMDR) estimator.
 For any $\bomega \in \{\bgamma, \bdelta, \balpha, \bbeta\}$, let $\Delta_{2,\bomega}:=\E\{\bnabla_{\bomega}\psi(\bW; \etabf^*)\}^\top (\widehat{\bomega} - \bomega^*)$ be the bias resulting from the nuisance estimation of $\bomega^*$. Then, we also have \(\Delta_2 \approx \sum_{\bomega \in \{\bgamma, \bdelta, \balpha, \bbeta\}} \Delta_{2,\bomega}\) with
\begin{align}
 \E\{\bnabla_{\bbeta}\psi(\bW; \etabf^*)\} &= \E\left[\left\{1 - \frac{A_1}{g(\bS_1^\top \bgamma^*)}\right\} \bS_1\right],\label{eq:moment-beta}
 \\
 \E\{\bnabla_{\balpha}\psi(\bW; \etabf^*)\} &= \E\left[\frac{A_1}{g(\bS_1^\top \bgamma^*)}\left\{1 - \frac{A_2}{g(\bar{\bS}_2^\top \bdelta^*)}\right\} \bar{\bS}_2\right],\label{eq:moment-alpha}\\
 \E\{\bnabla_{\bdelta}\psi(\bW; \etabf^*)\} &= \E\left\{-\frac{A_1 A_2 \exp(-\bar{\bS}_2^\top \bdelta^*) (Y - \bar{\bS}_2^\top \balpha^*)}{g(\bS_1^\top \bgamma^*)} \bar{\bS}_2\right\},\nonumber
 \\
 \E\{\bnabla_{\bgamma}\psi(\bW; \etabf^*)\} &= \E\left[-A_1 \exp(-\bS_1^\top \bgamma^*)\left\{\frac{A_2 (Y - \bar{\bS}_2^\top \balpha^*)}{g(\bar{\bS}_2^\top \bdelta^*)} + \bar{\bS}_2^\top \balpha^* - \bS_1^\top \bbeta^*\right\} \bS_1\right].\nonumber
\end{align}
The above equations are listed in an order such that the right-hand side of each equation involves progressively more nuisance parameters. For instance, \eqref{eq:moment-beta} only involves $\bgamma^*$, while \eqref{eq:moment-alpha} involves both $\bgamma^*$ and $\bdelta^*$.

 We first propose nuisance estimates \(\bgammahat\), \(\bdeltahat\), \(\balphahat\), and \(\bbetahat\) in a sequential manner. 
Define the index sets as \(\mathcal{I}_{\bgamma}, \mathcal{I}_{\bdelta}, \mathcal{I}_{\balpha}, \mathcal{I}_{\bbeta} \subseteq \{1, \dots, N\}\), see Algorithm \ref{alg:BRDR}. We define the estimate for the first PS, \(\pi^*(\bs_1)\), with \(\lambda_{\bgamma} > 0\), as 
\begin{align}
\bgammahat & := \arg\min_{\bgamma \in \R^{d_1}} \left\{ \lvert\mathcal{I}_{\bgamma}\rvert^{-1} \sum_{i \in \mathcal{I}_{\bgamma}} \ell_1(\bW_i; \bgamma) + \lambda_{\bgamma} \|\bgamma\|_1 \right\},\;\;\mbox{where}\label{def:alphahat}\\
\ell_1(\bW; \bgamma) & := (1 - A_1) \bS_1^\top \bgamma + A_1 \exp(-\bS_1^\top \bgamma). \label{def:alpha}
\end{align}
 
The loss function \eqref{def:alpha} is designed to achieve covariate balancing as
\[
 \E\{w_1 \bS_1\} = \E(\bS_1) , \qquad w_1:= A_1g^{-1}(\bS_1^\top\bgamma^*).
\]
Strong covariate balancing has been used for estimation of single-stage average treatment effect \citep{ning2020robust}.
Moreover, even when the PS model is misspecified, it still strictly controls the first term in the above decomposition by ensuring \(\Delta_{2,\beta}=\E\{\bnabla_{\bbeta}\psi(\bW; \etabf^*)\} = \bzero\) for \(\bgamma^* := \arg\min_{\bgamma \in \R^{d_1}} \E\{\ell_1(\bW; \bgamma)\}\). This in turn, effectively reduces the bias induced by the estimation error of \(\bbetahat\), the nuisance estimate for the first OR model. On the other hand, if the PS model is logistic, i.e., \(\pi(\bS_1) = g(\bS_1^\top\bgamma^0)\) for some \(\bgamma^0 \in \R^{d_1}\), then, \(\bgamma^* = \bgamma^0\); see Section \ref{sec:exist_unique} of the Supplementary Material.

Next, we construct the estimate for the second PS model, $\rho^*(\bar\bs_2)$, with \(\lambda_{\bdelta}>0\) and 
\begin{align}
\bdeltahat=\bdeltahat(\bgammahat)~:=~&\arg\min_{\bdelta\in\R^{d}}\biggr\{M^{-1}\sum_{i\in\mathcal I_{\bdelta}}\ell_2(\bW_i;\bgammahat,\bdelta)+\lambda_{\bdelta}\|\bdelta\|_1\biggr\},\;\;\mbox{where}\label{def:betahat}\\
\ell_2(\bW;\bgamma,\bdelta)~:=~&\frac{A_1}{g(\bS_1^\top\bgamma)}\left\{(1-A_2)\bar\bS_2^\top\bdelta+A_2\exp(-\bar\bS_2^\top\bdelta)\right\}.\label{def:l2}
\end{align}
 This loss function achieves a new kind of covariate balancing where 
 \[
 \E\{w_2 w_1\bar\bS_2\} = \E(w_1\bar\bS_2), \qquad w_2= A_2g^{-1}(\bar\bS_2^\top\bdelta^*).
 \]
 One can interpret the above as conditional covariate balancing: given the information from the first time period, we achieve classical balancing at the second time exposure. The weights \(w_1\) align with those of \(\pi^*\) but are essential for correctly targeting the bias term \(\Delta_2\). Specifically, this conditional covariate balancing ensures that
$
 \Delta_{2,\alpha} =\E\left\{\nabla_{\balpha} \psi(\bW; \etabf^*)\right\} = \nabla_{\bdelta} \E\left\{\ell_2(\bW; \bgamma^*, \bdelta^*)\right\} = \mathbf{0},
$
for
$
\bdelta^* := \arg\min_{\bdelta \in \mathbb{R}^{d}} \E\left\{\ell_2(\bW; \bgamma^*, \bdelta)\right\},
$
even when the models are misspecified.
This in turn, effectively reduces the bias induced by the estimation error of \(\balphahat\), the nuisance estimate for the second OR model. 

For the remaining two OR models, we define moment-targeted nuisance estimators with properly chosen \(\lambda_{\balpha},\lambda_{\bbeta}>0\) as
\begin{align}
\balphahat=\balphahat(\bgammahat,\bdeltahat)~:=~&\arg\min_{\balpha\in\R^d}\left\{M^{-1}\sum_{i\in\mathcal I_{\balpha}}\ell_3(\bW_i;\bgammahat,\bdeltahat,\balpha)+\lambda_{\balpha}\|\balpha\|_1\right\},\label{def:gammahat}
\\
\bbetahat=\bbetahat(\bgammahat,\bdeltahat,\balphahat)~:=~&\arg\min_{\bbeta\in\R^{d_1}}\left\{M^{-1}\sum_{i\in\mathcal I_{\bbeta}}\ell_4(\bW_i;\bgammahat,\bdeltahat,\balphahat,\bbeta)+\lambda_{\bbeta}\|\bbeta\|_1\right\}.\label{def:deltahat}
\end{align}
The corresponding loss functions are defined as:
\begin{align}
\ell_3(\bW;\bgamma,\bdelta,\balpha)&~:=~\frac{A_1A_2\exp(-\bar\bS_2^\top\bdelta)}{g(\bS_1^\top\bgamma)} \left(Y-\bar\bS_2^\top\balpha \right)^2,\label{def:l3}
\end{align}
\vspace{-2em}
\begin{align}
\ell_4(\bW;\bgamma,\bdelta,\balpha,\bbeta)&~:=~A_1\exp(-\bS_1^\top\bgamma)\left\{\bar\bS_2^\top\balpha+\frac{A_2(Y-\bar\bS_2^\top\balpha)}{g(\bar\bS_2^\top\bdelta)}-\bS_1^\top\bbeta\right\}^2.\label{def:l4}
\end{align}

The above \eqref{def:l3}-\eqref{def:l4} mitigate the estimation bias by ensuring 
\[
\E\left\{\bnabla_{\bgamma}\psi(\bW; \etabf^*)\right\} = \bnabla_{\bbeta}\E\{\ell_4(\bW; \bgamma^*, \bdelta^*, \balpha^*, \bbeta^*)\}/2 = \bzero\]
\[ \E\left\{\bnabla_{\bdelta}\psi(\bW; \etabf^*)\right\} = \bnabla_{\balpha}\E\{\ell_3(\bW; \bgamma^*, \bdelta^*, \balpha^*)\}/2 = \bzero, \]
leading to \(\Delta_{2,\bgamma} = \Delta_{2,\bdelta} = 0\) for population slopes $\balpha^*:=\arg\min_{\balpha\in\R^d}\E\{\ell_3(\bW;\bgamma^*,\bdelta^*,\balpha)\}$ and $\bbeta^*:= \arg\min_{\bbeta\in\R^{d_1}}\E\{\ell_4(\bW;\bgamma^*,\bdelta^*,\balpha^*,\bbeta)\}$, respectively.
 Uniqueness of \(\bgamma^*\), \(\bdelta^*\), \(\balpha^*\), and \(\bbeta^*\) are discussed in Section \ref{sec:exist_unique} of the Supplementary Material.
 The introduced loss functions are performing (imputed) residual covariate balancing of the outcome regressions as 
 \[
 \E\left\{w_1 w_2' ( Y-\bar\bS_2^\top\balpha) \right\}=0, \qquad \E\left\{w_1'( Y^{\mbox{\tiny DR}}-\bS_1^\top\bbeta ) \right\} =0
 \]
 where $w_1'=\bnabla_{\bgamma} w_1$ and $w_2'=\bnabla_{\bdelta} w_2$ and $Y^{\mbox{\tiny DR}}=\bar\bS_2^\top\balpha+ {A_2(Y-\bar\bS_2^\top\balpha)}/{g(\bar\bS_2^\top\bdelta)}$ is the double robust imputation based of \eqref{rep:DR-mu}. Here, we are ensuring that the (imputed) residuals are uncorrelated with the adjustments made to the second and first PS estimations.

 The loss functions \eqref{def:alpha}, \eqref{def:l2}, \eqref{def:l3}, and \eqref{def:l4} are termed \emph{moment-targeting loss functions}. The \emph{sequential model doubly robust} (SMDR) estimator of \(\theta_{1,1}\) utilizing a cross-fitting technique can be found in Algorithm \ref{alg:BRDR}.

\begin{algorithm}[ht]
\caption{The sequential model doubly robust (SMDR) estimator of $\theta_{1,1}$}\label{alg:BRDR}
\begin{algorithmic}[1]
\Require Observations $\S=(\bW_i)_{i=1}^N $ and the treatment path $(a_1,a_2)=(1,1)$.
\State Let $\mathcal I = \{1,2,\dots,N\}=\cup_{k=1}^{\K}\mathcal I_k$ with equal sized splits $n=N/\K$ and $\K\geq2$.
\For{$k=1,2,...,\K$}
\State $\mathcal I_{-k}\leftarrow\mathcal I\setminus\mathcal I_k$ 
\State $\mathcal I_{\bgamma},\mathcal I_{\bdelta}, \mathcal I_{\balpha}, \mathcal I_{\bbeta}$ $\leftarrow$ size $M$ disjoint partition of $\mathcal I_{-k}$ with $M=N(\K-1)/(4\K)$. 
\State Propensity estimate at the first exposure $\bgammahat_{-k}\leftarrow\bgammahat$ as in \eqref{def:alphahat}.
\State Propensity estimate at the last/second exposure $\bdeltahat_{-k} \leftarrow \bdeltahat$ as in \eqref{def:betahat}.
\State Outcome estimate at the last exposure $\balphahat_{-k} \leftarrow\balphahat$ as in \eqref{def:gammahat}.
\State Outcome estimate at the first exposure $\bbetahat_{-k} \leftarrow \bbetahat$ as in \eqref{def:deltahat}.
\EndFor\\
\Return The SMDR estimator is
\begin{equation}
\thetahat_{1,1}=N^{-1}\sum_{k=1}^{\K}\sum_{i\in\mathcal I_k}\psi(\bW_i;\etabfhat_{-k}),\;\;\mbox{where}\;\;\etabfhat_{-k} = (\bgammahat_{-k}^\top,\bdeltahat_{-k}^\top,\balphahat_{-k}^\top,\bbetahat_{-k}^\top)^\top\label{def:thetahat} 
\end{equation}
and $\psi(\bW_i;\etabfhat_{-k})$ is defined through \eqref{def:score} replacing $\etabf^*$ with $\etabfhat_{-k}$.
\end{algorithmic}
\end{algorithm}

Off-the-shelf methods cannot achieve model double robustness, even with doubly robust or Neyman orthogonal scores, for estimating average treatment effects (ATE) with a single time exposure \citep{smucler2019unifying, tan2020model, dukes2020doubly, avagyan2021high, dukes2021inference, bradic2019sparsity}. Our introduced loss functions reduce to \(\ell_2 \) and \(\ell_4 \) in the single time exposure setting, but our dynamic problem introduces more complex challenges compared to the static case. We believe our work is the first to achieve model double robustness fully under Assumption \ref{cond:mis}.

In dynamic settings, \citep{luedtke2017sequential, rotnitzky2017multiply, bradic2024high, diaz2023nonparametric} employ doubly robust imputation, \( Y^{\mbox{\tiny DR}} \), to improve the rate conditions of the final estimator. However, beyond this step, they rely solely on off-the-shelf methods, such as (regularized) maximum likelihood estimation. We show that this approach alone is insufficient for robustness against model misspecification and that our newly introduced loss functions are essential. While the classical logistic loss 
\[
 A_1\{(1-A_2)\bar\bS_2^\top\bdelta + A_2 h(\bar\bS_2^\top\bdelta)\}
\] 
with \( h(u) = -\log g(u) \) 
is commonly used, we propose a covariate (conditional) balancing-inspired modification, replacing \( h(\bar\bS_2^\top\bdelta) \) with \( \exp(-\bar\bS_2^\top\bdelta) \) and introducing an additional weight \( w_2 \). 
Furthermore, while classical least squares loss is typically applied in OR estimation, we identify the weights \( w_1 w_2' \) and \( w_1' \) as necessary to ensure score orthogonality under PS model misspecification. Our numerical experiments confirm that, in finite samples, our approach consistently achieves better bias control than existing off-the-shelf methods (see Tables \ref{table:settinga1}-\ref{table:lin}).

\begin{remark}[Correctness of nuisance models]\label{remark:correctness}
We now discuss the correct specification of nuisance models. The first outcome regression \(\mu\), is inherently challenging to interpret in dynamic settings \citep{babino2019multiple}. This complexity arises from the dynamic nature of the problem, not from our representation.
 Below, we introduce the specific meaning of a ``correctly specified model'':
\begin{enumerate}
\item[(a)] We say $\pi^*$ is correctly specified when $\pi^*=\pi$, which occurs if and only if (iff) there exists some $\bgamma^0\in\R^{d_1}$, such that $\pi(\bs_1)=g(\bs_1^\top\bgamma^0)$ holds. Additionally, $\bgamma^*=\bgamma^0$.
\item[(b)]
We say $\rho^*$ is correctly specified when $\rho^*=\rho$, which occurs iff there exists some $\bdelta^0\in\R^d$, such that $\rho(\bar\bs_2)=g(\bar\bs_2^\top\bdelta^0)$ holds. Additionally, $\bdelta^*=\bdelta^0$.
\item[(c)]
We say $\nu^*$ is correctly specified when $\nu^*=\nu$, which occurs iff there exists some $\balpha^0\in\R^d$, such that $\nu(\bar\bs_2)=\bar\bs_2^\top\balpha^0$ holds. Additionally, $\balpha^*=\balpha^0$.
\item[(d)]
We say $\mu^*$ is correctly specified when $\mu^*=\mu$, which occurs if there exists some $\bbeta^0\in\R^d$, such that $\mu(\bs_1)=\bs_1^\top\bbeta^0$ and, furthermore, either case (b) or (c) holds. Additionally, $\bbeta^*=\bbeta^0$.
\end{enumerate}

Note that $\bdelta^*$ depends on $\bgamma^*$. However, condition (b) ensures that the correctness of $\rho^*$ is independent of $\bgamma^*$. Similarly, conditions (a)-(c) establish that the correctness of $\pi^*$, $\rho^*$, and $\nu^*$ are mutually independent.
 However, this is not the case for the OR model at the first exposure, $\mu^*$. Specifically, if $\mu$ is linear with $\mu(\bs_1) = \bs_1^\top\bbeta^0$ for some $\bbeta^0$, this does not imply that $\mu^*$ is correctly specified, as $\bbeta^*$ may differ from $\bbeta^0$. In particular, $\mu^*$ is correctly specified if, additionally, either $\rho^*$ or $\nu^*$ is (or both are) correctly specified. That is, under Assumption \ref{cond:mis}, $\mu^*$ is correctly specified if and only if $\mu$ is indeed a linear function -- no further constraints are needed. In contrast, the correctness of a linear $\mu^*$ defined through nested approaches \citep{murphy2001marginal, bodory2022evaluating} typically relies additionally on the linearity of $\nu$; see Section \ref{sec:just} of the Supplementary Materials.\end{remark}

\section{Sequential model doubly robust estimation and inference}\label{sec:DTE}
 
In the following, we choose tuning parameters $\lambda_{\bgamma}\asymp\sqrt{\log d_1/N}$, $\lambda_{\bdelta}\asymp\sqrt{\log d/N}$, $\lambda_{\balpha}\asymp\sqrt{\log d/N}$, $\lambda_{\bbeta}\asymp\sqrt{\log d_1/N}$. 
Define $s_{\bgamma}:=\|\bgamma^*\|_0$, $s_{\bdelta}:=\|\bdelta^*\|_0$, $s_{\balpha}:=\|\balpha^*\|_0$, and $s_{\bbeta}:=\|\bbeta^*\|_0$ as the sparsity levels of the population nuisance parameters. 
\begin{assumption}[Sparsity]\label{cond:sparse}
Let $s_{\bgamma}+s_{\bbeta}=o(N/\log d_1)$, $s_{\bdelta}+s_{\balpha}=o(N/\log d)$, and $(s_{\bgamma}+s_{\balpha})\log d_1\log d+s_{\bdelta}\log^2d=O(N)$.
\end{assumption}

The sparsity conditions of the form $s=o\left(N/\log d\right)$ are very common in the high-dimensional statistics literature and guarantee estimation consistency. The additional condition $(s_{\bgamma}+s_{\balpha})\log d_1\log d+s_{\bdelta}\log^2d=O(N)$ is necessary since, in general, the imputed outcomes considered in the Lasso problem \eqref{def:deltahat} do not have a bounded $\psi_\alpha$-Orlicz norm. However, this condition is no longer required if we further assume that $\|\bar\bS_2\|_\infty<C$, as in, e.g., \cite{bradic2019sparsity,tan2020model,smucler2019unifying}.

The following assumption imposes some standard moment conditions where $\|X\|_{\psi_2}:=\inf\{c>0:\E[\psi_{2}(\lvert X\rvert/c)]\leq 1\}$, with $\psi_2(x)=\exp(x^2)-1$. 

\begin{assumption}[Sub-Gaussianity]\label{cond:subG}
Let $\bar\bS_2$ be a sub-Gaussian random vector with $\|\bv^\top\bar\bS_2\|_{\psi_2}\leq\sigma_{\bS}\|\bv\|_2$ for all $\bv\in\R^d$. Let $\varepsilon:=Y(1,1)-\bar\bS_2^\top\balpha^*$ and $\zeta:=\bar\bS_2^\top\balpha^*-\bS_1^\top\bbeta^*$ be sub-Gaussian with $\|\varepsilon\|_{\psi_2}\leq\sigma_\varepsilon$ and $\|\zeta\|_{\psi_2}\leq\sigma_\zeta$. 
In addition, let $\E[A_1A_2\{Y(1,1)-\nu(\bar\bS_2)\}^2]>c_Y$ and the smallest eigenvalue of $\Cov(A_1\bar\bS_2)$ is bounded below by $c_{\min}$. Here, $\sigma_{\bS},\sigma_\varepsilon,\sigma_\zeta,c_Y,c_{\min}$ are some positive constants. 
\end{assumption}

\begin{theorem}[Convergence rates]\label{thm:rate}
Let Assumptions \ref{cond:basic}-\ref{cond:subG} hold. Define
\begin{equation}
r_{\bgamma}:=\sqrt\frac{s_{\bgamma}\log d_1}{N},\;\;r_{\bdelta}:=\sqrt\frac{s_{\bdelta}\log d}{N},\;\;r_{\balpha}:=\sqrt\frac{s_{\balpha}\log d}{N},\;\;r_{\bbeta}:=\sqrt\frac{s_{\bbeta}\log d_1}{N}.\label{def:rs}
\end{equation}
Then, as $N,d_1,d_2\to\infty$, $\sigma^2:=\E\{\psi(\bW;\etabf^*)-\theta_{1,1}\}^2\asymp\|\bbeta^*\|_2+1$ and
\begin{align}
\thetahat_{1,1}-\theta_{1,1}&=O_p(\sigma N^{-1/2}+r_{\bgamma}r_{\bbeta}+r_{\bdelta}r_{\balpha})+\idf_{\rho\neq\rho^*}O_p(r_{\bgamma}r_{\balpha})\nonumber\\
&\qquad+\idf_{\nu\neq\nu^*}O_p(r_{\bgamma}r_{\bdelta}+r_{\bdelta}^2)+\idf_{\mu\neq\mu^*}O_p(r_{\bgamma}^2+r_{\bgamma}r_{\bdelta}+r_{\bgamma}r_{\balpha}).\label{rate:thetahat}
\end{align}
\end{theorem}

Theorem~\ref{thm:rate} characterizes the convergence rate of the SMDR estimator. When all nuisance models are correctly specified, we have
\[
\thetahat_{1,1} - \theta_{1,1} = O_p\left(\sigma N^{-1/2} + r_{\bgamma}r_{\bbeta} + r_{\bdelta}r_{\balpha}\right),
\]
which matches the S-DRL estimator \citep{bradic2024high} and outperforms the DTL estimator \citep{bradic2024high, bodory2022evaluating}. Under model misspecification, DTL and S-DRL exhibit convergence rates with additional linear terms (see Table~\ref{table:rate}), whereas our result in \eqref{rate:thetahat} involves only quadratic terms—products of nuisance estimation errors. 
When a nuisance model is misspecified, the convergence rate of S-DRL (and DTL) includes an additional term that is linearly dependent on the estimation error of the other nuisance model at the same exposure. In contrast, the proposed SMDR method mitigates such model misspecification errors by introducing a multiplicity factor that incorporates estimation errors from nuisance estimates constructed prior to the misspecified model. For instance, when \(\mu^*\) is misspecified, the convergence rates of DTL and S-DRL both involve \(r_{\bgamma}\), the nuisance estimation rate for \(\nu^*\). In comparison, the SMDR method reduces this term to a product of \(r_{\bgamma}\) and \(r_{\bgamma} + r_{\bdelta} + r_{\balpha}\). The sequentially designed loss functions in SMDR effectively leverage the structures of previously estimated models to downstream the impact of model misspecification when estimating subsequent models, thereby enhancing overall robustness and accuracy in the final DTE estimation.

\begin{table}[h]
\caption{Convergence rates of DTL, S-DRL, and SMDR estimators under model misspecification situations in high dimensions. The sequences $r_{\bgamma}, r_{\bdelta}, r_{\balpha}, r_{\bbeta}$ are defined in \eqref{def:rs}. For simplicity, we consider $\sigma \asymp 1$, and denote $r_0^2 := N^{-1/2} + r_{\bgamma}r_{\bbeta} + r_{\bdelta}r_{\balpha}$. The quantities in \textcolor{red}{red} denote the additional linear terms in the convergence rates of DTL and S-DRL. The quantities in {\green green} denote quadratic terms that decay faster than the corresponding {\red red} terms on the same row. The quantities in {\orange orange} denote the additional terms that appear in DTL's convergence rate only.} \label{table:rate}
\begin{center}
\begin{tabular}{| c | c | c | c | c | c | c |}
\hline
\multicolumn{4}{| c |}{Model correctness}&\multicolumn{3}{ c |}{Convergence raets}\\
\hline
$\pi^*$&$\rho^*$&$\nu^*$&$\mu^*$&DTL&S-DRL&SMDR\\
\hline
\cmark&\cmark&\cmark&\cmark&$r_0^2+{\orange r_{\bgamma}r_{\balpha}}$&$r_0^2$&$r_0^2$\\
\hline
\cmark&\cmark&\cmark&\xmark&$r_0^2+{\red r_{\bgamma}}$&$r_0^2+{\red r_{\bgamma}}$&$r_0^2+{\green r_{\bgamma}(r_{\bgamma}+r_{\bdelta}+r_{\balpha})}$\\
\hline
\cmark&\cmark&\xmark&\cmark&$r_0^2+{\orange r_{\bgamma}r_{\balpha}}+{\red r_{\bdelta}}$&$r_0^2+{\red r_{\bdelta}}$&$r_0^2+{\green r_{\bdelta}(r_{\bgamma}+r_{\bdelta})}$\\
\hline
\cmark&\xmark&\cmark&\cmark&$r_0^2+{\red r_{\balpha}}$&$r_0^2+{\red r_{\balpha}}$&$r_0^2+{\green r_{\balpha}r_{\bgamma}}$\\
\hline
\xmark&\cmark&\cmark&\cmark&$r_0^2+{\orange r_{\balpha}}+{\red r_{\bbeta}}$&$r_0^2+{\red r_{\bbeta}}$&$r_0^2$\\
\hline
\cmark&\cmark&\xmark&\xmark&$r_0^2+{\red r_{\bgamma}+r_{\bdelta}}$&$r_0^2+{\red r_{\bgamma}+r_{\bdelta}}$&$r_0^2+{\green r_{\bgamma}(r_{\bgamma}+r_{\bdelta}+r_{\balpha})+r_{\bdelta}^2}$\\
\hline
\cmark&\xmark&\cmark&\xmark&$r_0^2+{\red r_{\balpha}+r_{\bgamma}}$&$r_0^2+{\red r_{\balpha}+r_{\bgamma}}$&$r_0^2+{\green r_{\bgamma}(r_{\bgamma}+r_{\bdelta}+r_{\balpha})}$\\
\hline
\xmark&\cmark&\xmark&\cmark&$r_0^2+{\orange r_{\balpha}}+{\red r_{\bbeta}+r_{\bdelta}}$&$r_0^2+{\red r_{\bbeta}+r_{\bdelta}}$&$r_0^2+{\green r_{\bdelta}(r_{\bgamma}+r_{\bdelta})}$\\
\hline
\xmark&\xmark&\cmark&\cmark&$r_0^2+{\red r_{\balpha}+r_{\bbeta}}$&$r_0^2+{\red r_{\balpha}+r_{\bbeta}}$&$r_0^2+{\green r_{\balpha}r_{\bgamma}}$\\
\hline
\end{tabular}
\end{center}
\end{table}

\begin{theorem}[Inference under model misspecification]\label{thm:main}
Let Assumptions \ref{cond:basic}-\ref{cond:subG} hold. Let the following product sparsity conditions hold
\begin{equation}\label{cond:s2}
r_{\bgamma}r_{\bbeta}=o(N^{-1/2})\quad\mbox{and}\quad r_{\bdelta}r_{\balpha}=o(N^{-1/2}),
\end{equation}
where sequences $r_{\bgamma}$, $r_{\bdelta}$, $r_{\balpha}$, and $r_{\bbeta}$ are defined in \eqref{def:rs}.
We assume the following additional conditions if model misspecification occurs:
\begin{align}
&\text{if}\;\;\rho\neq\rho^*,\;\;\text{further let}\;\;r_{\bgamma}r_{\balpha}=o(N^{-1/2});\label{cond:s3}\\
&\text{if}\;\;\nu\neq\nu^*,\;\;\text{further let}\;\;r_{\bgamma}r_{\bdelta}=o(N^{-1/2}),\;\;r_{\bdelta}=o(N^{-1/4});\label{cond:s4}\\
&\text{if}\;\;\mu\neq\mu^*,\;\;\text{further let}\;\;r_{\bgamma}=o(N^{-1/4}),\;\;r_{\bgamma}r_{\bdelta}+r_{\bgamma}r_{\balpha}=o(N^{-1/2}).\label{cond:s5}
\end{align}
Then, as $N,d_1,d_2\to\infty$, 
$$\sigma^{-1}N^{1/2}(\thetahat_{1,1}-\theta_{1,1})\to\mathcal N(0,1)$$ in distribution and $\sigmahat^2=\sigma^2\{1+o_p(1)\}$, where $\sigmahat^2:=N^{-1}\sum_{k=1}^{\K}\sum_{i\in\mathcal I_k}\{\psi(\bW_i;\etabfhat_{-k})-\thetahat_{1,1}\}^2$.
\end{theorem}

\begin{remark}[Sequential model double robustness]\label{remark:MDR}
In Theorem \ref{thm:main}, we establish the ``sequential model double robustness'' (SMDR) property of our proposed estimator, ensuring root-\(N\) inference as long as at least one nuisance model is correctly specified at each exposure (see Assumption \ref{cond:mis}), and some of the working models are estimated with \(o(N^{-1/4})\) rates. The specific nuisance estimation conditions (or equivalently, sparsity conditions) will be detailed in Remark \ref{remark:sparsity} and Table \ref{table:sparsity}.

The SMDR property differs from the sequential double robustness (SDR) established by \cite{luedtke2017sequential}, which guarantees only \emph{consistency} of the DTE estimates under model misspecification. Root-\(N\) inference under model misspecification has been achieved in low-dimensional settings \citep{rotnitzky2017multiply}, where nuisance estimates are assumed to satisfy Donsker conditions and exhibit parametric convergence rates, i.e., \(r_{\balpha} \asymp r_{\bbeta} \asymp r_{\bgamma} \asymp r_{\bdelta} \asymp N^{-1/2}\). Due to the fast convergence rates they require, their method also fails to achieve the SMDR property that we aim to establish. In fact, provided these parametric convergence rates, the required conditions \eqref{cond:s2}-\eqref{cond:s5} above are automatically satisfied.

Our sequentially designed loss functions expand the complexity of the functional class to which the nuisance functions belong, requiring only some models to be estimated with a rate of \(o(N^{-1/4})\), while others can be estimated with even slower rates. For instance, when $\rho^*$ is misspecified, we allow for scenarios where \(r_{\balpha} \asymp r_{\bbeta} \asymp N^{-1/3}\) and \(r_{\bgamma} \asymp r_{\bdelta} \asymp N^{-1.01/6}\). This is equivalent to \(s_{\balpha} \log d \asymp s_{\bbeta} \log d_1 \asymp N^{1/3}\) and \(s_{\bgamma} \log d \asymp s_{\bdelta} \log d_1 \asymp N^{1.99/3}\), accommodating growing sparsity levels and dimensions that violate Donsker conditions. The proposed method extends root-\(N\) inference from low-dimensional to more challenging high-dimensional settings and achieves the SMDR property.

In high dimensions, the existing DTL and S-DRL methods \citep{bodory2022evaluating, bradic2024high} only ensure consistent DTE estimates without inferential guarantees, as long as at least one nuisance model is misspecified. Only \cite{viviano2021dynamic} provided valid inference without relying on specific parametric forms for the PS functions; however, they always require all OR models to be correctly specified, i.e., \eqref{CAN_a} holds. Our proposed strategy accommodates all cases \eqref{CAN_a}-\eqref{CAN_d}, achieving the best model robustness in high dimensions.
\end{remark}

\begin{remark}[Required sparsity conditions under model misspecification]\label{remark:sparsity}
We now discuss the sparsity conditions required for root-$N$ inference in Theorem \ref{thm:main}. 

First, we examine a simpler static scenario and compare the sparsity conditions with existing literature on settings with a single exposure. These settings can be viewed as a special case of our framework, where the exposure $A_1$ is independent of $\bar{\bS}_2$, $A_2$, and $Y$. In such cases, the DTE reduces to the average treatment effect (ATE), and robust inference under model misspecification has been established by \cite{smucler2019unifying}, \cite{tan2020model}, \cite{avagyan2021high}, \cite{ning2020robust}, among others. For these setups, the required sparsity conditions are $r_{\bdelta} r_{\balpha} = o(N^{-1/2})$, and if the outcome regression (OR) model $\nu^*$ is misspecified, we additionally require $r_{\bdelta} = o(N^{-1/4})$, as given in \eqref{cond:s2} and \eqref{cond:s4}. These conditions match those in \cite{smucler2019unifying}, but are weaker than the conditions in \cite{tan2020model}, \cite{avagyan2021high}, and \cite{ning2020robust}, where a stronger condition $r_{\bdelta} + r_{\balpha} = o(N^{-1/4})$ is required.

Next, we consider the more complex dynamic scenarios. Similar to the static case, the correctness of one PS model, $\pi^*$, does not affect the required sparsity conditions. However, the correctness of the other PS model, $\rho^*$, along with both OR models, $\nu^*$ and $\mu^*$, influences the sparsity requirements. The more models that are misspecified, the more stringent the sparsity conditions become. When $\rho^*$, $\nu^*$, and $\mu^*$ are all correctly specified, we require Assumption \ref{cond:sparse} and \eqref{cond:s2}. If any model at exposure $t \in \{1, 2\}$ is misspecified, we impose a product condition between (i) the sparsity level of the other (correctly specified) model at the same exposure $t$ and (ii) the summation of sparsity levels corresponds to all the nuisance estimators that such a misspecified estimator is constructed based on. Recall that we estimate the nuisance models sequentially in the order: $\bgammahat$, then $\bdeltahat$, followed by $\balphahat$ and $\bbetahat$. For instance, when the OR model $\mu^*$ is misspecified, as shown in \eqref{cond:s5}, we require a product condition between (i) $s_{\bgamma}$ and (ii) $s_{\bgamma} + s_{\bdelta} + s_{\balpha}$. Moreover, if the OR model at the $t$-th exposure is misspecified, an ultra-sparse PS parameter is required at that exposure, as the OR models are estimated based on the PS estimates.

Our results provide a clear framework for understanding the additional challenges posed by model misspecification. Since achieving the SMDR property requires sequential estimation of the PS models followed by backward estimation of the OR models, misspecification of early-stage OR models imposes the most stringent conditions on the model structures. In contrast, misspecification of the first-stage PS model does not impact the sparsity conditions. The sparsity conditions required in \cite{smucler2019unifying} can be viewed as a special case of the more general phenomenon we identify for single-exposure settings. Further details are provided in Table \ref{table:sparsity}.
\end{remark}

Whenever all nuisance models are correctly specified, we have the following result.
\begin{theorem}[Inference under correctly specified models]\label{cor:correct}
Suppose all the nuisance models are correctly specified. Let Assumptions \ref{cond:basic}, \ref{cond:sparse}, and \ref{cond:subG} hold, as well as the product sparsity \eqref{cond:s2}. Then, as $N,d_1,d_2\to\infty$, 
$$\sigma^{-1}N^{1/2}(\thetahat_{1,1}-\theta_{1,1})~\to~\mathcal N(0,1)$$ in distribution and $\sigmahat^2=\sigma^2\{1+o_p(1)\}$.
\end{theorem}

When all nuisance functions are correctly specified, the result coincides with that of \cite{bradic2024high} while also achieving the semi-parametric efficiency of \cite{bang2005doubly}. Hence, we do not lose accuracy when the nuisance models are correctly specified.
As shown in Theorem \ref{cor:correct}, root-$N$ inference requires product sparsity conditions between the nuisance parameters' sparsity levels at each exposure, i.e., \eqref{cond:s2}; we name such a property as ``sequential rate double robustness''. This condition is weaker than the DTL estimator where an additional product sparsity condition $s_{\bgamma}s_{\balpha}=o(N/(\log d_1\log d))$ is imposed.

\begin{table}[h]
\caption{Sparsity conditions required for the SMDR estimator asymptotically normal. For simplicity, $\|\bar\bS_2\|_\infty<C$, $d_1\asymp d$, and $s_{\bgamma}+s_{\bdelta}+s_{\balpha}+s_{\bbeta}=o\left(N/\log d\right)$. } \label{table:sparsity}
\begin{center}
\begin{tabular}{| c | c | c | c | c |}
\hline
\multicolumn{4}{| c |}{Model correctness}&\multirow{2}{*}{Required sparsity conditions}\\
\cline{1-4}
$\pi^*$&$\rho^*$&$\nu^*$&$\mu^*$&\\
\hline
\cmark&\cmark&\cmark&\cmark&$s_{\bgamma}s_{\bbeta}+s_{\bdelta}s_{\balpha}=o\left(\frac{N}{\log^2d}\right)$\\
\hline
\cmark&\cmark&\cmark&\xmark&$s_{\bgamma}=o\left(\frac{\sqrt N}{\log d}\right),\;s_{\bgamma}s_{\bdelta}+s_{\bgamma}s_{\balpha}+s_{\bgamma}s_{\bbeta}+s_{\bdelta}s_{\balpha}=o\left(\frac{N}{\log^2d}\right)$\\
\hline
\cmark&\cmark&\xmark&\cmark&$s_{\bdelta}=o\left(\frac{\sqrt N}{\log d}\right),\;s_{\bgamma}s_{\bdelta}+s_{\bgamma}s_{\bbeta}+s_{\bdelta}s_{\balpha}=o\left(\frac{N}{\log^2d}\right)$\\
\hline
\cmark&\xmark&\cmark&\cmark&$s_{\bgamma}s_{\balpha}+s_{\bgamma}s_{\bbeta}+s_{\bdelta}s_{\balpha}=o\left(\frac{N}{\log^2d}\right)$\\
\hline
\xmark&\cmark&\cmark&\cmark&$s_{\bgamma}s_{\bbeta}+s_{\bdelta}s_{\balpha}=o\left(\frac{N}{\log^2d}\right)$\\
\hline
\cmark&\cmark&\xmark&\xmark&$s_{\bgamma}+s_{\bdelta}=o\left(\frac{\sqrt N}{\log d}\right),\;s_{\bgamma}s_{\balpha}+s_{\bgamma}s_{\bbeta}+s_{\bdelta}s_{\balpha}=o\left(\frac{N}{\log^2d}\right)$\\
\hline
\cmark&\xmark&\cmark&\xmark&$s_{\bgamma}=o\left(\frac{\sqrt N}{\log d}\right),\;s_{\bgamma}s_{\bdelta}+s_{\bgamma}s_{\balpha}+s_{\bgamma}s_{\bbeta}+s_{\bdelta}s_{\balpha}=o\left(\frac{N}{\log^2d}\right)$\\
\hline
\xmark&\cmark&\xmark&\cmark&$s_{\bdelta}=o\left(\frac{\sqrt N}{\log d}\right),\;s_{\bgamma}s_{\bdelta}+s_{\bgamma}s_{\bbeta}+s_{\bdelta}s_{\balpha}=o\left(\frac{N}{\log^2d}\right)$\\
\hline
\xmark&\xmark&\cmark&\cmark&$s_{\bgamma}s_{\balpha}+s_{\bgamma}s_{\bbeta}+s_{\bdelta}s_{\balpha}=o\left(\frac{N}{\log^2d}\right)$\\
\hline
\end{tabular}
\end{center}
\end{table}

\section{Theoretical results for the nuisance estimators}\label{sec:nuis}

In the following, we develop theoretical properties of the proposed moment-targeted nuisance estimators $\bgammahat$, $\bdeltahat$, $\balphahat$, and $\bbetahat$, defined in \eqref{def:alphahat}-\eqref{def:deltahat}. The analysis of nuisance estimation is non-trivial since the nuisance estimates are constructed in a sequential manner. Section \ref{sec:asymp_nuisance} shows the nuisance estimators' consistency despite potential model misspecification, while Section \ref{sec:asymp_nuisance'} presents their faster consistency rates when certain models are correctly specified. Our findings show that the accuracy of nuisance models affects the estimation errors.

\subsection{Results for misspecified models}\label{sec:asymp_nuisance}

Our first focus is on the asymptotic behavior of moment-targeted nuisance estimators with possibly inaccurate models. 
 In determining convergence rates, we confront the complexities of RSC conditions caused by dependent loss functions in Lemma \ref{lemma:RSC}, and manage the increased gradient variability in Lemma \ref{lemma:gradient}, as expanded upon in the Supplementary Material.

\begin{theorem}\label{thm:nuisance}
Let Assumptions \ref{cond:basic} and \ref{cond:subG} hold. Define sequences $r_{\bgamma}$, $r_{\bdelta}$, $r_{\balpha}$, and $r_{\bbeta}$ as in \eqref{def:rs}. Then, as $N,d_1,d_2\to\infty$, the following holds:
\begin{enumerate}
\item[(a)] If $r_{\bgamma}=o(1)$, then
$\|\bgammahat-\bgamma^*\|_2=O_p(r_{\bgamma}).$
\item[(b)] In addition to (a), if $r_{\bdelta}=o(1)$, then
$\|\bdeltahat-\bdelta^*\|_2= O_p(r_{\bgamma}+r_{\bdelta}).$
\item[(c)] In addition to (a) and (b), if $r_{\balpha}=o(1)$, then
$\|\balphahat-\balpha^*\|_2= O_p(r_{\bgamma}+r_{\bdelta}+r_{\balpha}).$
\item[(d)] In addition to (a), (b), and (c), if $r_{\bbeta}=o(1)$, then $
\|\bbetahat-\bbeta^*\|_2= O_p(r_{\bgamma}+r_{\bdelta}+r_{\balpha}+r_{\bbeta}).$
\end{enumerate}
\end{theorem}

 Among the results in Theorem \ref{thm:nuisance}, part (b) is the most challenging to show. Notice that $\bdeltahat$ is constructed based on a first-stage estimate $\bgammahat$. Due to the occurrence of the imputation error $\bgammahat-\bgamma^*$, the estimation error $\bdeltahat-\bdelta^*$ no longer belongs to the usual cone set $\C(S,k):=\{\bDelta\in\R^d:\|\bDelta_{S^c}\|_1\leq k\|\bDelta_S\|_1\}$. A similar problem has been recently studied by \cite{bradic2024high}, where their Theorem 8 provides consistency rates of imputed Lasso estimates. The problem we consider here is even more technically challenging in that the loss function \eqref{def:l2} is non-quadratic with respect to $\bdelta$. We consider a cone set $\Ctil(s,k):=\{\bDelta\in\R^d:\|\bDelta\|_1\leq k\sqrt{s}\|\bDelta\|_2\}$ that is ``larger'' than the usual $\C(S,k)$ and also different from the cone set studied by \cite{bradic2024high}. We show that $\bdeltahat-\bdelta^*\in\Ctil(s,k)$ with high probability and some $k,s>0$; see details in Lemma \ref{lemma:beta2}. Together with some empirical process results as in Lemma \ref{lemma:beta1}, we control the imputation error's effect and finally reach the consistency rates introduced above; see Lemma \ref{lemma:beta3} and the proof of Theorem \ref{thm:nuisance}. Although we focus on a specific loss function \eqref{def:l2}, the results of part (b) in fact apply more broadly to other smooth and convex loss functions. 
 Since the nuisance estimators $\bgammahat,\bdeltahat,\balphahat,\bbetahat$ are constructed sequentially, and the later estimators depend on all the previous ones, the estimation errors of the nuisance parameters are cumulative, i.e., the consistency rate depends on the sparsity levels of all the nuisance parameters up to the current one.

\subsection{Results for correctly specified models}\label{sec:asymp_nuisance'}

If we have additional information that some of the nuisance models are correctly specified, we are able to achieve better consistency rates. 

\begin{theorem}\label{thm:nuisance'}
Let Assumptions \ref{cond:basic} and \ref{cond:subG} hold. Suppose that the sequences defined in \eqref{def:rs} satisfy $r_{\bgamma}+r_{\bdelta}+r_{\balpha}+r_{\bbeta}=o(1)$. Then, as $N,d_1,d_2\to\infty$, the following holds:
\begin{enumerate}
\item[(a)] Let $\rho=\rho^*$ and $s_{\bgamma}=O(N/(\log d_1\log d))$, then 
 $\|\bdeltahat-\bdelta^*\|_2=O_p(r_{\bdelta}).$
\item[(b)] Let $\nu=\nu^*$, $s_{\bgamma}$ as in (a) and $s_{\bdelta}=O(N/\log^2d)$, then 
 $\|\balphahat-\balpha^*\|_2=O_p(r_{\balpha}).$
\item[(c)] Let $\nu=\nu^*$, $\mu=\mu^*$, $s_{\bgamma}$ and $s_{\bdelta}$ are as in (b), then
 $\|\bbetahat-\bbeta^*\|_2 =O_p(r_{\balpha}+r_{\bbeta}).$
\item[(d)] Let $\rho=\rho^*$, $\mu=\mu^*$, and $s_{\bgamma}+s_{\bdelta}+s_{\balpha}=O(N/(\log d_1\log d))$, then 
 $\|\bbetahat-\bbeta^*\|_2 =O_p(r_{\bdelta}+r_{\bbeta}).$
 \item[(e)] Let $\rho=\rho^*$, $\nu=\nu^*$, $\mu=\mu^*$, $s_{\bgamma}$ and $s_{\bdelta}$ are as in (b) and $s_{\balpha}=O(N/(\log d_1\log d))$, then
 $\|\bbetahat-\bbeta^*\|_2 =O_p(r_{\bdelta}r_{\balpha}+r_{\bbeta}).$
\end{enumerate}
\end{theorem}

The new convergence rates in Theorem \ref{thm:nuisance'} are established through Lemmas \ref{lemma:RSC} and \ref{lemma:gradient'} of the Supplementary Material. Assuming certain nuisance models being correct, unlike Theorem \ref{thm:nuisance} and Lemma \ref{lemma:gradient}, we can control the gradients involving the \emph{estimated} nuisance parameters and control the imputation errors from the previous steps' nuisances in a more efficient way; see more details in Lemma \ref{lemma:gradient'}. As a result, we obtain faster convergence rates than Theorem \ref{thm:nuisance} given additional model correctness information.
 
 First, we see that the subsequent estimates, with correct model specifications, lose a factor of $r_{\bgamma}$ in their rates of estimation. Secondly, specific estimates demonstrate asymptotic decoupling: (i) the convergence rate of $\bdeltahat$ depends only on $s_{\bdelta}$ when $\rho^*$ is correctly specified; (ii) the convergence rate of $\balphahat$ depends only on $s_{\balpha}$ when $\nu^*$ is correctly specified. Lastly, the convergence of \(\bbetahat\) relies on the model correctness of \(\mu^*\), as well as the preceding models \(\rho^*\) and \(\nu^*\).
 Specifically, if either $\rho^*$ or $\nu^*$ is correctly specified, as explored in cases (c) and (d), the consistency rate of $\bbetahat$ depends on $s_{\bbeta}$ 	
 and the sparsity level of the correctly specified model, be it $\rho^*$ or $\nu^*$. When both $\rho^*$ and $\nu^*$ are accurate, as in case (e), the consistency rate of $\bbetahat$ depends on $s_{\bbeta}$ and a product sparsity $s_{\bdelta}s_{\balpha}$. 
 When a product sparsity condition, $s_{\bdelta}s_{\balpha}=o(N/\log^2d)$, is assumed as in \eqref{cond:s2} of Theorem \ref{thm:main}, $\bbetahat$ also becomes asymptotically decoupled from the other three estimates.

\section{Numerical Experiments}\label{sec:num}
\subsection{Simulation studies}\label{sec:sim}
We illustrate the finite sample properties of the introduced estimator on a number of simulated experiments. We focus on the estimation of $\theta=\theta_a-\theta_{a'}$ where $a=(a_1,a_2)=(1,1)$ and $a'=(a_1',a_2')=(0,0)$. We describe the considered data generating processes below. The outcome variables are generated as $Y_i=A_{1i}A_{2i}Y_i(1,1)+(1-A_{1i})(1-A_{2i})Y_i(0,0)$. 

Setting (a): Non-linear $\mu$ and non-logistic $\rho$.
Generate covariates at the first exposure: for each $i\leq N$, $\bS_{1i}\sim^\mathrm{iid} N_{d_1}(\mathbf{0},\mathbf{I}_{d_1}).$ The treatment indicators of the first exposure are generated as $A_{1i}\mid\bS_{1i}\sim\mathrm{Bernoulli}(g(\bS_{1i}^\top\bgamma))$. A $d$ dimensional vector of all ones and zeros are denoted with $\mathbf{1}_{(d)}$ and $\mathbf{0}_{(d)}$, respectively. Covariates at the second exposure satisfy
$\bS_{2i}=0.5Q(A_{1i})(\bS_{1i}^2-1)+Q(A_{1i})\bS_{1i}+A_{i}(1+\delta_{1i})\mathbf{1}_{(d_2)}+\bdelta_{1i},$ where $\bS_{1i}^2\in\R^{d_1}$ is the coordinate-wise square of $\bS_{1i}$, $\bdelta_{1i}\sim^\mathrm{iid} N_{d_2}(0,\mathbf{I}_{d_2})$, and a matrix $Q$ is defined with $\{Q(1)\}_{i,j}=0.8^{|i-j|}\mathbbm1\{|i-j|\leq1\}$ and $\{Q(0)\}_{i,j}=0.7^{|i-j|}\mathbbm1\{|i-j|\leq2\}$ for $i\leq d_2$ and $j\leq d_1$. The treatment indicators at the second exposure are generated as $A_{2i}\mid(\bar\bS_{2i},A_{1i})\sim\mathrm{Bernoulli}(A_{1i}\tilde g(\bar\bS_{2i}^\top\bdelta)+(1-A_{1i})\tilde g(-\bar\bS_{2i}^\top\bdelta))$, where $\tilde g(u):=(|u+1|+0.1)/(|u+1|+1)$. Lastly, $Y_i(1,1)=\bar\bS_{2i}^\top\balpha+1+\epsilon_i$, $Y_i(0,0)=-\bar\bS_{2i}^\top\balpha-1+\epsilon_i$ and $\epsilon_i\sim^\mathrm{iid}N(0,1)$. We consider $\balpha=(1,\mathbf0_{(d_1-1)},0.5,0.5,0.5,0.5,\mathbf0_{(d_2-4)})^\top$, $\bgamma=(1,1,\mathbf0_{(d_1-2)})^\top$ and $\bdelta=(1,\mathbf0_{(d_1-1)},0.5,0.5,0.5,0.5,\mathbf0_{(d_2-4)})^\top$.

Setting (b): Non-linear $\mu$ and non-linear $\nu$.
At the first exposure, generate covariates from a centered Beta distribution, i.e., $\bS_{1ij}\sim^\mathrm{iid} \mathrm{Beta}(1,2)-1/3$ for each $i\leq N$ and $j\leq d_1$; generate $A_{1i}\mid\bS_{1i}\sim\mathrm{Bernoulli}(g(\bS_{1i}^\top\bgamma))$. At the second exposure, generate
$\bS_{2i}=W(A_{1i})\bS_{1i}+A_{2i}\mathbf{1}_{(d_2)}+\bdelta_i$ and $A_{2i}\mid(\bar\bS_{2i},A_{1i})\sim\mathrm{Bernoulli}(A_{1i}g(\bar\bS_{2i}^\top\bdelta)+(1-A_{1i})g(-\bar\bS_{2i}^\top\bdelta))$, where $\bdelta_{ij}\sim^\mathrm{iid} \mathrm{Beta}(1,4)-1/5$, $\{W(1)\}_{i,j}=0.2^{|i-j|}\mathbbm1\{|i-j|\leq1\}$ and $\{W(0)\}_{i,j}=0.2^{|i-j|}\mathbbm1\{|i-j|\leq2\}+0.1\mathbbm1\{|i-j|=2\}$ for each $i\leq d_2$ and $j\leq d_1$. Here, $Y_i(1,1)=\bar\bS_{2i}^\top\balpha-1+2r_i+\epsilon_i$, $Y_i(0,0)=-\bar\bS_{2i}^\top\balpha+1-2r_i+\epsilon_i$ and $\epsilon_i\sim^\mathrm{iid}N(0,1)$. Here, we consider non-linear signals with $r_i$ as the standardized version of $\bS_{1i1}\bS_{1i2}\mathbbm1\{\bS_{1i2}>0.3\}+\bS_{1i1}\bS_{1i3}\mathbbm1\{\bS_{1i1}>0.3\}+\bS_{1i2}\bS_{1i3}\mathbbm1\{\bS_{1i1}>0.3\}$. The parameters are $\balpha=(-1,0,0,1/18,\mathbf0_{(d_1-4)},-1,-1,-1,\mathbf0_{(d_2-3)})^\top$, $\bgamma=(1,1,\mathbf0_{(d_1-2)})^\top$ and $\bdelta=(-2,-2,\mathbf0_{(d_1+d_2-2)})^\top$.

For each setting, we consider dimensions $d_1=100$ and $d_2=50$ (resulting in $d=d_1+d_2=150$), with total sample sizes $N$ ranging from $400$ to $16,000$. The experiments are repeated 200 times. Our proposed SMDR estimator is denoted as SMDR1 (see Algorithm \ref{alg:BRDR} with $\K=5$). Additionally, we present a slightly modified version, SMDR2, which constructs all the nuisances on the entire sub-sample of $\mathcal{I}_{-k}$ in Steps 4-7 of Algorithm \ref{alg:BRDR}. 

For comparison, we include several existing estimators: the inverse probability weighting (IPW) estimator, where propensity score models are estimated using $\ell_1$-regularized logistic regression without cross-fitting; the sequential doubly robust (SDR) estimator by Luedtke et al. (2017), where nuisance functions are estimated through linear and logistic regression without $\ell_1$-regularization; a cross-fitted version of the SDR estimator \citep{rotnitzky2017multiply,diaz2023nonparametric} using random forest nuisance estimates, denoted as SDR-RF; the sequential doubly robust Lasso (S-DRL) estimator proposed by Bradic et al. (2021); and two versions of the dynamic treatment Lasso (DTL) estimator proposed by Bradic et al. (2021) and Bodory et al. (2022), denoted as DTL2 and DTL1, respectively. DTL2's nuisances are estimated using samples in $\mathcal{I}_{-k}$, while DTL1's nuisances use different sub-samples in $\mathcal{I}_{\bgamma}$, $\mathcal{I}_{\bdelta}$, $\mathcal{I}_{\balpha}$, and $\mathcal{I}_{\bbeta}$. Here, DTL1 and SMDR1 share the same type of sample splitting, while DTL2 and SMDR2 share the same type of sample splitting. The tuning parameters are chosen through 5-fold cross-validations. Additionally, we report the performance of a naive empirical difference estimator (empdiff), $\thetahat_\mathrm{empdiff}:=\sum_{i=1}^NA_{1i}A_{2i}Y_i/\sum_{i=1}^NA_{1i}A_{2i}-\sum_{i=1}^N(1-A_{1i})(1-A_{2i})Y_i/\sum_{i=1}^N(1-A_{1i})(1-A_{2i})$, as well as an oracle doubly robust estimator, $\thetahat_\mathrm{oracle}$, which uses the doubly robust score with correct nuisance functions. The results are reported in Tables \ref{table:settinga1}-\ref{table:settingb1}.

\begin{table}[h!]
\centering
\caption{Simulation under Setting (a) with $d_1=100$, $d=150$. Bias: empirical bias; RMSE: root mean square error; Length: average length of the $95\%$ confidence intervals; Coverage: average coverage of the $95\%$ confidence intervals; ESD: empirical standard deviation; ASD: average of estimated standard deviations. All the reported values (except Coverage) are based on robust (median-type) estimates. $N_1$ and $N_0$ denote the expected numbers of observations in the treatment groups $(1,1)$ and $(0,0)$, respectively.} \label{table:settinga1}

\begin{tabular}{lcccccccccc}
\toprule
Method&Bias&RMSE&Length&Coverage&&Bias&RMSE&Length&Coverage\\
\hline
\multicolumn{1}{c}{ } & \multicolumn{4}{c}{\cellcolor{gray!50} $N=400,N_1\approx136,N_0\approx68$}&&\multicolumn{4}{c}{ \cellcolor{gray!50} $N=1000,N_1\approx341,N_0\approx169$}\\
\cline{2-5}\cline{7-11}
oracle&0.037&0.298&1.627&0.955&&0.004&0.190&1.056&0.975\\
\cdashline{2-5}\cdashline{7-11}
empdiff&-0.409&0.460&0.491&0.325&&-0.392&0.392&0.310&0.120\\
\cdashline{2-5}\cdashline{7-11}
IPW&0.553&0.582&1.941&0.770&&0.730&0.730&1.254&0.390\\
\cdashline{2-5}\cdashline{7-11}
SDR&0.597&4.067&14.304&0.800&&0.126&0.479&2.396&0.930\\
\cdashline{2-5}\cdashline{7-11}
SDR-RF&0.294&0.353&1.564&0.885&&0.423&0.423&0.945&0.595\\
\cdashline{2-5}\cdashline{7-11}
DTL1&0.643&0.775&2.086&0.775&&0.488&0.500&1.048&0.580\\
\cdashline{2-5}\cdashline{7-11}
DTL2&0.466&0.499&1.597&0.775&&0.278&0.286&1.044&0.775\\
\cdashline{2-5}\cdashline{7-11}
S-DRL&0.507&0.520&1.632&0.710&&0.311&0.313&1.055&0.785\\
\cdashline{2-5}\cdashline{7-11}
SMDR1&0.664&0.666&1.633&0.600&&0.462&0.462&0.975&0.530\\
\cdashline{2-5}\cdashline{7-11}
SMDR2&0.307&0.359&1.546&0.840&&0.091&0.193&0.998&0.890\\
\hline
\multicolumn{1}{c}{ } & \multicolumn{4}{c}{\cellcolor{gray!50} $N=12000,N_1\approx4103,N_0\approx2033$}&&\multicolumn{4}{c}{ \cellcolor{gray!50} $N=16000,N_1\approx5471,N_0\approx2710$}\\
\cline{2-5}\cline{7-11}
oracle&-0.002&0.053&0.317&0.945&&0.007&0.056&0.276&0.955\\
\cdashline{2-5}\cdashline{7-11}
empdiff&-0.401&0.401&0.090&0.000&&-0.397&0.397&0.078&0.000\\
\cdashline{2-5}\cdashline{7-11}
IPW&0.946&0.946&0.418&0.000&&0.957&0.957&0.369&0.000 \\
\cdashline{2-5}\cdashline{7-11}
SDR&0.012&0.090&0.571&0.955&&0.011&0.073&0.508&0.945\\
\cdashline{2-5}\cdashline{7-11}
SDR-RF&0.560&0.560&0.276&0.000&&0.567&0.567&0.241&0.000\\
\cdashline{2-5}\cdashline{7-11}
DTL1&0.243&0.243&0.329&0.260&&0.212&0.212&0.293&0.235\\
\cdashline{2-5}\cdashline{7-11}
DTL2&0.137&0.141&0.355&0.670&&0.122&0.122&0.314&0.655\\
\cdashline{2-5}\cdashline{7-11}
S-DRL&0.143&0.143&0.356&0.650&&0.123&0.123&0.313&0.670 \\
\cdashline{2-5}\cdashline{7-11}
SMDR1&0.053&0.075&0.311&0.890&&0.048&0.069&0.269&0.920\\
\cdashline{2-5}\cdashline{7-11}
SMDR2&0.020&0.058&0.319&0.935&&0.013&0.053&0.277&0.925\\
\bottomrule
\end{tabular}
\end{table}

\begin{table}[h]
\centering
\caption{Simulation under Setting (b) with $d_1=100$, $d=150$. The rest of the caption details remain the same as those in Table \ref{table:settinga1}.} \label{table:settingb1}
\begin{tabular}{lcccccccccc}
\toprule
Method&Bias&RMSE&Length&Coverage&&Bias&RMSE&Length&Coverage\\
\hline
\multicolumn{1}{c}{ } & \multicolumn{4}{c}{\cellcolor{gray!50} $N=400,N_1\approx109,N_0\approx92$}&&\multicolumn{4}{c}{ \cellcolor{gray!50} $N=1000,N_1\approx271,N_0\approx227$}\\
\cline{2-5}\cline{7-11}
oracle&-0.016&0.140&1.038&0.980&&-0.007&0.110&0.663&0.940\\
\cdashline{2-5}\cdashline{7-11}
empdiff&-0.066&0.220&0.369&0.450&&-0.097&0.173&0.239&0.330\\
\cdashline{2-5}\cdashline{7-11}
IPW&-0.170&0.266&1.581&0.950&&-0.044&0.153&0.941&0.980\\
\cdashline{2-5}\cdashline{7-11}
SDR&-0.259&3.874&14.351&0.740&&-0.156&0.386&1.303&0.875\\
\cdashline{2-5}\cdashline{7-11}
SDR-RF&-0.109&0.219&1.382&0.940&&-0.121&0.169&0.844&0.925\\
\cdashline{2-5}\cdashline{7-11}
DTL1&-0.078&0.366&1.908&0.940&&-0.161&0.220&0.972&0.895\\
\cdashline{2-5}\cdashline{7-11}
DTL2&-0.161&0.252&1.410&0.905&&-0.160&0.194&0.825&0.865\\
\cdashline{2-5}\cdashline{7-11}
S-DRL&-0.160&0.265&1.418&0.915&&-0.153&0.187&0.815&0.895\\
\cdashline{2-5}\cdashline{7-11}
SMDR1&-0.083&0.223&1.642&0.925&&-0.154&0.193&0.954&0.880\\
\cdashline{2-5}\cdashline{7-11}
SMDR2&-0.156&0.232&1.393&0.910&&-0.114&0.146&0.816&0.925\\
\hline
\multicolumn{1}{c}{ } & \multicolumn{4}{c}{\cellcolor{gray!50} $N=12000,N_1\approx3264,N_0\approx2726$}&&\multicolumn{4}{c}{ \cellcolor{gray!50} $N=16000,N_1\approx4352,N_0\approx3634$}\\
\cline{2-5}\cline{7-11}
oracle&-0.005&0.034&0.192&0.950&&0.002&0.030&0.166&0.960\\
\cdashline{2-5}\cdashline{7-11}
empdiff&-0.103&0.103&0.069&0.135&&-0.094&0.094&0.060&0.125\\
\cdashline{2-5}\cdashline{7-11}
IPW&0.034&0.048&0.274&0.945&&0.033&0.045&0.237&0.965\\
\cdashline{2-5}\cdashline{7-11}
SDR&-0.010&0.050&0.266&0.940&&-0.009&0.042&0.231&0.965\\
\cdashline{2-5}\cdashline{7-11}
SDR-RF&-0.090&0.090&0.219&0.600&&-0.087&0.087&0.189&0.565\\
\cdashline{2-5}\cdashline{7-11}
DTL1&-0.123&0.123&0.237&0.475&&-0.100&0.100&0.210&0.540\\
\cdashline{2-5}\cdashline{7-11}
DTL2&-0.072&0.073&0.246&0.775&&-0.060&0.062&0.217&0.880\\
\cdashline{2-5}\cdashline{7-11}
S-DRL&-0.060&0.069&0.246&0.790&&-0.048&0.051&0.215&0.815\\
\cdashline{2-5}\cdashline{7-11}
SMDR1&-0.035&0.051&0.227&0.890&&-0.018&0.037&0.198&0.950\\
\cdashline{2-5}\cdashline{7-11}
SMDR2&-0.013&0.047&0.228&0.940&&0.000&0.034&0.199&0.935\\
\bottomrule
\end{tabular}
\end{table}

Due to the confounding factors, the naive empirical difference estimator $\thetahat_\mathrm{empdiff}$ is not consistent with large biases and poor coverage; see Tables \ref{table:settinga1}-\ref{table:settingb1}. The IPW estimator also has very large biases (especially when $N$ is large) and provides bad coverage results under Setting (a), where the PS model at the second exposure is misspecified. In Setting (b) where both PS models are correctly specified, surprisingly, the IPW estimator provides acceptable coverages although there is no theoretical guarantees from existing work in high dimensions. However, the RMSEs of IPW are comparable with SMDR2 only when $N=12000$, and worse than SMDR2 in other cases; see Table \ref{table:settingb1}.

In scenarios where the sample size is relatively small compared to the dimensionality, the SDR estimator exhibits substantial estimation errors due to the absence of regularization in the nuisance estimation process. As the sample size increases, the SDR estimator tends to yield more acceptable estimation and inference results; however, its efficiency remains notably inferior to that of the proposed SMDR1 and SMDR2 estimators. On the other hand, the SDR-RF estimator typically yields satisfactory results in scenarios with relatively small sample sizes. However, as the sample size increases, the inferential outcomes, as well as the estimation results in Setting (a), tend to deteriorate. This phenomenon occurs due to the relatively slow convergence rates of random forests for nuisance estimation, leading to notable biases that cannot be ignored and consequently compromising the accuracy of inference.

The DTL1 estimator exhibits relatively poor performance overall, with biases often close to RMSE, and coverages far below the desired $95\%$. This suboptimal performance arises from two main factors: (i) The DTL estimators are only proven to be consistent when model misspecification occurs \citep{bradic2024high} and are not necessarily $\sqrt N$-consistent nor asymptotically normal; (ii) The sample splitting method of DTL1 is inefficient in finite samples, as only $1/5$ of the samples are used to obtain each nuisance estimator when $\K=5$. DTL2 is constructed using a more efficient sample splitting, leading to smaller biases than DTL1. However, it fails to achieve satisfactory coverage guarantees even with a large sample size. 

The S-DRL estimator is constructed similarly to DTL2, except with a different doubly robust estimation strategy for the first OR model. When the sample size is large enough, the S-DRL method provides RMSEs similar to (see Table \ref{table:settinga1}) or smaller than (see Table \ref{table:settingb1}) the DTL2 estimator. However, the coverages based on S-DRL remain below the desired $95\%$, even with a large sample size. Interestingly, the DTL and S-DRL methods yield coverages closer to $95\%$ when the sample size $N$ is small; however, an increase in the sample size does not lead to better coverage. This is due to the biases of these methods decaying slower than the parametric rate under model misspecification, making the normal approximation inaccurate.

For sufficiently large sample sizes, the proposed SMDR1 estimator consistently outperforms IPW, SDR, SDR-RF, DTL1, DTL2, and S-DRL in terms of estimation, exhibiting smaller biases and RMSEs across all considered settings (refer to Tables \ref{table:settinga1}-\ref{table:settingb1}). Moreover, SMDR1 provides satisfactory coverages that closely approach the desired $95\%$. However, its performance can be occasionally inferior to that of SDR-RF, DTL2, and S-DRL when the total sample size $N$ is small, as observed in Table \ref{table:settinga1} for $N\in\{400,1000\}$. Notably, SMDR1 consistently outperforms the DTL1 method, which utilizes the same type of sample splitting. This observation suggests the inadequacy of the sample splitting technique introduced in SMDR1. Indeed, when $N=1000$, only approximately $N_0\approx341$ samples are observed with the treatment path $(1,1)$ under Setting (a). Consequently, Step 4 of Algorithm \ref{alg:BRDR} results in only $N_0(\K-1)/(4\K)=N_0/5\approx68$ training samples for nuisance estimation, while the nuisance parameters have dimensions $d_1=100$ and $d=150$ for the first and second exposures, respectively. In contrast, the SMDR2 method employs $N_0(\K-1)/\K=0.8N_0\approx273$ training samples under the same setup. Notably, the SMDR2 method yields more stable results when $N$ is small, especially under Setting (a); see Table \ref{table:settinga1}. Therefore, while we acknowledge that the SMDR2 estimator may demand more stringent sparsity conditions compared to SMDR1 from a theoretical perspective, we recommend also considering the more efficient sample splitting technique of SMDR2, particularly when the sample size is small.

\subsection{A semi-synthetic analysis based on the National Job Corps Study (NJCS)}

\begin{table}[h!]
\centering
\caption{Semi-synthetic analysis. Bias: empirical bias; CI: the $95\%$ confidence interval; p-value: the p-value of $H_0: \theta=0$ v.s. $H_1: \theta\neq0$.} \label{table:lin}
\resizebox{15cm}{!}{
\begin{tabular}{lcccccccccccc}
\toprule
Method&$\thetahat_O$&$\thetahat$&Bias&CI&p-value&&$\thetahat_O$&$\thetahat$&Bias&CI&p-value&\\
\hline
\multicolumn{1}{c}{ } & \multicolumn{12}{c}{\cellcolor{gray!50} Setting (a)}\\
\cline{2-13}
DTL2&\multirow{3}{*}{0.15}&0.103&-0.047&[-0.090, 0.297]&0.295&&\multirow{3}{*}{0.20}&0.153&-0.047&[-0.040, 0.347]&0.120\\
\cdashline{3-6}\cdashline{9-12}
S-DRL&&0.092&-0.058&[-0.101, 0.308]&0.378&&&0.142&-0.058&[-0.051, 0.358]&0.173\\
\cdashline{3-6}\cdashline{9-12}
SMDR1&&0.180&0.030&[-0.019, 0.380]&0.076&&&0.230&0.030&[0.031, 0.430]&0.024\\
\hline
DTL2&\multirow{3}{*}{0.25}&0.203&-0.047&[0.010, 0.397]&0.039&&\multirow{3}{*}{0.30}&0.253&-0.047&[0.060, 0.447]&0.010\\
\cdashline{3-6}\cdashline{9-12}
S-DRL&&0.192&-0.058&[-0.001, 0.408]&0.066&&&0.241&-0.059&[0.049, 0.458]&0.021\\
\cdashline{3-6}\cdashline{9-12}
SMDR1&&0.280&0.030&[0.080, 0.480]&0.005&&&0.330&0.030&[0.131, 0.530]&0.001\\
\hline
\multicolumn{1}{c}{ } & \multicolumn{12}{c}{\cellcolor{gray!50} Setting (b)}\\
\cline{2-13}
DTL2&\multirow{3}{*}{0.25}&-0.027&-0.277&[-0.223,0.168]&0.783&&\multirow{3}{*}{0.30}&0.023&-0.277&[-0.172,0.218]&0.817\\
\cdashline{3-6}\cdashline{9-12}
S-DRL&&0.013&-0.237&[-0.232, 0.178]&0.902&&&0.063&-0.237&[-0.182,0.228]&0.548\\
\cdashline{3-6}\cdashline{9-12}
SMDR1&&0.141&-0.109&[-0.038, 0.320]&0.123&&&0.192&-0.108&[0.013, 0.371]&0.035\\
\hline
DTL2&\multirow{3}{*}{0.35}&0.072&-0.278&[-0.123,0.268]&0.468&&\multirow{3}{*}{0.40}&0.122&-0.278&[-0.074,0.317]&0.222\\
\cdashline{3-6}\cdashline{9-12}
S-DRL&&0.113&-0.237&[-0.133,0.277]&0.281&&&0.162&-0.238&[-0.083,0.327]&0.121\\
\cdashline{3-6}\cdashline{9-12}
SMDR1&&0.242&-0.108&[0.063, 0.421]&0.008&&&0.291&-0.109&[0.111, 0.470]&0.001\\
\hline
DTL2&\multirow{3}{*}{0.45}&0.172&-0.278&[-0.023,0.367]&0.084&&\multirow{3}{*}{0.50}&0.223&-0.277&[0.028,0.418]&0.025\\
\cdashline{3-6}\cdashline{9-12}
S-DRL&&0.212&-0.238&[-0.033,0.377]&0.043&&&0.263&-0.237&[0.018,0.428]&0.012\\
\cdashline{3-6}\cdashline{9-12}
SMDR1&&0.340&-0.110&[0.161, 0.519]&0.000&&&0.406&-0.094&[0.225, 0.586]&0.000\\
\bottomrule
\end{tabular}}
\end{table}

In this section, we compare the estimation and inference performance of the DTE estimators through semi-synthetic experiments. We consider a dataset from the National Job Corps Study, which is the largest and most comprehensive job training program in the US established in 1964, and serves approximately 50,000 disadvantaged youths aged 16-24 each year by providing vocational training and academic education. A detailed description of the original design and main effects can be found in \cite{schochet2008does,schochet2001national}.

We consider a dataset of 11,313 individuals, with 6,828 assigned to the Job Corps and 4,485 not. Treatments, denoted as $Z_{ti}\in\{0,1,2,3\}$ ($t\in{1,2}$), are assigned to the $i$th individual in the first and second years after the initial randomization, where $Z_{ti}=0$ represents non-enrollment, $Z_{ti}=1$ enrollment without program participation, $Z_{ti}=2$ high-school-level education, and $Z_{ti}=3$ vocational training. The baseline covariate vector, $\bS_{1i}$, has 909 characteristics, while $\bS_{2i}$ includes 1,427 characteristics that are evaluated before the second-year treatment assignment. We exclude 2,610 individuals whose treatment stages are missing completely at random \citep{Schochet2003national}, resulting in a final sample of 8,703 individuals. We also exclude the binary characteristics, if the 0/1 groups are extremely unbalanced in that the minority group's size is less than 10 within the 8703 individuals, resulting in the final $\bS_{1i}$ with 891 characteristics and $\bS_{2i}$ with 1350 characteristics. After standardizing the covariates, we generate the potential outcomes $\Ytil_i(z)$ corresponding to treatment paths $z=(z_1,z_2)\in\{0,1,2,3\}^2$ based on Settings (a) and (b) below. The observed outcome is generated as $Y_i=Y_i(Z_{i1},Z_{i2})=\sum_{z\in\{0,1,2,3\}^2}\mathbbm1_{\{(Z_{1i},Z_{2i})=z\}}Y_i(z)$.

We consider estimation of the DTE, $\theta=E\{\Ytil_i(z)\}-E\{\Ytil_i(z')\}$, focusing on the treatment path $z=(z_1,z_2)=(3,3)$ and the control path $z'=(z_1',z_2')=(1,1)$. To estimate the expected potential outcome $E\{\Ytil_i(z)\}$, we set $A_{1i}=\mathbbm1_{\{Z_{1i}=z_1\}}$, $A_{2i}=\mathbbm1_{\{Z_{2i}=z_2\}}$, and $Y_i(1,1)=\Ytil_i(z)$. Then $E\{\Ytil_i(z)\}=E\{Y_i(1,1)\}$ can be estimated using Algorithm \ref{alg:BRDR}. The control arm $E\{\Ytil_i(z')\}=E\{Y_i(0,0)\}$ can be estimated analogously, and the final DTE estimator is constructed as the difference of the obtained estimates. Let $\epsilon_i\sim^\mathrm{iid}N(0,1)$. The potential outcomes $Y_i(1,1)=\Ytil_i(z)$ and $Y_i(0,0)=\Ytil_i(z')$ are generated as below.

\begin{figure*}
	\captionsetup[subfloat]{labelformat=empty}
	\subfloat[Setting (a)]{
	\includegraphics[height=0.35\linewidth,width=0.45\linewidth]{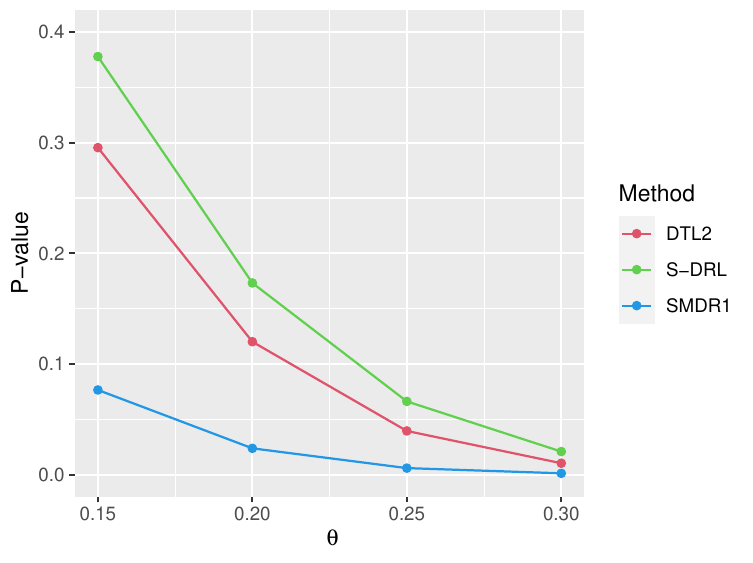}
	}
 \hfill
 \subfloat[Setting (b)]{
 \includegraphics[height=0.35\linewidth,width=0.45\linewidth]{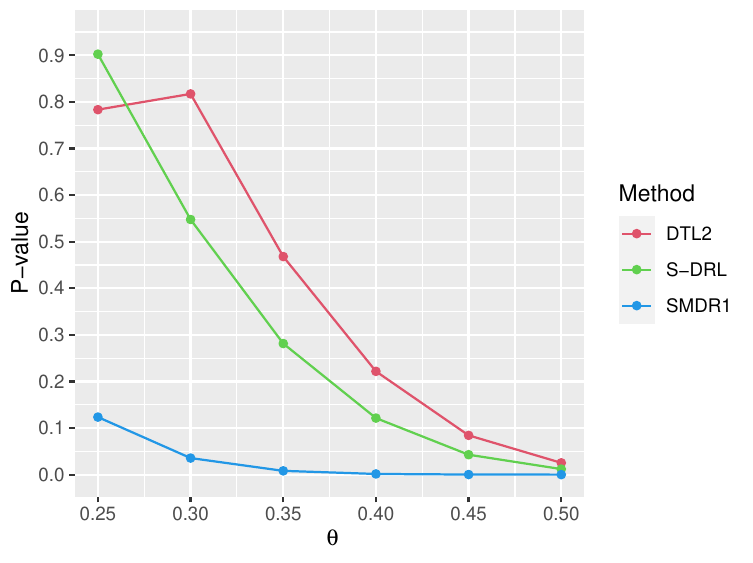}
 }
\caption{\centering The p-values for the null $H_0: \theta = 0$ as $\hat{\theta}_O$ varies. Since the true $\theta$ is unknown, the x-axis denotes the oracle difference-in-mean estimate of $\theta$.}\label{fig_JC}
\end{figure*}

Setting (a): Linear $\nu$. Let $Y_i(1,1)=\bar\bS_{2i}^\top\balpha+\epsilon_i$ and $Y_i(0,0)=-\bar\bS_{2i}^\top\balpha+\epsilon_i$, where $\balpha=0.5\cdot(\alpha_0,\mathbf{1}_{(8)},\mathbf{0}_{(d_1-8)}, \mathbf{1}_{(4)},\mathbf0_{(d_2-4)})^\top$ with $\alpha_0$ varying from 0.15 to 0.3.

Setting (b): Non-linear $\nu$. Let $Y_i(1,1)=\bS_{1i}^\top\balpha_1+ (\bS_{2i}^2-1)^\top\balpha_2+\epsilon_i$ and $Y_i(0,0)=-\bS_{1i}^\top\balpha_1-(\bS_{2i}^2-1)^\top\balpha_2+\epsilon_i$, where $\bS_{2i}^2 $ is the coordinate-wise square of $\bS_{2i}$, $\balpha_1=0.5\cdot(\alpha_0,\mathbf{1}_{(8)},\mathbf{0}_{(d_1-8)})^\top$, $\balpha_2=0.05\cdot (\mathbf{1}_{(4)},\mathbf0_{(d_2-4)})^\top$, and $\alpha_0$ varies from 0.25 to 0.5. 

For each setting, we implement the DTL2, S-DRL, and the proposed SMDR1 estimators (see Section \ref{sec:sim}). The results are reported in Table \ref{table:lin}, where the biases are calculated based on the oracle difference-in-mean estimate $\thetahat_O:=N^{-1}\sum_{i=1}^N\{Y_i(1,1)-Y_i(0,0)\}=\alpha_0$. Recall that under our simulated outcome setting, we get to see all potential outcomes: two per individual. Under both Settings (a) and (b), the proposed SMDR1 method provides smaller absolute biases than the DTL2 and S-DRL estimators. In addition, under Setting (a), where the OR model at the second exposure is truly linear, all the constructed confidence intervals contain the oracle estimate $\thetahat_O$. However, when the potential outcome is generated through a quadratic function (under Setting (b)), the oracle estimate $\thetahat_O$ does not lie in the confidence intervals based on the DTL2 and S-DRL methods; on the other hand, the proposed SMDR1 method leads to confidence intervals containing the oracle estimate. Moreover, considering the hypothesis testing problem with the null $H_0:\theta=0$ and the alternative $H_1:\theta\neq0$, the reported p-values decay as $\thetahat_O=\alpha_0$ grows; see Figure \ref{fig_JC}. When $\alpha_0$ is large enough, all the methods return p-values smaller than $0.05$; however, different methods require different signal levels to detect the causal effect and reject the null successfully. Under Setting (a), the proposed SMDR1 method is able to detect the causal effect with a significance level of $95\%$ when $\alpha_0=0.2$; however, under the same signal level, both the DTL2 and S-DRL methods fail to reject the null as the corresponding p-values are larger than $0.05$. Similarly, under Setting (b) with $\thetahat_O=\alpha_0=0.3$, the proposed SMDR1 method is able to detect the causal effect, whereas the p-value based on the DTL2 and S-DRL methods are both very large. Therefore, we observed a significantly better power in the SMDR1 method than both DTL2 and S-DRL.

\section{Discussion}\label{sec:dis}

This paper introduces new techniques to enable statistical inference for treatment effects in dynamic, high-dimensional, and potentially misspecified settings. By proposing a set of novel loss functions for nuisance models, we develop a sequential model doubly robust (SMDR) method that achieves root-\(N\) inference under minimal requirements. Our findings highlight the critical role of nuisance model estimation—naive, off-the-shelf estimators fail to achieve the desired robustness, even within doubly robust frameworks. While some nuisance models can be estimated independently, our results demonstrate that adopting a sequential estimation approach with nested designs significantly reduces the final estimation error for causal parameters. This observation raises an intriguing question: does this phenomenon persist in other statistical estimation problems, particularly in complex longitudinal settings requiring multi-stage estimation?

In the context of dynamic treatment regimes, a related but distinct doubly robust (DR) property has been explored. Existing methods for consistently estimating the optimal regime often require correctly specified contrast models at all later stages of estimation \citep{schulte2014q, shi2018high}. This condition is highly restrictive, especially in settings with multiple exposure occasions, and is not required in our framework. A natural question arises: How should decisions be made if contrast models cannot be accurately specified at later stages? Our results, outlined in Theorem \ref{thm:nuisance}, suggest that outcome regression (OR) models, such as \(\E\{Y(a_1,a_2) \mid \bS_1 = \bs_1\}\), can still be consistently estimated under these conditions. Consequently, optimizing the estimated OR functions at the first exposure over possible treatment paths offers a conservative yet viable strategy for newly arriving individuals. Notably, this approach eliminates the need for additional covariate evaluations at later stages, making it particularly useful in applications where accessing longitudinal covariates is costly or impractical.

The proposed algorithms can also be implemented using linear or logistic forms with basis functions, such as B-splines. However, further theoretical analysis is required for such approaches, as well as for the application of other non-parametric methods, including random forests and boosting. Additionally, while our methods leverage sparse structures in the models, future research should explore strategies that accommodate dense models with robust guarantees, broadening the applicability of our framework.

\section*{Supplementary Material}\label{supp_mat}

Sections \ref{sec:notation}-\ref{sec:proof_lemmas} contain additional discussions, justifications, and proofs of the main results. Additional notations used in the supplementary material are introduced in Section \ref{sec:notation}. Section \ref{sec:exist_unique} discusses the uniqueness of the moment-targeted parameters; the justification of their identification is provided in Section \ref{sec:just}. We introduce some useful auxiliary lemmas in Section \ref{sec:lemmas}. The proofs of main results and auxiliary lemmas are in Sections \ref{sec:proof_DTE} and \ref{sec:proof_lemmas}, respectively.

\appendix
\bibliographystyle{abbrvnat}
\bibliography{DTE}

\begin{thebibliography}{41}
\providecommand{\natexlab}[1]{#1}
\providecommand{\url}[1]{\texttt{#1}}
\expandafter\ifx\csname urlstyle\endcsname\relax
  \providecommand{\doi}[1]{doi: #1}\else
  \providecommand{\doi}{doi: \begingroup \urlstyle{rm}\Url}\fi

\bibitem[Avagyan and Vansteelandt(2021)]{avagyan2021high}
V.~Avagyan and S.~Vansteelandt.
\newblock High-dimensional inference for the average treatment effect under
  model misspecification using penalized bias-reduced double-robust estimation.
\newblock \emph{Biostatistics $\&$ Epidemiology}, pages 1--18, 2021.

\bibitem[Babino et~al.(2019)Babino, Rotnitzky, and Robins]{babino2019multiple}
L.~Babino, A.~Rotnitzky, and J.~Robins.
\newblock Multiple robust estimation of marginal structural mean models for
  unconstrained outcomes.
\newblock \emph{Biometrics}, 75\penalty0 (1):\penalty0 90--99, 2019.

\bibitem[Bang and Robins(2005)]{bang2005doubly}
H.~Bang and J.~M. Robins.
\newblock Doubly robust estimation in missing data and causal inference models.
\newblock \emph{Biometrics}, 61\penalty0 (4):\penalty0 962--973, 2005.

\bibitem[Bodory et~al.(2022)Bodory, Huber, and
  Laff{\'e}rs]{bodory2022evaluating}
H.~Bodory, M.~Huber, and L.~Laff{\'e}rs.
\newblock Evaluating (weighted) dynamic treatment effects by double machine
  learning.
\newblock \emph{The Econometrics Journal}, 25\penalty0 (3):\penalty0 628--648,
  2022.

\bibitem[Bradic et~al.(2019)Bradic, Wager, and Zhu]{bradic2019sparsity}
J.~Bradic, S.~Wager, and Y.~Zhu.
\newblock Sparsity double robust inference of average treatment effects.
\newblock \emph{arXiv preprint arXiv:1905.00744}, 2019.

\bibitem[Bradic et~al.(2024)Bradic, Ji, and Zhang]{bradic2024high}
J.~Bradic, W.~Ji, and Y.~Zhang.
\newblock High-dimensional inference for dynamic treatment effects.
\newblock \emph{The Annals of Statistics}, 52\penalty0 (2):\penalty0 415--440,
  2024.

\bibitem[Buja et~al.(2019)Buja, Brown, Berk, George, Pitkin, Traskin, Zhang,
  and Zhao]{buja2019models}
A.~Buja, L.~Brown, R.~Berk, E.~George, E.~Pitkin, M.~Traskin, K.~Zhang, and
  L.~Zhao.
\newblock Models as approximations i.
\newblock \emph{Statistical Science}, 34\penalty0 (4):\penalty0 523--544, 2019.

\bibitem[Chakrabortty et~al.(2019)Chakrabortty, Lu, Cai, and
  Li]{chakrabortty2019high}
A.~Chakrabortty, J.~Lu, T.~T. Cai, and H.~Li.
\newblock High dimensional m-estimation with missing outcomes: A
  semi-parametric framework.
\newblock \emph{arXiv preprint arXiv:1911.11345}, 2019.

\bibitem[Chen and Zhang(2023)]{chen2023enhancing}
K.~Chen and Y.~Zhang.
\newblock Enhancing efficiency and robustness in high-dimensional linear
  regression with additional unlabeled data.
\newblock \emph{arXiv preprint arXiv:2311.17685}, 2023.

\bibitem[Chernozhukov et~al.(2017)Chernozhukov, Chetverikov, Demirer, Duflo,
  Hansen, and Newey]{chernozhukov2017double}
V.~Chernozhukov, D.~Chetverikov, M.~Demirer, E.~Duflo, C.~Hansen, and W.~Newey.
\newblock Double/debiased/neyman machine learning of treatment effects.
\newblock \emph{American Economic Review}, 107\penalty0 (5):\penalty0 261--265,
  2017.

\bibitem[D{\'\i}az et~al.(2023)D{\'\i}az, Williams, Hoffman, and
  Schenck]{diaz2023nonparametric}
I.~D{\'\i}az, N.~Williams, K.~L. Hoffman, and E.~J. Schenck.
\newblock Nonparametric causal effects based on longitudinal modified treatment
  policies.
\newblock \emph{Journal of the American Statistical Association}, 118\penalty0
  (542):\penalty0 846--857, 2023.

\bibitem[Dukes and Vansteelandt(2021)]{dukes2021inference}
O.~Dukes and S.~Vansteelandt.
\newblock Inference for treatment effect parameters in potentially misspecified
  high-dimensional models.
\newblock \emph{Biometrika}, 108\penalty0 (2):\penalty0 321--334, 2021.

\bibitem[Dukes et~al.(2020)Dukes, Avagyan, and Vansteelandt]{dukes2020doubly}
O.~Dukes, V.~Avagyan, and S.~Vansteelandt.
\newblock Doubly robust tests of exposure effects under high-dimensional
  confounding.
\newblock \emph{Biometrics}, 76\penalty0 (4):\penalty0 1190--1200, 2020.

\bibitem[D{\"u}mbgen et~al.(2010)D{\"u}mbgen, Van De~Geer, Veraar, and
  Wellner]{dumbgen2010nemirovski}
L.~D{\"u}mbgen, S.~A. Van De~Geer, M.~C. Veraar, and J.~A. Wellner.
\newblock Nemirovski's inequalities revisited.
\newblock \emph{The American Mathematical Monthly}, 117\penalty0 (2):\penalty0
  138--160, 2010.

\bibitem[Hern{\'a}n et~al.(2001)Hern{\'a}n, Brumback, and
  Robins]{hernan2001marginal}
M.~A. Hern{\'a}n, B.~Brumback, and J.~M. Robins.
\newblock Marginal structural models to estimate the joint causal effect of
  nonrandomized treatments.
\newblock \emph{Journal of the American Statistical Association}, 96\penalty0
  (454):\penalty0 440--448, 2001.

\bibitem[Kallus and Santacatterina(2021)]{kallus2018optimal}
N.~Kallus and M.~Santacatterina.
\newblock Optimal balancing of time-dependent confounders for marginal
  structural models.
\newblock \emph{Journal of Causal Inference}, 9\penalty0 (1):\penalty0
  345--369, 2021.

\bibitem[Lewis and Syrgkanis(2021)]{lewis2021double}
G.~Lewis and V.~Syrgkanis.
\newblock Double/debiased machine learning for dynamic treatment effects.
\newblock \emph{Advances in Neural Information Processing Systems},
  34:\penalty0 22695--22707, 2021.

\bibitem[Luedtke et~al.(2017)Luedtke, Sofrygin, van~der Laan, and
  Carone]{luedtke2017sequential}
A.~R. Luedtke, O.~Sofrygin, M.~J. van~der Laan, and M.~Carone.
\newblock Sequential double robustness in right-censored longitudinal models.
\newblock \emph{arXiv preprint arXiv:1705.02459}, 2017.

\bibitem[Murphy(2003)]{murphy2003optimal}
S.~A. Murphy.
\newblock Optimal dynamic treatment regimes.
\newblock \emph{Journal of the Royal Statistical Society: Series B (Statistical
  Methodology)}, 65\penalty0 (2):\penalty0 331--355, 2003.

\bibitem[Murphy et~al.(2001)Murphy, van~der Laan, Robins, and
  Group]{murphy2001marginal}
S.~A. Murphy, M.~J. van~der Laan, J.~M. Robins, and C.~P. P.~R. Group.
\newblock Marginal mean models for dynamic regimes.
\newblock \emph{Journal of the American Statistical Association}, 96\penalty0
  (456):\penalty0 1410--1423, 2001.

\bibitem[Ning et~al.(2020)Ning, Sida, and Imai]{ning2020robust}
Y.~Ning, P.~Sida, and K.~Imai.
\newblock Robust estimation of causal effects via a high-dimensional covariate
  balancing propensity score.
\newblock \emph{Biometrika}, 107\penalty0 (3):\penalty0 533--554, 2020.

\bibitem[Orellana et~al.(2010)Orellana, Rotnitzky, and
  Robins]{orellana2010dynamic}
L.~Orellana, A.~Rotnitzky, and J.~M. Robins.
\newblock Dynamic regime marginal structural mean models for estimation of
  optimal dynamic treatment regimes, part i: main content.
\newblock \emph{The International Journal of Biostatistics}, 6\penalty0 (2),
  2010.

\bibitem[Robins(1986)]{robins1986new}
J.~Robins.
\newblock A new approach to causal inference in mortality studies with a
  sustained exposure period—application to control of the healthy worker
  survivor effect.
\newblock \emph{Mathematical Modelling}, 7\penalty0 (9-12):\penalty0
  1393--1512, 1986.

\bibitem[Robins(1987)]{robins1987addendum}
J.~M. Robins.
\newblock Addendum to “a new approach to causal inference in mortality
  studies with a sustained exposure period—application to control of the
  healthy worker survivor effect”.
\newblock \emph{Computers $\&$ Mathematics with Applications}, 14\penalty0
  (9-12):\penalty0 923--945, 1987.

\bibitem[Robins(1997)]{Robins1997causal}
J.~M. Robins.
\newblock Causal inference from complex longitudinal data.
\newblock In \emph{Latent Variable Modeling and Applications to Causality},
  pages 69--117. Springer, 1997.

\bibitem[Robins(2004)]{robins2004optimal}
J.~M. Robins.
\newblock Optimal structural nested models for optimal sequential decisions.
\newblock In \emph{Proceedings of the Second Seattle Symposium in
  Biostatistics}, pages 189--326. Springer, 2004.

\bibitem[Robins et~al.(2000)Robins, Hern{\'a}n, and
  Brumback]{robins2000marginal}
J.~M. Robins, M.~{\'A}. Hern{\'a}n, and B.~Brumback.
\newblock Marginal structural models and causal inference in epidemiology.
\newblock \emph{Epidemiology}, pages 550--560, 2000.

\bibitem[Rotnitzky et~al.(2017)Rotnitzky, Robins, and
  Babino]{rotnitzky2017multiply}
A.~Rotnitzky, J.~Robins, and L.~Babino.
\newblock On the multiply robust estimation of the mean of the g-functional.
\newblock \emph{arXiv preprint arXiv:1705.08582}, 2017.

\bibitem[Schochet et~al.(2003)Schochet, Bellotti, Ruo-Jiao, Glazerman, Grady,
  Gritz, McConnell, Johnson, and Burghardt]{Schochet2003national}
P.~Schochet, J.~Bellotti, C.~Ruo-Jiao, S.~Glazerman, A.~Grady, M.~Gritz,
  S.~McConnell, T.~Johnson, and J.~Burghardt.
\newblock National job corps study: data documentation and public use files.
\newblock \emph{vols. I-IV). Washington, DC: Mathematica Policy Research, Inc},
  2003.

\bibitem[Schochet(2001)]{schochet2001national}
P.~Z. Schochet.
\newblock \emph{National Job Corps Study: The impacts of Job Corps on
  participants' employment and related outcomes}.
\newblock US Department of Labor, Employment and Training Administration,
  Office of Policy and Research, 2001.

\bibitem[Schochet et~al.(2008)Schochet, Burghardt, and
  McConnell]{schochet2008does}
P.~Z. Schochet, J.~Burghardt, and S.~McConnell.
\newblock Does job corps work? impact findings from the national job corps
  study.
\newblock \emph{American Economic Review}, 98\penalty0 (5):\penalty0
  1864--1886, 2008.

\bibitem[Schulte et~al.(2014)Schulte, Tsiatis, Laber, and
  Davidian]{schulte2014q}
P.~J. Schulte, A.~A. Tsiatis, E.~B. Laber, and M.~Davidian.
\newblock Q-and a-learning methods for estimating optimal dynamic treatment
  regimes.
\newblock \emph{Statistical science: a review journal of the Institute of
  Mathematical Statistics}, 29\penalty0 (4):\penalty0 640, 2014.

\bibitem[Shi et~al.(2018)Shi, Fan, Song, and Lu]{shi2018high}
C.~Shi, A.~Fan, R.~Song, and W.~Lu.
\newblock High-dimensional a-learning for optimal dynamic treatment regimes.
\newblock \emph{The Annals of Statistics}, 46\penalty0 (3):\penalty0 925, 2018.

\bibitem[Smucler et~al.(2019)Smucler, Rotnitzky, and
  Robins]{smucler2019unifying}
E.~Smucler, A.~Rotnitzky, and J.~M. Robins.
\newblock A unifying approach for doubly-robust $l_1$ regularized estimation of
  causal contrasts.
\newblock \emph{arXiv preprint arXiv:1904.03737}, 2019.

\bibitem[Tan(2020)]{tan2020model}
Z.~Tan.
\newblock Model-assisted inference for treatment effects using regularized
  calibrated estimation with high-dimensional data.
\newblock \emph{The Annals of Statistics}, 48\penalty0 (2):\penalty0 811--837,
  2020.

\bibitem[Van~der Vaart(2000)]{van2000asymptotic}
A.~W. Van~der Vaart.
\newblock \emph{Asymptotic statistics}, volume~3.
\newblock Cambridge university press, 2000.

\bibitem[Viviano and Bradic(2021)]{viviano2021dynamic}
D.~Viviano and J.~Bradic.
\newblock Dynamic covariate balancing: estimating treatment effects over time.
\newblock \emph{arXiv preprint arXiv:2103.01280}, 2021.

\bibitem[Wainwright(2019)]{wainwright2019high}
M.~J. Wainwright.
\newblock \emph{High-dimensional statistics: A non-asymptotic viewpoint},
  volume~48.
\newblock Cambridge University Press, 2019.

\bibitem[Yiu and Su(2018)]{yiu2018covariate}
S.~Yiu and L.~Su.
\newblock Covariate association eliminating weights: a unified weighting
  framework for causal effect estimation.
\newblock \emph{Biometrika}, 105\penalty0 (3):\penalty0 709--722, 2018.

\bibitem[Yu and van~der Laan(2006)]{yu2006double}
Z.~Yu and M.~van~der Laan.
\newblock Double robust estimation in longitudinal marginal structural models.
\newblock \emph{Journal of Statistical Planning and Inference}, 136\penalty0
  (3):\penalty0 1061--1089, 2006.

\bibitem[Zhang et~al.(2023)Zhang, Chakrabortty, and Bradic]{zhang2023double}
Y.~Zhang, A.~Chakrabortty, and J.~Bradic.
\newblock Double robust semi-supervised inference for the mean: Selection bias
  under mar labeling with decaying overlap.
\newblock \emph{Information and Inference: A Journal of the IMA}, 12\penalty0
  (3):\penalty0 2066--2159, 2023.

\end{thebibliography}

%%%%%%%%%%%%%%%%%%%%%%%%%%%%%%%%%%%%%%%%%%%%%%

\renewcommand{\thetheorem}{S.\arabic{theorem}}
\renewcommand{\thelemma}{S.\arabic{lemma}}

%%%%%%%%%%%%%%%%%%%%%%%%%%%%%%%%%%%%%%%%%%%%%%

\clearpage\newpage 
\begin{center}
\textbf{\uppercase{Supplementary Material to ``Robust inference for the dynamic treatment effect in high dimensions''}}
\end{center}

\par\medskip

\section{Notation}\label{sec:notation}
 We use the following notation throughout. Let $\P$ and $\E$ denote the probability measure and expectation characterizing the joint distribution of the underlying random vector $\mathbf W:=(\{Y(a_1,a_2)\}_{a_1,a_2\in\{0,1\}},A_1,A_2,\bS_1,\bS_2)$ (independent of the observed samples), respectively. 
For any $\alpha>0$, let $\psi_{\alpha}(x):=\exp(x^\alpha)-1$, $\forall x>0$. The $\psi_{\alpha}$-Orlicz norm $\|\cdot\|_{\psi_{\alpha}}$ of a random variable $X\in\R$ is defined as $\|X\|_{\psi_{\alpha}}:=\inf\{c>0:\E[\psi_{\alpha}(\lvert X\rvert/c)]\leq 1\}$. Two special cases are given by $\psi_2(x)=\exp(x^2)-1$ and $\psi_{1}(x)=\exp(x)-1$. We use $a_N\asymp b_N$ to denote $cb_N\leq a_N \leq Cb_N$ for all $N \geq 1$ and constants $c,C>0$. For any $\tilde{\S}\subseteq\S=(\bZ_i)_{i=1}^N$, define $\P_{\tilde{\S}}$ as the joint distribution of $\tilde{\S}$ and $\E_{\tilde{\S}}(f)=\int fd\P_{\tilde{\S}}$. For $r\geq1$, define the $l_r$-norm of a vector $\bz$ with $\|\bz\|_r:=(\sum_{j=1}^p\lvert\bz_j\rvert^r)^{1/r}$, $\|\bz\|_0 := \lvert\{j:\bz_j\neq0\}\rvert$, and $\|\bz\|_\infty:=\max_j\lvert\bz_j\rvert$. We denote the logistic function with $g(u)= {\exp(u)}/{(1+\exp(u))}$, for all $u\in\R.$ A $d$ dimensional vector of all ones and zeros are denoted with $\mathbf{1}_{(d)}$ and $\mathbf{0}_{(d)}$, respectively.
 
  Constants $c,C>0$, independent of $N$ and $d$, may change from one line to the other. For any $r > 0$, let $ \| f(\cdot) \|_{r,\P} := \{\E \lvert f(\bZ)\rvert^r\}^{1/r}$. Denote $\be_j$ as the vector whose $j$-th element is $1$ and other elements are $0$s. For any symmetric matrices $\mathbf{A}$ and $\mathbf{B}$, $\mathbf{A}\succ\mathbf{B}$ denotes that $\mathbf{A}-\mathbf{B}$ is positive definite and $\mathbf{A}\succeq\mathbf{B}$ denotes that $\mathbf{A}-\mathbf{B}$ is positive semidefinite; denote $\lambda_{\min}(\mathbf{A})$ as the smallest eigenvalue of $\mathbf{A}$. For the sake of simplicity, we denote $Y_{1,1}:=Y(1,1)$ in this supplementary document. Let $\S_{\bgamma}=(\bW_i)_{i\in\mathcal I_{\bgamma}}$, $\S_{\bdelta}=(\bW_i)_{i\in\mathcal I_{\bdelta}}$, $\S_{\balpha}=(\bW_i)_{i\in\mathcal I_{\balpha}}$, and $\S_{\bbeta}=(\bW_i)_{i\in\mathcal I_{\bbeta}}$ be the subsets of $\S$ corresponding to the index sets $\mathcal I_{\bgamma}$, $\mathcal I_{\bdelta}$, $\mathcal I_{\balpha}$, and $\mathcal I_{\bbeta}$, respectively.

\section{Uniqueness of moment-targeted  parameters}\label{sec:exist_unique}

In this section, we discuss the uniqueness of moment-targeted nuisance parameters $\bgamma^*,\bdelta^*,\balpha^*,\bbeta^*$ defined in Section \ref{sec:DR-DTE}. Note that the parameters are defined through loss functions \eqref{def:alpha}, \eqref{def:l2}, \eqref{def:l3}, and \eqref{def:l4}. We first consider the Hessian matrices of the objective functions:
\begin{align*}
\bnabla_{\bgamma}^2\E\{\ell_1(\bW;\bgamma)\}&=\E\left[A_1\{g^{-1}(\bS_1^\top\bgamma)-1\}\bS_1\bS_1^\top\right],\\
\bnabla_{\bdelta}^2\E\{\ell_2(\bW;\bgamma^*,\bdelta)\}&=\E\left[A_1A_2g^{-1}(\bS_1^\top\bgamma^*)\{g^{-1}(\bar\bS_2^\top\bdelta)-1\}\bar\bS_2\bar\bS_2^\top\right],\\
\bnabla_{\balpha}^2\E\{\ell_3(\bW;\bgamma^*,\bdelta^*,\balpha)\}&=2\E\left[A_1A_2g^{-1}(\bS_1^\top\bgamma^*)\{g^{-1}(\bar\bS_2^\top\bdelta^*)-1\}\bar\bS_2\bar\bS_2^\top\right],\\
\bnabla_{\bbeta}^2\E\{\ell_4(\bW;\bgamma^*,\bdelta^*,\balpha^*,\bbeta)\}&=2\E\left[A_1\{g^{-1}(\bS_1^\top\bgamma^*)-1\}\bS_1\bS_1^\top\right],
\end{align*}
where $g(u)=\exp(u)/\{1+\exp(u)\}$ is the logistic function. In the following, we will prove by contradiction to demonstrate that
\begin{align}\label{eq:pd_gamma}
\bnabla_{\bgamma}^2\E\{\ell_1(\bW;\bgamma)\}\succ\bzero,\;\;\forall\bgamma\in\R^{d_1}.
\end{align}
Assume there exists some $\ba,\bgamma\in\R^{d_1}$ such that $\ba^\top\bnabla_{\bgamma}^2\E\{\ell_1(\bW;\bgamma)\}\ba=0$ and $\|\ba\|_2=1$. Then $A_1\{g^{-1}(\bS_1^\top\bgamma)-1\}(\bS_1^\top\ba)^2=0$ almost surely. Since $g(u)\in(0,1)$ for all $u\in\R$, $g^{-1}(\bs_1^\top\bgamma)-1>0$ for all $\bs_1\in\R^{d_1}$ and hence $A_1(\bS_1^\top\ba)^2=0$ almost surely. It follows that $\E\{A_1(\bS_1^\top\ba)^2\}=0$ and hence $\lambda_{\min}(\E(A_1\bar\bS_2\bar\bS_2^\top))\leq\lambda_{\min}(\E(A_1\bS_1\bS_1^\top))=0$, which conflicts Assumption \ref{cond:subG}. Therefore, \eqref{eq:pd_gamma} holds, the solution $\bgamma^*$ is unique and can be equivalently defined as the solution of the first-order optimality condition $\E\{\bnabla_{\bgamma}\ell_1(\bW;\bgamma)\}=\bzero$.

In addition, we also note that
\begin{align*}
&\bnabla_{\bdelta}^2\E\{\ell_2(\bW;\bgamma^*,\bdelta)\}\\
&\qquad=\E\left(\E\left[A_2g^{-1}(\bS_1^\top\bgamma^*)\{g^{-1}(\bar\bS_2^\top\bdelta)-1\}\bar\bS_2\bar\bS_2^\top\mid\bar\bS_2,A_1=1\right]\E(A_1\mid\bar\bS_2)\right)\\
&\qquad=\E\left[\rho(\bar\bS_2)g^{-1}(\bS_1^\top\bgamma^*)\{g^{-1}(\bar\bS_2^\top\bdelta)-1\}\bar\bS_2\bar\bS_2^\top\E(A_1\mid\bar\bS_2)\right]\\
&\qquad=\E\left[A_1\rho(\bar\bS_2)g^{-1}(\bS_1^\top\bgamma^*)\{g^{-1}(\bar\bS_2^\top\bdelta)-1\}\bar\bS_2\bar\bS_2^\top\right]\\
&\qquad\succeq c_0\E\left[A_1g^{-1}(\bS_1^\top\bgamma^*)\{g^{-1}(\bar\bS_2^\top\bdelta)-1\}\bar\bS_2\bar\bS_2^\top\right]
\end{align*}
under Assumption \ref{cond:basic}. In the following, we will prove by contradiction to show that 
\begin{align}\label{eq:pd_delta}
\bnabla_{\bdelta}^2\E\{\ell_2(\bW;\bgamma^*,\bdelta)\}\succ\bzero,\;\;\forall\bdelta\in\R^d.
\end{align}
Assume there exists some $\ba,\bdelta\in\R^d$ such that $\ba^\top\bnabla_{\bdelta}^2\E\{\ell_2(\bW;\bgamma^*,\bdelta)\}\ba=0$ and $\|\ba\|_2=1$. Then $A_1g^{-1}(\bS_1^\top\bgamma^*)\{g^{-1}(\bar\bS_2^\top\bdelta)-1\}(\bar\bS_2^\top\ba)^2=0$ almost surely. Since $g(u)\in(0,1)$ for all $u\in\R$, $g^{-1}(\bs_1^\top\bgamma^*)\{g^{-1}(\bar\bs_2^\top\bdelta)-1\}>0$ for all $\bar\bs_2\in\R^d$ and hence $A_1(\bar\bS_2^\top\ba)^2$ almost surely. It follows that $\E\{A_1(\bar\bS_2^\top\ba)^2\}=0$ and hence $\lambda_{\min}(\E(A_1\bar\bS_2\bar\bS_2^\top))=0$, which conflicts Assumption \ref{cond:subG}. Therefore, \eqref{eq:pd_delta} holds, the solution $\bdelta^*$ is unique and can be equivalently defined as the solution of $\E\{\bnabla_{\bdelta}\ell_2(\bW;\bgamma^*,\bdelta)\}=\bzero$.

Besides, we note that
\begin{align}
\bnabla_{\balpha}^2\E\{\ell_3(\bW;\bgamma^*,\bdelta^*,\balpha)\}&=2\bnabla_{\bdelta}^2\E\{\ell_2(\bW;\bgamma^*,\bdelta^*)\}\succ\bzero,\;\;\forall\balpha\in\R^d,\label{eq:pd_alpha}\\
\bnabla_{\bbeta}^2\E\{\ell_4(\bW;\bgamma^*,\bdelta^*,\balpha^*,\bbeta)\}&=2\bnabla_{\bgamma}^2\E\{\ell_1(\bW;\bgamma^*)\}\succ\bzero,\;\;\forall\bbeta\in\R^{d_1}.\label{eq:pd_beta}
\end{align}
Therefore, the solutions $\balpha^*$ and $\bbeta^*$ are also unique.

\section{Justifications for Remark \ref{remark:correctness}}\label{sec:just}
In this section, we provide justifications for the results in Remark \ref{remark:correctness}. We first introduce equivalent expressions for the OR functions $\mu(\cdot)$ and $\nu(\cdot)$: under Assumption \ref{cond:basic},
\begin{align}
\nu(\bar\bs_2):=&\E(Y_{1,1}\mid\bar\bS_2=\bs_2,A_1=1)\overset{(iii)}{=}\E(Y_{1,1}\mid\bar\bS_2=\bs_2,A_1=1,A_2=1),\nonumber\\
\mu(\bs_1):=&\E(Y_{1,1}\mid\bS_1=\bs_1)\overset{(ii)}=\E(Y_{1,1}\mid\bS_1=\bs_1,A_1=1)\nonumber\\
\overset{(iii)}{=}&\E\{\nu(\bar\bS_2)\mid\bS_1=\bs_1,A_1=1\}.\label{eq:rep_mu}
\end{align}
Here, (i) holds since $Y_{1,1}\independent A_2\mid(\bar\bS_2,A_1=1)$; (ii) holds since $Y_{1,1}\independent A_1\mid\bS_1$; (iii) holds by the tower rule.

The justifications for cases (a)-(d) of Remark \ref{remark:correctness} are provided below:
\begin{enumerate}[label=(a.\arabic*),leftmargin=1mm, labelsep=0.4em, itemindent=3.9em]
\item Assume $\pi^*(\cdot)=\pi(\cdot)$. Then $\pi(\bs_1)=\pi^*(\bs_1)=g(\bs_1^\top\bgamma^*)$ and hence $\pi(\bs_1)=g(\bs_1^\top\bgamma^0)$ with $\bgamma^0=\bgamma^*$. 
\item \label{a.2}  Assume there exists some $\bgamma^0\in\R^{d_1}$ such that $\pi(\bs_1)=g(\bs_1^\top\bgamma^0)$. By the construction of $\bgamma^*$ and note that the Hessian matrix satisfies \eqref{eq:pd_gamma}, $\bgamma=\bgamma^*$ is the unique solution of
$$\E\{\bnabla_{\bgamma}\ell_1(\bW;\bgamma)\}=\E\left[\{1-A_1g^{-1}(\bS_1^\top\bgamma)\}\bS_1\right]=\bzero\in\R^{d_1}.$$
Meanwhile, we also have
\begin{align*}
&\E\left[\{1-A_1g^{-1}(\bS_1^\top\bgamma^0)\}\bS_1\right]=\E\left(\E\left[\{1-A_1g^{-1}(\bS_1^\top\bgamma^0)\}\bS_1\mid\bS_1\right]\right)\\
&\qquad=\E\left[\{1-\pi(\bS_1)g^{-1}(\bS_1^\top\bgamma^0)\}\bS_1\right]=\bzero\in\R^{d_1}.
\end{align*}
By the uniqueness of $\bgamma^*$, we conclude that $\bgamma^0=\bgamma^*$ and hence $\pi(\cdot)=\pi^*(\cdot)$.
\end{enumerate}

\begin{enumerate}[label=(b.\arabic*),leftmargin=1mm, labelsep=0.4em, itemindent=3.9em]
\item Assume $\rho^*(\cdot)=\rho(\cdot)$. Then $\rho(\bar\bs_2)=\rho^*(\bar\bs_2)=g(\bar\bs_2^\top\bdelta^*)$ and hence $\rho(\bar\bs_2)=g(\bar\bs_2^\top\bdelta^0)$ with $\bdelta^0=\bdelta^*$.

\item \label{b.2}  Assume there exists some $\bdelta^0\in\R^d$ such that $\rho(\bar\bs_2)=g(\bar\bs_2^\top\bdelta^0)$. By the construction of $\bdelta^*$ and note that the Hessian matrix satisfies \eqref{eq:pd_delta}, $\bdelta=\bdelta^*$ is the unique solution of
$$\E\{\bnabla_{\bdelta}\ell_2(\bW;\bgamma^*,\bdelta)\}=\E\left[A_1g^{-1}(\bS_1^\top\bgamma^*)\{1-A_2g^{-1}(\bar\bS_2^\top\bdelta)\}\bar\bS_2\right]=\bzero\in\R^{d}.$$
Meanwhile, we also have
\begin{align*}
&\E\left[A_1g^{-1}(\bS_1^\top\bgamma^*)\{1-A_2g^{-1}(\bar\bS_2^\top\bdelta^0)\}\bar\bS_2\right]\\
&\qquad=\E\left(\E\left[g^{-1}(\bS_1^\top\bgamma^*)\{1-A_2g^{-1}(\bar\bS_2^\top\bdelta^0)\}\bar\bS_2\mid\bar\bS_2,A_1=1\right]\E(A_1\mid\bar\bS_2)\right)\\
&\qquad=\E\left[g^{-1}(\bS_1^\top\bgamma^*)\{1-\rho(\bar\bS_2)g^{-1}(\bar\bS_2^\top\bdelta^0)\}\bar\bS_2\E(A_1\mid\bar\bS_2)\right]\\
&\qquad=\E\left[A_1g^{-1}(\bS_1^\top\bgamma^*)\{1-\rho(\bar\bS_2)g^{-1}(\bar\bS_2^\top\bdelta^0)\}\bar\bS_2\right]=\bzero\in\R^{d}.
\end{align*}
By the uniqueness of $\bdelta^*$, we conclude that $\bdelta^0=\bdelta^*$ and hence $\rho(\cdot)=\rho^*(\cdot)$.
\end{enumerate}

\begin{enumerate}[label=(c.\arabic*),leftmargin=1mm, labelsep=0.4em, itemindent=3.9em]
\item Assume $\nu^*(\cdot)=\nu(\cdot)$. Then $\nu(\bar\bs_2)=\nu^*(\bar\bs_2)=\bar\bs_2^\top\balpha^*$ and hence $\nu(\bar\bs_2)=\bar\bs_2^\top\balpha^0$ with $\balpha^0=\balpha^*$.

\item \label{c.2}  Assume there exists some $\balpha^0\in\R^d$ such that $\nu(\bar\bs_2)=\bar\bs_2^\top\balpha^0$. By the construction of $\balpha^*$ and note that the Hessian matrix satisfies \eqref{eq:pd_alpha}, $\balpha=\balpha^*$ is the unique solution of
\begin{align*}
\E\{\bnabla_{\balpha}\ell_3(\bW;\bgamma^*,\bdelta^*,\balpha)\}&=2\E\left[A_1A_2g^{-1}(\bS_1^\top\bgamma^*)\{g^{-1}(\bar\bS_2^\top\bdelta^*)-1\}(\bar\bS_2^\top\balpha-Y)\bar\bS_2\right]=\bzero\in\R^{d}.
\end{align*}
Meanwhile, we also have
\begin{align*}
&\E\left[A_1A_2g^{-1}(\bS_1^\top\bgamma^*)\{g^{-1}(\bar\bS_2^\top\bdelta^*)-1\}(\bar\bS_2^\top\balpha^0-Y)\bar\bS_2\right]\\
&\qquad=\E\left(\E\left[A_2g^{-1}(\bS_1^\top\bgamma^*)\{g^{-1}(\bar\bS_2^\top\bdelta^*)-1\}(\bar\bS_2^\top\balpha^0-Y_{1,1})\bar\bS_2\mid\bar\bS_2,A_1=1\right]\E(A_1\mid\bar\bS_2)\right)\\
&\qquad=\E\left[\E\left\{A_2(\bar\bS_2^\top\balpha^0-Y_{1,1})\mid\bar\bS_2,A_1=1\right\}g^{-1}(\bS_1^\top\bgamma^*)\{g^{-1}(\bar\bS_2^\top\bdelta^*)-1\}\bar\bS_2\E(A_1\mid\bar\bS_2)\right]\\
&\qquad\overset{(i)}{=}\E\left[\rho(\bar\bS_2)\{\bar\bS_2^\top\balpha^0-\nu(\bar\bS_2)\}g^{-1}(\bS_1^\top\bgamma^*)\{g^{-1}(\bar\bS_2^\top\bdelta^*)-1\}\bar\bS_2\E(A_1\mid\bar\bS_2)\right]=\bzero\in\R^{d},
\end{align*}
where (i) holds since $Y_{1,1}\independent A_2\mid(\bar\bS_2,A_1=1)$ under Assumption \ref{cond:basic}. By the uniqueness of $\balpha^*$, we conclude that $\balpha^0=\balpha^*$ and hence $\nu(\cdot)=\nu^*(\cdot)$.
\end{enumerate}

\begin{enumerate}[label=(d),leftmargin=1mm, labelsep=0.4em, itemindent=3.9em]
\item  Assume there exists some $\bbeta^0\in\R^{d_1}$ such that $\mu(\bs_1)=\bs_1^\top\bbeta^0$ and either $\rho^*(\cdot)=\rho(\cdot)$ or $\nu^*(\cdot)=\nu(\cdot)$ holds. By the construction of $\bbeta^*$ and note that the Hessian matrix satisfies \eqref{eq:pd_beta}, $\bbeta=\bbeta^*$ is the unique solution of
\begin{align*}
&\E\{\bnabla_{\bbeta}\ell_4(\bW;\bgamma^*,\bdelta^*,\balpha^*,\bbeta)\}\\
&\qquad=2\E\left[A_1\{g^{-1}(\bS_1^\top\bgamma^*)-1\}\left\{\bS_1^\top\bbeta-\bar\bS_2^\top\balpha^*-A_2g^{-1}(\bar\bS_2^\top\bdelta^*)(Y-\bar\bS_2^\top\balpha^*)\right\}\bS_1\right]\\
&\qquad=\bzero\in\R^{d_1}.
\end{align*}
Meanwhile, note that
\begin{align*}
&\E\left[\exp(-\bS_1^\top\bgamma^*)\left\{\bS_1^\top\bbeta^0-\bar\bS_2^\top\balpha^*-A_2g^{-1}(\bar\bS_2^\top\bdelta^*)(Y-\bar\bS_2^\top\balpha^*)\right\}\bS_1\mid\bar\bS_2,A_1=1\right]\\
&\qquad=\exp(-\bS_1^\top\bgamma^*)\left(\mu(\bS_1)-\nu^*(\bar\bS_2)-\frac{\E\left[A_2\{Y_{1,1}-\nu^*(\bar\bS_2)\}\mid\bar\bS_2,A_1=1\right]}{\rho^*(\bar\bS_2)}\right)\bS_1\\
&\qquad\overset{(i)}{=}\exp(-\bS_1^\top\bgamma^*)\left[\mu(\bS_1)-\nu^*(\bar\bS_2)-\frac{\rho(\bar\bS_2)\{\nu(\bar\bS_2)-\nu^*(\bar\bS_2)\}}{\rho^*(\bar\bS_2)}\right]\bS_1\\
&\qquad=\exp(-\bS_1^\top\bgamma^*)\left[\mu(\bS_1)-\nu(\bar\bS_2)+\left\{1-\frac{\rho(\bar\bS_2)}{\rho^*(\bar\bS_2)}\right\}\left\{\nu(\bar\bS_2)-\nu^*(\bar\bS_2)\right\}\right]\bS_1\\
&\qquad\overset{(ii)}{=}\exp(-\bS_1^\top\bgamma^*)\{\mu(\bS_1)-\nu(\bar\bS_2)\}\bS_1,
\end{align*}
where (i) holds since $Y_{1,1}\independent A_2\mid(\bar\bS_2,A_1=1)$ under Assumption \ref{cond:basic}; (ii) follows since either $\rho^*(\cdot)=\rho(\cdot)$ or $\nu^*(\cdot)=\nu(\cdot)$ holds. Hence, by the tower rule,
\begin{align*}
&\E\left[A_1\{g^{-1}(\bS_1^\top\bgamma^*)-1\}\left\{\bS_1^\top\bbeta^0-\bar\bS_2^\top\balpha^*-A_2g^{-1}(\bar\bS_2^\top\bdelta^*)(Y-\bar\bS_2^\top\balpha^*)\right\}\bS_1\right]\\
&\qquad=\E\left[\exp(-\bS_1^\top\bgamma^*)\{\mu(\bS_1)-\nu(\bar\bS_2)\}\bS_1\E(A_1\mid\bar\bS_2)\right]\\
&\qquad=\E\left[A_1\exp(-\bS_1^\top\bgamma^*)\{\mu(\bS_1)-\nu(\bar\bS_2)\}\bS_1\right]\\
&\qquad=\E\left(\exp(-\bS_1^\top\bgamma^*)[\mu(\bS_1)-\E\{\nu(\bar\bS_2)\mid\bS_1,A_1=1\}]\bS_1\E(A_1\mid\bS_1)\right)\overset{(i)}{=}\bzero\in\R^{d_1},
\end{align*}
where (i) follows from \eqref{eq:rep_mu}. By the uniqueness of $\bbeta^*$, we conclude that $\bbeta^0=\bbeta^*$ and hence $\mu(\cdot)=\mu^*(\cdot)$.
\end{enumerate}

\section{Auxiliary lemmas}\label{sec:lemmas}

The following Lemmas will be useful in the proofs.
\begin{lemma}[Lemmas D.1 and D.2 of \cite{zhang2023double}]\label{l1}
Let $(X_N)_{N\geq1}$ and $(Y_N)_{N\geq1}$ be sequences of random variables in $\mathbb R$. If $\E(\lvert X_N\rvert^r\mid Y_N)=O_p(1)$ for any $r\geq1$, then $X_N=O_p(1)$. If $\E(\lvert X_N\rvert^r\mid Y_N)=o_p(1)$ for any $r\geq1$, then $X_N=o_p(1)$.
\end{lemma}

\begin{lemma}[Lemma S.4 of \cite{bradic2024high}]\label{lemma:subG'}
Let Assumptions \ref{cond:basic} and \ref{cond:subG} hold. Then the smallest eigenvalues of $\E(A_1\bS_1\bS_1^\top)$ and $\E(A_1A_2\bar\bS_2\bar\bS_2^\top)$ are both lower bounded by some constant $c_{\min}'>0$. Additionally, with some constant $\sigma_{\bS}'>0$, $\|\bv^\top\bS_1\|_{\psi_2}\leq\sigma_{\bS}'\|\bv\|_2$, $\|A_1\bv^\top\bS_1\|_{\psi_2}\leq\sigma_{\bS}'\|\bv\|_2$, and $\|A_1A_2\bv^\top\bar\bS_2\|_{\psi_2}\leq\sigma_{\bS}'\|\bv\|_2$ for all $\bv\in\R^d$.
\end{lemma}

\begin{lemma}[Lemma D.1 (ii), (iv), (vi), and (v) of \cite{chakrabortty2019high}]\label{lemma:psi2norm}
Let $X, Y\in \R$ be random variables. If $\lvert X\rvert\leq|Y|$ a.s., then $\|X\|_{\psi_2}\leq\|Y\|_{\psi_2}$. For any $c\in\R$, $\|cX\|_{\psi_2}=|c|\|X\|_{\psi_2}$.  If $\|X\|_{\psi_2}\leq\sigma$, then $E(X)\leq \sigma\sqrt\pi$ and $E(\lvert X\rvert^m)\leq2\sigma^m(m/2)^{m/2}$ for all $m\geq2$. Moreover, $\|XY\|_{\psi_1}\leq\|X\|_{\psi_2}\|Y\|_{\psi_2}$.
Let $\{X_i\}_{i=1}^n$ be a sequence of possibly dependent random variables with $\max_{1\leq i \leq n}\|X_i\|_{\psi_2}\leq\sigma$, then $\|\max_{1\leq i \leq n}|X_i|\|_{\psi_2}\leq\sigma(\log n+2)^{1/2}$.
\end{lemma}

\begin{lemma}[Corollary 2.3 of \cite{dumbgen2010nemirovski}]\label{Nemirovski}
Let $\{X_i\}_{i=1}^n$ be identically distributed with $\E(X_i)=\bzero$, then
$$\E\left[\left\|n^{-1}\sum_{i=1}^{n} X_i\right\|_{\infty}^2\right]\leq n^{-1} (2e\log d-e) \E\left[\|X_i\|_{\infty}^2\right].$$
\end{lemma}

\begin{lemma}\label{lemma:prop2}
Let $\S'=(\bU_i)_{i\in\mathcal J}$ be independent and identically distributed (i.i.d.) sub-Gaussian random vectors, i.e., $\|\ba^\top\bU\|_{\psi_2}\leq\sigma_{\bU}\|\ba\|_2$ for all $\ba\in\R^d$. Suppose that $\lambda_{\min}(\E(\bU\bU^\top))>\lambda_{\bU}$ with some constant $\lambda_{\bU}>0$. Let $M=\lvert\mathcal J\rvert$. For any continuous function $\phi:\R\to(0,\infty)$, $v\in[0,1]$, and $\etabf\in\R^d$ satisfying $\E\{\lvert\bU^\top\etabf\rvert^c\}<C$ with some constants $c,C>0$, there exists constants $\kappa_1,\kappa_2,c_1,c_2>0$, such that
\begin{align}
&\P_{\S'}\left(M^{-1}\sum_{i\in\mathcal J}\phi(\bU_i^\top(\etabf+v\bDelta))(\bU_i^\top\bDelta)^2\geq\kappa_1\|\bDelta\|_2^2-\kappa_2\frac{\log d}{M}\|\bDelta\|_1^2,\;\;\forall\|\bDelta\|_2\leq1\right)\nonumber\\
&\qquad\geq1-c_1\exp(-c_2M).\label{bound:RSC_gen}
\end{align}
\end{lemma}

Lemma \ref{lemma:prop2} follows directly from repeating the proof of Lemma 4.3 in \cite{zhang2023double}. For a slightly different version, see Theorem 9.36 and Example 9.17 in \cite{wainwright2019high}.

Note that Lemma 4.3 of \cite{zhang2023double} provides a stronger result than the bound in \eqref{bound:RSC_gen}. Instead of the required lower bound \(\kappa_1\|\bDelta\|_2^2 - \kappa_2 (\log d/M) \|\bDelta\|_1^2\), they obtained a larger lower bound of the form \(\kappa_1' \|\bDelta\|_2^2 - \kappa_2' \sqrt{\log d / M} \|\bDelta\|_2 \|\bDelta\|_1\), with constants $\kappa_1',\kappa_2'>0$. In the following, we explain why the required lower bound in \eqref{bound:RSC_gen} is indeed weaker. By the fact that $2ab\leq a^2+b^2$, we get
\begin{align*}
&\kappa_1'\|\bDelta\|_2^2-\kappa_2'\sqrt\frac{\log d}{M}\|\bDelta\|_2\|\bDelta\|_1=\kappa_1'\|\bDelta\|_2^2-\sqrt{\kappa_1'}\|\bDelta\|_2\cdot\frac{1}{\sqrt{\kappa_1'}}\kappa_2'\sqrt\frac{\log d}{M}\|\bDelta\|_1\\
&\qquad\geq\kappa_1'\|\bDelta\|_2^2-\frac{\kappa_1'}{2}\|\bDelta\|_2^2-\frac{{\kappa_2'}^2}{2\kappa_1'}\frac{\log d}{M}\|\bDelta\|_1^2=\frac{\kappa_1'}{2}\|\bDelta\|_2^2-\frac{{\kappa_2'}^2}{2\kappa_1'}\frac{\log d}{M}\|\bDelta\|_1^2.
\end{align*}
Therefore, \eqref{bound:RSC_gen} holds with constants $\kappa_1=\kappa_1'/2$ and $\kappa_2=\kappa_2'^2/(2\kappa_1')$.

The following lemma is a direct consequence of (C.8) in \cite{chen2023enhancing}.
\begin{lemma}\label{lemma:beta0}
Let $(\bX_i)_{i=1}^m$ be i.i.d. sub-Gaussian random vectors in $\R^d$ and let $\bX$ be an independent copy of $\bX_i$. Let $S\subseteq\{1,\dots,d\}$ and $s=|S|$. Then, as $m,d\to\infty$,
$$\sup_{\bDelta\in\C(S,3)\cap\|\bDelta\|_2=1}\left\lvert m^{-1}\sum_{i=1}^m(\bX_i^\top\bDelta)^2-\E\{(\bX^\top\bDelta)^2\}\right\rvert=O_p\left(\sqrt\frac{s\log d}{m}+\frac{s\log d}{m}\right),$$
where $\C(S,3):=\{\bDelta\in\R^d:\|\bDelta_{S^c}\|_1\leq3\|\bDelta_S\|_1\}$.
\end{lemma}

Proofs of the following Lemmas \ref{lemma:beta1}-\ref{lemma:preliminary_delta} are provided in Section \ref{sec:proof_lemmas}.

\begin{lemma}\label{lemma:beta1}
Let $(\bX_i)_{i=1}^m$ be i.i.d. sub-Gaussian random vectors in $\R^d$. Then, for any (possibly random) $\bDelta\in\R^{d}$, as $m,d\to\infty$,
$$\sup_{\bDelta\in\R^{d}/\{\bzero\}}\frac{m^{-1}\sum_{i=1}^m(\bX_i^\top\bDelta)^2}{\|\bDelta\|_1^2m^{-1}\log d+\|\bDelta\|_2^2}=O_p(1).$$
\end{lemma}

\begin{lemma}\label{lemma:sol}
Suppose $a,b,c,x\in\R$, $a>0$, and $b,c>0$. Let $ax^2-bx-c\leq0$. Then
$$x\leq\frac{b}{a}+\sqrt\frac{c}{a}.$$
\end{lemma}

\begin{lemma}\label{lemma:preliminary_alpha}
Let the assumptions in part (a) of Theorem \ref{thm:nuisance} hold. Let $r>0$ be any positive constant. Then, as $N,d_1\to\infty$,
$$\left\|\bS_1^\top(\bgammahat-\bgamma^*)\right\|_{\P,r}=O\left(\|\bgammahat-\bgamma^*\|_2\right)=O_p\left(\sqrt\frac{s_{\bgamma}\log d_1}{N}\right)$$
and
\begin{align}
&\left\|\exp(-\bS_1^\top\bgammahat)-\exp(-\bS_1^\top\bgamma^*)\right\|_{\P,r}=\left\|g^{-1}(\bS_1^\top\bgammahat)-g^{-1}(\bS_1^\top\bgamma^*)\right\|_{\P,r}\nonumber\\
&\qquad=O\left(\|\bgammahat-\bgamma^*\|_2\right)=O_p\left(\sqrt\frac{s_{\bgamma}\log d_1}{N}\right).\label{rate:invps_alpha}
\end{align}
Define
\begin{align}
\mathcal E_1:=\biggr\{\|\bgammahat-\bgamma^*\|_2\leq1,\left\|g^{-1}(\bS_1^\top\bgamma)\right\|_{\P,12}\leq C,\;\forall\bgamma\in\left\{w\bgamma^*+(1-w)\bgammahat:w\in[0,1]\right\}\biggr\}.\label{def:E1}
\end{align}
Then, as $N,d_1\to\infty$,
$$\P_{\S_{\bgamma}}(\mathcal E_1)=1-o(1).$$
On the event $\mathcal E_1$, for any $r'\in[1,12]$ and $\bgamma\in\{w\bgamma^*+(1-w)\bgammahat:w\in[0,1]\}$, we also have
$$\left\|g^{-1}(\bS_1^\top\bgamma)\right\|_{\P,r'}\leq C,\quad\left\|\exp(-\bS_1^\top\bgamma)\right\|_{\P,r'}\leq C,\quad\left\|\exp(\bS_1^\top\bgamma)\right\|_{\P,r'}\leq C',$$
with some constant $C'>0$.
\end{lemma}

\begin{lemma}\label{lemma:preliminary_beta}
Let $r>0$ be any positive constant.

 (a) Let the assumptions in part (b) of Theorem \ref{thm:nuisance} hold. Then, as $N,d_1,d_2\to\infty$,
$$\left\|\bar\bS_2^\top(\bdeltahat-\bdelta^*)\right\|_{\P,r}=O\left(\|\bdeltahat-\bdelta^*\|_2\right)=O_p\left(\sqrt\frac{s_{\bgamma}\log d_1+s_{\bdelta}\log d}{N}\right)$$
and
\begin{align}
&\left\|\exp(-\bar\bS_2^\top\bdeltahat)-\exp(-\bar\bS_2^\top\bdelta^*)\right\|_{\P,r}=\left\|g^{-1}(\bar\bS_2^\top\bdeltahat)-g^{-1}(\bar\bS_2^\top\bdelta^*)\right\|_{\P,r}\nonumber\\
&\qquad=O\left(\|\bdeltahat-\bdelta^*\|_2\right)=O_p\left(\sqrt\frac{s_{\bgamma}\log d_1+s_{\bdelta}\log d}{N}\right).\label{rate:invps_beta1}
\end{align}

 (b) Let the assumptions in part (a) of Theorem \ref{thm:nuisance'} hold. Then, as $N,d_1,d_2\to\infty$,
$$\left\|\bar\bS_2^\top(\bdeltahat-\bdelta^*)\right\|_{\P,r}=O\left(\|\bdeltahat-\bdelta^*\|_2\right)=O_p\left(\sqrt\frac{s_{\bdelta}\log d}{N}\right)$$
and
\begin{align}
&\left\|\exp(-\bar\bS_2^\top\bdeltahat)-\exp(-\bar\bS_2^\top\bdelta^*)\right\|_{\P,r}=\left\|g^{-1}(\bar\bS_2^\top\bdeltahat)-g^{-1}(\bar\bS_2^\top\bdelta^*)\right\|_{\P,r}\nonumber\\
&\qquad=O\left(\|\bdeltahat-\bdelta^*\|_2\right)=O_p\left(\sqrt\frac{s_{\bdelta}\log d}{N}\right).\label{rate:invps_beta2}
\end{align}

Let either (a) or (b) holds. Let $C>0$ be some constant, define
\begin{align}
\mathcal E_2:=\biggr\{\|\bdeltahat-\bdelta^*\|_2\leq1,\;\left\|g^{-1}(\bar\bS_2^\top\bdelta)\right\|_{\P,6}\leq C,\;\forall\bdelta\in\left\{w\bdelta^*+(1-w)\bdeltahat:w\in[0,1]\right\}\biggr\}.\label{def:E2}
\end{align}
Then, as $N,d_1,d_2\to\infty$,
$$\P_{\S_{\bgamma}\cup\S_{\bdelta}}(\mathcal E_2)=1-o(1).$$
On the event $\mathcal E_2$, for any $r'\in[1,12]$ and $\bdelta\in\{w\bdelta^*+(1-w)\bdeltahat:w\in[0,1]\}$, we also have
\begin{align*}
\left\|g^{-1}(\bar\bS_2^\top\bdelta)\right\|_{\P,r'}\leq C,\quad\left\|\exp(-\bar\bS_2^\top\bdelta)\right\|_{\P,r'}\leq C,\quad\left\|\exp(\bar\bS_2^\top\bdelta)\right\|_{\P,r'}\leq C',
\end{align*}
with some constant $C'>0$.
\end{lemma}

\begin{lemma}\label{lemma:preliminary_gamma}
Let $r>0$ be any positive constant.

 (a) Let the assumptions in part (c) of Theorem \ref{thm:nuisance} hold. Then, as $N,d_1,d_2\to\infty$,
$$\left\|\bar\bS_2^\top(\balphahat-\balpha^*)\right\|_{\P,r}=O\left(\|\balphahat-\balpha^*\|_2\right)=O_p\left(\sqrt\frac{s_{\bgamma}\log d_1+s_{\bdelta}\log d+s_{\balpha}\log d}{N}\right).$$

 (b) Let the assumptions in part (b) of Theorem \ref{thm:nuisance'} hold. Then, as $N,d_1,d_2\to\infty$,
$$\left\|\bar\bS_2^\top(\balphahat-\balpha^*)\right\|_{\P,r}=O\left(\|\balphahat-\balpha^*\|_2\right)=O_p\left(\sqrt\frac{s_{\balpha}\log d}{N}\right).$$

Let either (a) or (b) holds. For any $v_1\in[0,1]$, let $\balphatil=v_1\balpha^*+(1-v_1)\balphahat$. Define $\varepsilontil:=Y_{1,1}-\bar\bS_2^\top\balphatil$. Then, for any constant $r>0$, $\|\varepsilontil\|_{\P,r}=O_p(1)$.
\end{lemma}

\begin{lemma}\label{lemma:preliminary_delta}
Let $r>0$ be any positive constant.

 (a) Let the assumptions in part (d) of Theorem \ref{thm:nuisance} hold. Then, as $N,d_1,d_2\to\infty$,
$$\left\|\bS_1^\top(\bbetahat-\bbeta^*)\right\|_{\P,r}=O\left(\|\bbetahat-\bbeta^*\|_2\right)=O_p\left(\sqrt\frac{(s_{\bgamma}+s_{\bbeta})\log d_1+(s_{\bdelta}+s_{\balpha})\log d}{N}\right).$$

 (b) Let the assumptions in part (c) or part (d) or part (e) of Theorem \ref{thm:nuisance'} hold. Then, as $N,d_1,d_2\to\infty$,
$$\left\|\bS_1^\top(\bbetahat-\bbeta^*)\right\|_{\P,r}=O\left(\|\bbetahat-\bbeta^*\|_2\right).$$

Let either (a) or (b) holds, and let either (a) or (b) of \ref{lemma:preliminary_gamma} holds. For any $v_1,v_2\in[0,1]$, let $\balphatil=v_1\balpha^*+(1-v_1)\balphahat$ and $\bbetatil=v_1\bbeta^*+(1-v_1)\bbetahat$. Define $\zetatil:=\bar\bS_2^\top\balphatil-\bS_1^\top\bbetatil$. Then, for any constant $r>0$, $\|\zetatil\|_{\P,r}=O_p(1)$.
\end{lemma}

\section{Proofs of the main results}\label{sec:proof_DTE}

\subsection{Proofs of the results in Section \ref{sec:DTE}}

\begin{proof}[Proof of Theorem \ref{thm:rate}.]
Recall the definition of the score function, \eqref{def:score}. Observe that
\begin{align*}
\bnabla_{\bgamma}\psi(\bW;\etabf)=&-A_1\exp(-\bS_1^\top\bgamma)\left\{\frac{A_2(Y-\bar\bS_2^\top\balpha)}{g(\bar\bS_2^\top\bdelta)}+\bar\bS_2^\top\balpha-\bS_1^\top\bbeta\right\}\bS_1,\\
\bnabla_{\bdelta}\psi(\bW;\etabf)=&-\frac{A_1A_2\exp(-\bar\bS_2^\top\bdelta)(Y-\bar\bS_2^\top\balpha)}{g(\bS_1^\top\bgamma)}\bar\bS_2,\\
\bnabla_{\balpha}\psi(\bW;\etabf)=&\frac{A_1}{g(\bS_1^\top\bgamma)}\left\{1-\frac{A_2}{g(\bar\bS_2^\top\bdelta)}\right\}\bar\bS_2,\\
\bnabla_{\bbeta}\psi(\bW;\etabf)=&\left\{1-\frac{A_1}{g(\bS_1^\top\bgamma)}\right\}\bS_1.
\end{align*}
By the constructions in \eqref{def:alpha}, \eqref{def:l2}, \eqref{def:l3}, and \eqref{def:l4}, we have
\begin{align}
\E\left\{\bnabla_{\bgamma}\psi(\bW;\etabf^*)\right\}&=-\E\left[A_1\exp(-\bS_1^\top\bgamma^*)\left\{\frac{A_2(Y-\bar\bS_2^\top\balpha^*)}{g(\bar\bS_2^\top\bdelta^*)}+\bar\bS_2^\top\balpha^*-\bS_1^\top\bbeta^*\right\}\bS_1\right]\nonumber\\
&=\frac{1}{2}\bnabla_{\bbeta}\E\{\ell_4(\bW;\bgamma^*,\bdelta^*,\balpha^*,\bbeta^*)\}=\bzero\in\R^{d_1},\label{eq:gradient1}\\
\E\left\{\bnabla_{\bdelta}\psi(\bW;\etabf^*)\right\}&=-\E\left[\frac{A_1A_2\exp(-\bar\bS_2^\top\bdelta^*)(Y-\bar\bS_2^\top\balpha^*)}{g(\bS_1^\top\bgamma^*)}\bar\bS_2\right]\nonumber\\
&=\frac{1}{2}\bnabla_{\balpha}\E\{\ell_3(\bW;\bgamma^*,\bdelta^*,\balpha^*)\}=\bzero\in\R^{d},\label{eq:gradient2}\\
\E\left\{\bnabla_{\balpha}\psi(\bW;\etabf^*)\right\}&=\E\left[\frac{A_1}{g(\bS_1^\top\bgamma^*)}\left\{1-\frac{A_2}{g(\bar\bS_2^\top\bdelta^*)}\right\}\bar\bS_2\right]=\bnabla_{\bdelta}\E\{\ell_2(\bW;\bgamma^*,\bdelta^*)\}\nonumber\\
&=\bzero\in\R^{d},\label{eq:gradient3}\\
\E\left\{\bnabla_{\bbeta}\psi(\bW;\etabf^*)\right\}&=\E\left[\left\{1-\frac{A_1}{g(\bS_1^\top\bgamma^*)}\right\}\bS_1\right]=\bnabla_{\bgamma}\E\{\ell_1(\bW;\bgamma^*)\}=\bzero\in\R^{d_1}.\label{eq:gradient4}
\end{align}

Note that,
\begin{align}
&\thetahat_{1,1}-\theta_{1,1}=N^{-1}\sum_{k=1}^{\K}\sum_{i\in\mathcal I_k}\psi(\bW_i;\etabfhat_{-k})-\theta_{1,1}\nonumber\\
&\quad=\K^{-1}\sum_{k=1}^{\K}n^{-1}\sum_{i\in\mathcal I_k}\left\{\psi(\bW_i;\etabfhat_{-k})-\psi(\bW_i;\etabf^*)\right\}+N^{-1}\sum_{i=1}^N\psi(\bW_i;\etabf^*)-\theta_{1,1}\nonumber\\
&\quad=N^{-1}\sum_{i=1}^N\psi(\bW_i;\etabf^*)-\theta_{1,1}+\K^{-1}\sum_{k=1}^{\K}(\Delta_{k,1}+\Delta_{k,2}),\label{rep:thetahat}
\end{align}
where
\begin{align*}
\Delta_{k,1}&=n^{-1}\sum_{i\in\mathcal I_k}\left\{\psi(\bW_i;\etabfhat_{-k})-\psi(\bW_i;\etabf^*)\right\}-\E\left\{\psi(\bW;\etabfhat_{-k})-\psi(\bW;\etabf^*)\right\},\\
\Delta_{k,2}&=\E\left\{\psi(\bW;\etabfhat_{-k})-\psi(\bW;\etabf^*)\right\}.
\end{align*}

\textbf{Step 1.} We demonstrate that
\begin{equation}\label{result:step1}
\E\{\psi(\bW;\etabf^*)\}-\theta_{1,1}=0.
\end{equation}
Here, \eqref{result:step1} can be shown under Assumption \ref{cond:mis}:
\begin{align*}
&\E\{\psi(\bW;\etabf^*)\}-\theta_{1,1}\\
&\qquad=\E\left[\left\{1-\frac{A_1}{g(\bS_1^\top\bgamma^*)}\right\}(\bS_1^\top\bbeta^*-Y_{1,1})\right]+\E\left[\frac{A_1}{g(\bS_1^\top\bgamma^*)}\left\{1-\frac{A_2}{g(\bar\bS_2^\top\bdelta^*)}\right\}(\bar\bS_2^\top\balpha^*-Y_{1,1})\right]\\
&\qquad\overset{(i)}{=}\E\left[\left\{1-\frac{\pi(\bS_1)}{g(\bS_1^\top\bgamma^*)}\right\}\{\bS_1^\top\bbeta^*-\mu(\bS_1)\}\right]\\
&\qquad\qquad+\E\left[\frac{\pi(\bS_1)}{g(\bS_1^\top\bgamma^*)}\left\{1-\frac{\rho(\bar\bS_2)}{g(\bar\bS_2^\top\bdelta^*)}\right\}\{\bar\bS_2^\top\balpha^*-\nu(\bar\bS_2)\}\right]\\
&\qquad\overset{(ii)}{=}0,
\end{align*}
where (i) holds by the tower rule; (ii) holds under Assumption \ref{cond:mis}.

\vspace{0.5em}

\textbf{Step 2.} We demonstrate that, for each $k\leq\K$, as $N,d_1,d_2\to\infty$,
\begin{align}
\Delta_{k,2}&=O_p(r_{\bgamma}r_{\bbeta}+r_{\bdelta}r_{\balpha})+\idf_{\rho\neq\rho^*}O_p(r_{\bgamma}r_{\balpha})+\idf_{\nu\neq\nu^*}O_p(r_{\bgamma}r_{\bdelta}+r_{\bdelta}^2)\nonumber\\
&\qquad+\idf_{\mu\neq\mu^*}O_p(r_{\bgamma}^2+r_{\bgamma}r_{\bdelta}+r_{\bgamma}r_{\balpha}),\label{rate:Deltak2}
\end{align}
where
\begin{align*}
r_{\bgamma}:=\sqrt\frac{s_{\bgamma}\log d_1}{N},\;\;r_{\bdelta}:=\sqrt\frac{s_{\bdelta}\log d}{N},\;\;r_{\balpha}:=\sqrt\frac{s_{\gamma}\log d}{N},\;\;r_{\bbeta}:=\sqrt\frac{s_{\bbeta}\log d_1}{N}.
\end{align*}
Note that,
$$\Delta_{k,2}=\Delta_{k,3}+\Delta_{k,4}+\Delta_{k,5}+\Delta_{k,6}+\Delta_{k,7}+\Delta_{k,8}+\Delta_{k,9},$$
where
\begin{align*}
\Delta_{k,3}&=\E\left[\frac{A_1}{g(\bS_1^\top\bgammahat_{-k})}\left\{1-\frac{g(\bar\bS_2^\top\bdelta^*)}{g(\bar\bS_2^\top\bdeltahat_{-k})}\right\}\bar\bS_2^\top(\balphahat_{-k}-\balpha^*)\right],\\
\Delta_{k,4}&=\E\left[\left\{1-\frac{g(\bS_1^\top\bgamma^*)}{g(\bS_1^\top\bgammahat_{-k})}\right\}\bS_1^\top(\bbetahat_{-k}-\bbeta^*)\right],\\
\Delta_{k,5}&=\E\left[\frac{1}{g(\bS_1^\top\bgammahat_{-k})}\left\{g(\bS_1^\top\bgamma^*)-A_1\right\}\bS_1^\top(\bbetahat_{-k}-\bbeta^*)\right],\\
\Delta_{k,6}&=\E\left[\frac{A_1}{g(\bS_1^\top\bgammahat_{-k})g(\bar\bS_2^\top\bdeltahat_{-k})}\left\{g(\bar\bS_2^\top\bdelta^*)-A_2\right\}\bar\bS_2^\top(\balphahat_{-k}-\balpha^*)\right],\\
\Delta_{k,7}&=\E\left[\frac{A_1}{g(\bS_1^\top\bgammahat_{-k})}\left\{\frac{A_2}{g(\bar\bS_2^\top\bdeltahat_{-k})}-\frac{A_2}{g(\bar\bS_2^\top\bdelta^*)}\right\}(Y_{1,1}-\bar\bS_2^\top\balpha^*)\right],\\
\Delta_{k,8}&=\E\left[\left\{\frac{A_1}{g(\bS_1^\top\bgammahat_{-k})}-\frac{A_1}{g(\bS_1^\top\bgamma^*)}\right\}(Y_{1,1}-\bS_1^\top\bbeta^*)\right],\\
\Delta_{k,9}&=\E\left[\left\{\frac{A_1}{g(\bS_1^\top\bgammahat_{-k})}-\frac{A_1}{g(\bS_1^\top\bgamma^*)}\right\}\left\{\frac{A_2}{g(\bar\bS_2^\top\bdelta^*)}-1\right\}(Y_{1,1}-\bar\bS_2^\top\balpha^*)\right].
\end{align*}
By the tower rule, we have
\begin{align*}
\Delta_{k,5}&=\E\left(\E\left[\frac{\pi^*(\bS_1)-A_1}{g(\bS_1^\top\bgammahat_{-k})}\bS_1^\top(\bbetahat_{-k}-\bbeta^*)\mid\bS_1\right]\right)\\
&=\E\left[\frac{\pi^*(\bS_1)-\pi(\bS_1)}{g(\bS_1^\top\bgammahat_{-k})}\bS_1^\top(\bbetahat_{-k}-\bbeta^*)\right]=0,\;\;\text{if $\pi=\pi^*$};\\
\Delta_{k,6}&=\E\left(\E\left[\frac{\rho^*(\bar\bS_2)-A_2}{g(\bS_1^\top\bgammahat_{-k})g(\bar\bS_2^\top\bdeltahat_{-k})}\bar\bS_2^\top(\balphahat_{-k}-\balpha^*)\mid\bar\bS_2,A_1=1\right]\E(A_1\mid\bar\bS_2)\right)\\
&=\E\left[\frac{A_1\{\rho^*(\bar\bS_2)-\rho(\bar\bS_2)\}}{g(\bS_1^\top\bgammahat_{-k})g(\bar\bS_2^\top\bdeltahat_{-k})}\bar\bS_2^\top(\balphahat_{-k}-\balpha^*)\right]=0,\;\;\text{if $\rho=\rho^*$};\\
\Delta_{k,7}&=\E\Biggr(\E\left[\frac{A_2\{Y_{1,1}-\nu^*(\bar\bS_2)\}}{g(\bS_1^\top\bgammahat_{-k})}\left\{\frac{1}{g(\bar\bS_2^\top\bdeltahat_{-k})}-\frac{1}{g(\bar\bS_2^\top\bdelta^*)}\right\}\mid\bar\bS_2,A_1=1\right]\E(A_1\mid\bar\bS_2)\Biggr)\\
&\overset{(i)}{=}\E\left[\frac{A_1\rho(\bar\bS_2)\{\nu(\bar\bS_2)-\nu^*(\bar\bS_2)\}}{g(\bS_1^\top\bgammahat_{-k})}\left\{\frac{1}{g(\bar\bS_2^\top\bdeltahat_{-k})}-\frac{1}{g(\bar\bS_2^\top\bdelta^*)}\right\}\right]=0,\;\text{if $\nu=\nu^*$};\\
\Delta_{k,8}&=\E\left(\E\left[\left\{\frac{A_1}{g(\bS_1^\top\bgammahat_{-k})}-\frac{A_1}{g(\bS_1^\top\bgamma^*)}\right\}\{Y_{1,1}-\mu^*(\bS_1)\}\mid\bS_1\right]\right)\\
&\overset{(ii)}{=}\E\left[\left\{\frac{\pi(\bS_1)}{g(\bS_1^\top\bgammahat_{-k})}-\frac{\pi(\bS_1)}{g(\bS_1^\top\bgamma^*)}\right\}\{\mu(\bS_1)-\mu^*(\bS_1)\}\right]=0,\;\text{if $\mu=\mu^*$}.
\end{align*}
Here, (i) holds since $Y_{1,1}\independent A_2\mid(\bar\bS_2,A_1=1)$ under Assumption \ref{cond:basic}; (ii) holds since $Y_{1,1}\independent A_1\mid\bS_1$ under Assumption \ref{cond:basic}. Note that, the expectation $\E(\cdot)$ corresponds to the joint distribution of the underlying random vector $\mathbb W:=(\{Y(a_1,a_2)\}_{a_1,a_2\in\{0,1\}},A_1,A_2,\bS_1,\bS_2)$, which is independent of the observed samples $(\bW_i)_{i=1}^N$ and hence also independent of the nuisance estimators $\etabfhat_{-k}$. Such an expectation will be used throughout the document unless otherwise stated.

Additionally, we also have
\begin{align*}
\Delta_{k,9}&=\E\Biggr(\E\left[\left\{\frac{Y_{1,1}-\nu^*(\bar\bS_2)}{g(\bS_1^\top\bgammahat_{-k})}-\frac{Y_{1,1}-\nu^*(\bar\bS_2)}{g(\bS_1^\top\bgamma^*)}\right\}\frac{A_2-\rho^*(\bar\bS_2)}{\rho^*(\bar\bS_2)}\mid\bar\bS_2,A_1=1\right]\E(A_1\mid\bar\bS_2)\Biggr)\\
&\overset{(i)}{=}\E\left[\left\{\frac{A_1}{g(\bS_1^\top\bgammahat_{-k})}-\frac{A_1}{g(\bS_1^\top\bgamma^*)}\right\}\frac{\rho(\bar\bS_2)-\rho^*(\bar\bS_2)}{\rho^*(\bar\bS_2)}\left\{\nu(\bar\bS_2)-\nu^*(\bar\bS_2)\right\}\right]\overset{(ii)}{=}0,
\end{align*}
where (i) holds since $Y_{1,1}\independent A_2\mid(\bar\bS_2,A_1=1)$ under Assumption \ref{cond:basic}; (ii) holds since either $\rho(\cdot)=\rho^*(\cdot)$ or $\nu(\cdot)=\nu^*(\cdot)$ under Assumption \ref{cond:mis}. Therefore,
$$\Delta_{k,2}=\Delta_{k,3}+\Delta_{k,4}+\Delta_{k,5}\idf_{\pi\neq\pi^*}+\Delta_{k,6}\idf_{\rho\neq\rho^*}+\Delta_{k,7}\idf_{\nu\neq\nu^*}+\Delta_{k,8}\idf_{\mu\neq\mu^*}.$$
Now, condition on the event $\mathcal E_1\cap\mathcal E_2$, where $\mathcal E_1$ and $\mathcal E_2$ are defined as \eqref{def:E1} and \eqref{def:E2}, respectively. By Lemmas \ref{lemma:preliminary_alpha} and \ref{lemma:preliminary_beta}, $\mathcal E_1\cap\mathcal E_2$ occurs with probability $1-o(1)$. By H\"older's inequality,
\begin{align*}
\lvert\Delta_{k,3}\rvert&\leq\left\|g^{-1}(\bS_1^\top\bgammahat_{-k})\right\|_{\P,4}\left\|g^{-1}(\bar\bS_2^\top\bdeltahat_{-k})\right\|_{\P,4}\left\|g^{-1}(\bar\bS_2^\top\bdeltahat_{-k})-g^{-1}(\bar\bS_2^\top\bdelta^*)\right\|_{\P,4}\\
&\qquad\cdot\left\|\bar\bS_2^\top(\balphahat_{-k}-\balpha^*)\right\|_{\P,4}\\
&\overset{(i)}{=}O_p\left(\|\bdeltahat_{-k}-\bdelta^*\|_2\|\balphahat_{-k}-\balpha^*\|_2\right),
\end{align*}
where (i) follows from Lemmas \ref{lemma:preliminary_alpha}, \ref{lemma:preliminary_beta}, and \ref{lemma:preliminary_gamma}. Similarly,
\begin{align*}
\lvert\Delta_{k,4}\rvert&\leq\left\|g^{-1}(\bS_1^\top\bgammahat_{-k})\right\|_{\P,4}\left\|g^{-1}(\bS_1^\top\bgammahat_{-k})-g^{-1}(\bS_1^\top\bgamma^*)\right\|_{\P,4}\left\|\bS_1^\top(\bbetahat_{-k}-\bbeta^*)\right\|_{\P,2}\\
&\overset{(i)}{=}O_p\left(\|\bgammahat_{-k}-\bgamma^*\|_2\|\bbetahat_{-k}-\bbeta^*\|_2\right),
\end{align*}
where (i) follows from Lemmas \ref{lemma:preliminary_alpha} and \ref{lemma:preliminary_delta}. In addition,
\begin{align*}
\lvert\Delta_{k,5}\rvert&\overset{(i)}{=}\left\lvert\E\left[\left\{\frac{1}{g(\bS_1^\top\bgammahat_{-k})}-\frac{1}{g(\bS_1^\top\bgamma^*)}\right\}\left\{g(\bS_1^\top\bgamma^*)-A_1\right\}\bS_1^\top(\bbetahat_{-k}-\bbeta^*)\right]\right\rvert\\
&\overset{(ii)}{=}\left\|g^{-1}(\bS_1^\top\bgammahat_{-k})-g^{-1}(\bS_1^\top\bgamma^*)\right\|_{\P,2}\left\|\bS_1^\top(\bbetahat_{-k}-\bbeta^*)\right\|_{\P,2}\\
&\overset{(iii)}{=}O_p\left(\|\bgammahat_{-k}-\bgamma^*\|_2\|\bbetahat_{-k}-\bbeta^*\|_2\right),
\end{align*}
where (i) follows from \eqref{eq:gradient4}; (ii) holds by H\"older's inequality and the fact that $\lvert g(\bS_1^\top\bgamma^*)-A_1\rvert\leq1$; (iii) follows from Lemmas \ref{lemma:preliminary_alpha} and \ref{lemma:preliminary_delta}. Besides,
\begin{align*}
&\lvert\Delta_{k,6}\rvert\overset{(i)}{=}\left\lvert\E\left(\left[\frac{A_1\left\{g(\bar\bS_2^\top\bdelta^*)-A_2\right\}}{g(\bS_1^\top\bgammahat_{-k})g(\bar\bS_2^\top\bdeltahat_{-k})}-\frac{A_1\left\{g(\bar\bS_2^\top\bdelta^*)-A_2\right\}}{g(\bS_1^\top\bgamma^*)g(\bar\bS_2^\top\bdelta^*)}\right]\bar\bS_2^\top(\balphahat_{-k}-\balpha^*)\right)\right\rvert\\
&\;\;\overset{(ii)}{\leq}\left\|g^{-1}(\bS_1^\top\bgammahat_{-k})g^{-1}(\bar\bS_2^\top\bdeltahat_{-k})-g^{-1}(\bS_1^\top\bgamma^*)g^{-1}(\bar\bS_2^\top\bdelta^*)\right\|_{\P,2}\left\|\bar\bS_2^\top(\balphahat_{-k}-\balpha^*)\right\|_{\P,2}\\
&\;\;\overset{(iii)}{=}O_p\left(\left(\|\bgammahat_{-k}-\bgamma^*\|_2+\|\bdeltahat_{-k}-\bdelta^*\|_2\right)\|\balphahat_{-k}-\balpha^*\|_2\right),
\end{align*}
where (i) follows from \eqref{eq:gradient3}; (ii) holds by H\"older's inequality and the fact that $\lvert A_1\{g(\bar\bS_2^\top\bdelta^*)-A_2\}\rvert\leq1$; (iii) follows from Lemma \ref{lemma:preliminary_gamma} and the fact that
\begin{align}
&\left\|g^{-1}(\bS_1^\top\bgammahat_{-k})g^{-1}(\bar\bS_2^\top\bdeltahat_{-k})-g^{-1}(\bS_1^\top\bgamma^*)g^{-1}(\bar\bS_2^\top\bdelta^*)\right\|_{\P,2}=O_p\left(\|\bgammahat_{-k}-\bgamma^*\|_2+\|\bdeltahat_{-k}-\bdelta^*\|_2\right).\label{eq:bound_diff_prod}
\end{align}
We verify \eqref{eq:bound_diff_prod} below:
\begin{align*}
&\left\|g^{-1}(\bS_1^\top\bgammahat_{-k})g^{-1}(\bar\bS_2^\top\bdeltahat_{-k})-g^{-1}(\bS_1^\top\bgamma^*)g^{-1}(\bar\bS_2^\top\bdelta^*)\right\|_{\P,2}\\
&\qquad\overset{(i)}{\leq}\left\|g^{-1}(\bS_1^\top\bgammahat_{-k})\left\{g^{-1}(\bar\bS_2^\top\bdeltahat_{-k})-g^{-1}(\bar\bS_2^\top\bdelta^*)\right\}\right\|_{\P,2}\\
&\qquad\qquad+\left\|g^{-1}(\bar\bS_2^\top\bdelta^*)\left\{g^{-1}(\bS_1^\top\bgammahat_{-k})-g^{-1}(\bS_1^\top\bgamma^*)\right\}\right\|_{\P,2}\\
&\qquad\overset{(ii)}{\leq}\left\|g^{-1}(\bS_1^\top\bgammahat_{-k})\right\|_{\P,4}\left\|g^{-1}(\bar\bS_2^\top\bdeltahat_{-k})-g^{-1}(\bar\bS_2^\top\bdelta^*)\right\|_{\P,4}\\
&\qquad\qquad+\left\|g^{-1}(\bar\bS_2^\top\bdelta^*)\right\|_{\P,4}\left\|g^{-1}(\bS_1^\top\bgammahat_{-k})-g^{-1}(\bS_1^\top\bgamma^*)\right\|_{\P,4}\\
&\qquad=O_p\left(\|\bgammahat_{-k}-\bgamma^*\|_2+\|\bdeltahat_{-k}-\bdelta^*\|_2\right),
\end{align*}
where (i) holds by Minkowski inequality; (ii) holds by (generalized) H\"older's inequality; (iii) follows from Lemmas \ref{lemma:preliminary_alpha} and \ref{lemma:preliminary_beta}. As for the term $\Delta_{k,7}$, with some $\bgammatil_1$ lies between $\bgamma^*$ and $\bgammahat_{-k}$, some $\bdeltatil$ lies between $\bdelta^*$ and $\bdeltahat_{-k}$, we have
\begin{align*}
&\lvert\Delta_{k,7}\rvert=\left\lvert\E\left[\frac{A_1}{g(\bS_1^\top\bgammahat_{-k})}\left\{\frac{A_2}{g(\bar\bS_2^\top\bdeltahat_{-k})}-\frac{A_2}{g(\bar\bS_2^\top\bdelta^*)}\right\}\varepsilon\right]\right\rvert\\
&\quad\overset{(i)}{\leq}\left\lvert\E\left\{\frac{A_1A_2}{g(\bS_1^\top\bgamma^*)}\exp(-\bar\bS_2^\top\bdelta^*)\varepsilon\bar\bS_2^\top\right\}(\bdeltahat_{-k}-\bdelta^*)\right\rvert\\
&\quad\quad+\left\lvert\E\left\{A_1A_2\exp(-\bS_1^\top\bgammatil_1)\exp(-\bar\bS_2^\top\bdeltatil)\varepsilon\bar\bS_2^\top(\bdeltahat_{-k}-\bdelta^*)\bS_1^\top(\bgammahat_{-k}-\bgamma^*)\right\}\right\rvert\\
&\quad\quad+\left\lvert\E\left[\frac{A_1A_2}{g(\bS_1^\top\bgamma^*)}\exp(-\bar\bS_2^\top\bdeltatil)\varepsilon\left\{\bar\bS_2^\top(\bdeltahat_{-k}-\bdelta^*)\right\}^2\right]\right\rvert\\
&\quad\overset{(ii)}{\leq}0+\left\|\exp(-\bS_1^\top\bgammatil_1)\right\|_{\P,4}\left\|\exp(-\bar\bS_2^\top\bdeltatil)\right\|_{\P,4}\left\|\bar\bS_2^\top(\bdeltahat_{-k}-\bdelta^*)\right\|_{\P,8}\left\|\bS_1^\top(\bgammahat_{-k}-\bgamma^*)\right\|_{\P,8}\left\|\varepsilon\right\|_{\P,4}\\
&\quad\quad+\left\|g^{-1}(\bS_1^\top\bgamma^*)\right\|_{\P,4}\left\|\exp(-\bar\bS_2^\top\bdeltatil)\right\|_{\P,4}\left\|\varepsilon\right\|_{\P,4}\left\|\bar\bS_2^\top(\bdeltahat_{-k}-\bdelta^*)\right\|_{\P,8}^2\\
&\quad\overset{(iii)}{=}O_p\left(\|\bgammahat_{-k}-\bgamma^*\|_2\|\bdeltahat_{-k}-\bdelta^*\|_2+\|\bdeltahat_{-k}-\bdelta^*\|_2^2\right),
\end{align*}
where (i) holds by Taylor's theorem; (ii) holds by \eqref{eq:gradient2} and H{\"o}lder's inequality; (iii) holds by Lemmas \ref{lemma:preliminary_alpha} and \ref{lemma:preliminary_beta}. Similarly, by Taylor's theorem, with some $\bgammatil_2$ lies between $\bgamma^*$ and $\bgammahat_{-k}$, we have
\begin{align*}
\lvert\Delta_{k,8}\rvert&\leq\left\lvert\E\left\{A_1\exp(-\bS_1^\top\bgamma^*)(Y_{1,1}-\bS_1^\top\bbeta^*)\bS_1^\top(\bgammahat_{-k}-\bgamma^*)\right\}\right\rvert\\
&\qquad+\left\lvert\E\left\{A_1\exp(-\bS_1^\top\bgammatil_2)(Y_{1,1}-\bS_1^\top\bbeta^*)\left\{\bS_1^\top(\bgammahat_{-k}-\bgamma^*)\right\}^2\right\}\right\rvert\\
&\overset{(i)}{=}\left\lvert\E\left[A_1\exp(-\bS_1^\top\bgamma^*)\left\{\frac{A_2(Y-\bar\bS_2^\top\balpha^*)}{g(\bar\bS_2^\top\bdelta^*)}+\bar\bS_2^\top\balpha^*-\bS_1^\top\bbeta^*\right\}\bS_1^\top(\bgammahat_{-k}-\bgamma^*)\right]\right\rvert\\
&\quad+\left\lvert\E\left[A_1\exp(-\bS_1^\top\bgammatil_2)(\varepsilon+\zeta)\left\{\bS_1^\top(\bgammahat_{-k}-\bgamma^*)\right\}^2\right]\right\rvert\\
&\overset{(ii)}{\leq}0+\left\|\exp(-\bS_1^\top\bgammatil_2)\right\|_{\P,4}\left\|\varepsilon+\zeta\right\|_{\P,4}\left\|\bS_1^\top(\bgammahat_{-k}-\bgamma^*)\right\|_{\P,4}^2\\
&\overset{(iii)}{=}O_p\left(\|\bgammahat_{-k}-\bgamma^*\|_2^2\right),
\end{align*}
where (i) holds since either $\rho(\cdot)=\rho^*(\cdot)$ or $\nu(\cdot)=\nu^*(\cdot)$; (ii) holds by \eqref{eq:gradient1} and H{\"o}lder's inequality; (iii) holds by Lemma \ref{lemma:preliminary_alpha}. To sum up, we have
\begin{align}
\Delta_{k,2}&=O_p\left(\|\bgammahat_{-k}-\bgamma^*\|_2\|\bbetahat_{-k}-\bbeta^*\|_2+\|\bdeltahat_{-k}-\bdelta^*\|_2\|\balphahat_{-k}-\balpha^*\|_2\right)\nonumber\\
&\qquad+\idf_{\rho\neq\rho^*}O_p\left(\|\bgammahat_{-k}-\bgamma^*\|_2\|\balphahat_{-k}-\balpha^*\|_2\right)\nonumber\\
&\qquad+\idf_{\nu\neq\nu^*}O_p\left(\|\bgammahat_{-k}-\bgamma^*\|_2\|\bdeltahat_{-k}-\bdelta^*\|_2+\|\bdeltahat_{-k}-\bdelta^*\|_2^2\right)\nonumber\\
&\qquad+\idf_{\mu\neq\mu^*}O_p\left(\|\bgammahat_{-k}-\bgamma^*\|_2^2\right).\label{rate:bias}
\end{align}
By Theorems \ref{thm:nuisance} and \ref{thm:nuisance'},
\begin{align*}
\|\bgammahat_{-k}-\bgamma^*\|_2&=O_p\left(r_{\bgamma}\right),\\
\|\bdeltahat_{-k}-\bdelta^*\|_2&=\|\bdeltahat_{-k}-\bdelta^*\|_2\idf_{\rho=\rho^*}+\|\bdeltahat_{-k}-\bdelta^*\|_2\idf_{\rho\neq\rho^*}=O_p\left(r_{\bdelta}+r_{\bgamma}\idf_{\rho\neq\rho^*}\right),\\
\|\balphahat_{-k}-\balpha^*\|_2&=\|\balphahat_{-k}-\balpha^*\|_2\idf_{\nu=\nu^*}+\|\balphahat_{-k}-\balpha^*\|_2\idf_{\nu\neq\nu^*}\\
&=O_p\left(r_{\balpha}+(r_{\bgamma}+r_{\bdelta})\idf_{\nu\neq\nu^*}\right).
\end{align*}
Additionally, note that either $\rho(\cdot)=\rho^*(\cdot)$ or $\nu(\cdot)=\nu^*(\cdot)$ (or both) holds. By Theorems \ref{thm:nuisance} and \ref{thm:nuisance'}, we have
\begin{align*}
\|\bbetahat_{-k}-\bbeta^*\|_2&=\|\bbetahat_{-k}-\bbeta^*\|_2\idf_{\rho=\rho^*,\nu=\nu^*,\mu=\mu^*}+\|\bbetahat_{-k}-\bbeta^*\|_2\idf_{\rho=\rho^*,\nu\neq\nu^*,\mu=\mu^*}\\
&\qquad+\|\bbetahat_{-k}-\bbeta^*\|_2\idf_{\rho\neq\rho^*,\nu=\nu^*,\mu=\mu^*}+\|\bbetahat_{-k}-\bbeta^*\|_2\idf_{\mu\neq\mu^*}\\
&=O_p\left(r_{\bbeta}+r_{\bdelta}\idf_{\nu\neq\nu^*}+r_{\balpha}\idf_{\rho\neq\rho^*}+(r_{\bgamma}+r_{\bdelta}+r_{\balpha})\idf_{\mu\neq\mu^*}\right).
\end{align*}
Therefore,
\begin{align*}
\Delta_{k,2}&=O_p\left(r_{\bgamma}\left\{r_{\bbeta}+r_{\bdelta}\idf_{\nu\neq\nu^*}+r_{\balpha}\idf_{\rho\neq\rho^*}+(r_{\bgamma}+r_{\bdelta}+r_{\balpha})\idf_{\mu\neq\mu^*}\right\}\right)\\
&\qquad+O_p\left(\left(r_{\bdelta}+r_{\bgamma}\idf_{\rho\neq\rho^*}\right)\left(r_{\balpha}+(r_{\bgamma}+r_{\bdelta})\idf_{\nu\neq\nu^*}\right)\right)\\
&\qquad+\idf_{\rho\neq\rho^*}O_p\left(r_{\bgamma}\left\{r_{\balpha}+(r_{\bgamma}+r_{\bdelta})\idf_{\nu\neq\nu^*}\right\}\right)\\
&\qquad+\idf_{\nu\neq\nu^*}O_p\left(\left(r_{\bgamma}+r_{\bdelta}+r_{\bgamma}\idf_{\rho\neq\rho^*}\right)\left(r_{\bdelta}+r_{\bgamma}\idf_{\rho\neq\rho^*}\right)\right)+\idf_{\mu\neq\mu^*}O_p\left(r_{\bgamma}^2\right)\\
&\overset{(i)}=O_p\left(r_{\bgamma}r_{\bbeta}+r_{\bgamma}r_{\bdelta}\idf_{\nu\neq\nu^*}+r_{\bgamma}r_{\balpha}\idf_{\rho\neq\rho^*}+r_{\bgamma}(r_{\bgamma}+r_{\bdelta}+r_{\balpha})\idf_{\mu\neq\mu^*}\right)\\
&\qquad+O_p\left(r_{\bdelta}r_{\balpha}+r_{\bgamma}r_{\balpha}\idf_{\rho\neq\rho^*}+r_{\bdelta}(r_{\bgamma}+r_{\bdelta})\idf_{\nu\neq\nu^*}\right)\\
&\qquad+\idf_{\rho\neq\rho^*}O_p(r_{\bgamma}r_{\balpha})+\idf_{\nu\neq\nu^*}O_p\left(\left(r_{\bgamma}+r_{\bdelta}\right)r_{\bdelta}\right)+\idf_{\mu\neq\mu^*}O_p\left(r_{\bgamma}^2\right)\\
&=O_p(r_{\bgamma}r_{\bbeta}+r_{\bdelta}r_{\balpha})+\idf_{\rho\neq\rho^*}O_p(r_{\bgamma}r_{\balpha})+\idf_{\nu\neq\nu^*}O_p(r_{\bgamma}r_{\bdelta}+r_{\bdelta}^2)\\
&\qquad+\idf_{\mu\neq\mu^*}O_p(r_{\bgamma}^2+r_{\bgamma}r_{\bdelta}+r_{\bgamma}r_{\balpha}),
\end{align*}
where (i) holds since $\idf_{\rho\neq\rho^*}\idf_{\nu\neq\nu^*}=0$ that either $\rho(\cdot)=\rho^*(\cdot)$ or $\nu(\cdot)=\nu^*(\cdot)$ holds.

\vspace{0.5em}

\textbf{Step 3.} We demonstrate that, for each $k\leq\K$, as $N,d_1,d_2\to\infty$,
\begin{equation}\label{rate:Deltak1}
\Delta_{k,1}=o_p(N^{-1/2}).
\end{equation}
By construction, we have $\E_{\S_k}(\Delta_{k,1})=0$, where $\E_{\S_k}(\cdot)$ denotes the expectation corresponding to the joint distribution of the sub-sample $\S_k$. In addition, by Taylor's theorem, with some $\etabftil=(\bgammatil^\top,\bdeltatil^\top,\balphatil^\top,\bbetatil^\top)^\top$ lies between $\etabf^*$ and $\etabfhat_{-k}$,
\begin{align*}
&\E_{\S_k}(\Delta_{k,1}^2)=n^{-1}\E\left[\left\{\psi(\bW;\etabfhat_{-k})-\psi(\bW;\etabf^*)\right\}^2\right]\\
&\qquad=2n^{-1}\E\left[\left\{\psi(\bW;\etabftil)-\psi(\bW;\etabf^*)\right\}\bnabla_{\etabf}\psi(\bW;\etabftil)^\top(\etabfhat_{-k}-\etabf^*)\right]\\
&\qquad\overset{(i)}{\leq}2n^{-1}\left\{\left\|\psi(\bW;\etabftil)-\bS_1^\top\bbeta^*\right\|_{\P,2}+\left\|\psi(\bW;\etabf^*)-\bS_1^\top\bbeta^*\right\|_{\P,2}\right\}\\
&\qquad\qquad\cdot\left\|\bnabla_{\etabf}\psi(\bW;\etabftil)^\top(\etabfhat_{-k}-\etabf^*)\right\|_{\P,2},
\end{align*}
where (i) holds by H\"older's inequality and Minkowski inequality. Note that
\begin{align*}
\left\|\psi(\bW;\etabf^*)-\bS_1^\top\bbeta^*\right\|_{\P,2}&\leq c^{-1}\left\|\zeta\right\|_{\P,2}+c^{-2}\left\|\varepsilon\right\|_{\P,2}=O(1).
\end{align*}
Define
\begin{align*}
\varepsilontil:=Y_{1,1}-\bar\bS_2^\top\balphatil,\quad\zetatil:=\bar\bS_2^\top\balphatil-\bS_1^\top\bbetatil.
\end{align*}
Condition on the event $\mathcal E_1\cap\mathcal E_2$. By H\"older's inequality and Lemmas \ref{lemma:preliminary_gamma} and \ref{lemma:preliminary_delta}, we also have
\begin{align*}
&\left\|\psi(\bW;\etabftil)-\bS_1^\top\bbeta^*\right\|_{\P,2}\\
&\qquad\leq\left\|g^{-1}(\bS_1^\top\bgammatil)\right\|_{\P,4}\left\|\zetatil\right\|_{\P,4}+\left\|g^{-1}(\bS_1^\top\bgammatil)\right\|_{\P,6}\left\|g^{-1}(\bar\bS_2^\top\bdeltatil)\right\|_{\P,6}\left\|\varepsilontil\right\|_{\P,6}+\left\|\bS_1^\top(\bbetatil-\bbeta^*)\right\|_{\P,2}\\
&\qquad=O_p\left(1+\|\bbetatil-\bbeta^*\|_2\right)=O_p(1).
\end{align*}
In addition, by Minkowski inequality,
\begin{align*}
&\left\|\bnabla_{\etabf}\psi(\bW;\etabftil)^\top(\etabfhat_{-k}-\etabf^*)\right\|_{\P,2}\\
&\quad\leq\left\|\bnabla_{\bgamma}\psi(\bW;\etabftil)^\top(\bgammahat_{-k}-\bgamma^*)\right\|_{\P,2}+\left\|\bnabla_{\bdelta}\psi(\bW;\etabftil)^\top(\bdeltahat_{-k}-\bdelta^*)\right\|_{\P,2}\\
&\quad\quad+\left\|\bnabla_{\balpha}\psi(\bW;\etabftil)^\top(\balphahat_{-k}-\balpha^*)\right\|_{\P,2}+\left\|\bnabla_{\bbeta}\psi(\bW;\etabftil)^\top(\bbetahat_{-k}-\bbeta^*)\right\|_{\P,2}.
\end{align*}
By H\"older's inequality, Minkowski inequality, and Lemma \ref{lemma:preliminary_alpha},
\begin{align*}
&\left\|\bnabla_{\bgamma}\psi(\bW;\etabftil)^\top(\bgammahat_{-k}-\bgamma^*)\right\|_{\P,2}\\
&\quad\leq\left\|\exp(-\bS_1^\top\bgammatil)\right\|_{\P,6}\left\{\left\|g^{-1}(\bar\bS_2^\top\bdeltatil)\right\|_{\P,6}\left\|\varepsilonhat\right\|_{\P,6}+\left\|\zetahat\right\|_{\P,3}\right\}\left\|\bS_1^\top(\bgammahat_{-k}-\bgamma^*)\right\|_{\P,6}\\
&\quad=O_p\left(\|\bgammahat_{-k}-\bgamma^*\|_2\right).
\end{align*}
Similarly, for the second term, using Lemmas \ref{lemma:preliminary_beta} and \ref{lemma:preliminary_gamma},
\begin{align*}
&\left\|\bnabla_{\bdelta}\psi(\bW;\etabftil)^\top(\bdeltahat_{-k}-\bdelta^*)\right\|_{\P,2}\\
&\quad\leq\left\|g^{-1}(\bS_1^\top\bgammatil)\right\|_{\P,6}\left\|\exp(-\bar\bS_2^\top\bdeltatil)\right\|_{\P,6}\left\|\varepsilontil\right\|_{\P,12}\left\|\bar\bS_2^\top(\bdeltahat_{-k}-\bdelta^*)\right\|_{\P,12}\\
&\quad=O_p\left(\|\bdeltahat_{-k}-\bdelta^*\|_2\right).
\end{align*}
For the third term, using Lemma \ref{lemma:preliminary_gamma},
\begin{align*}
&\left\|\bnabla_{\balpha}\psi(\bW;\etabftil)^\top(\balphahat_{-k}-\balpha^*)\right\|_{\P,2}\\
&\qquad\leq\left\|g^{-1}(\bS_1^\top\bgammatil)\right\|_{\P,6}\left\{1+\left\|g^{-1}(\bar\bS_2^\top\bdeltatil)\right\|_{\P,6}\right\}\left\|\bar\bS_2^\top(\balphahat_{-k}-\balpha^*)\right\|_{\P,6}\\
&\qquad=O_p\left(\|\balphahat_{-k}-\balpha^*\|_2\right).
\end{align*}
Lastly, using Lemma \ref{lemma:preliminary_delta},
\begin{align*}
&\left\|\bnabla_{\bbeta}\psi(\bW;\etabftil)^\top(\bbetahat_{-k}-\bbeta^*)\right\|_{\P,2}\\
&\quad\leq\left\{1+\left\|g^{-1}(\bS_1^\top\bgammatil)\right\|_{\P,4}\right\}\left\|\bS_1^\top(\bbetahat_{-k}-\bbeta^*)\right\|_{\P,4}=O_p\left(\|\bbetahat_{-k}-\bbeta^*\|_2\right).
\end{align*}
To sum up, we have
\begin{align*}
&\left\|\bnabla_{\etabf}\psi(\bW;\etabftil)^\top(\etabfhat_{-k}-\etabf^*)\right\|_{\P,2}\\
&\qquad=O_p\left(\|\bgammahat_{-k}-\bgamma^*\|_2+\|\bdeltahat_{-k}-\bdelta^*\|_2+\|\balphahat_{-k}-\balpha^*\|_2+\|\bbetahat_{-k}-\bbeta^*\|_2\right).
\end{align*}
It follows that
\begin{align}
&\E_{\S_k}(\Delta_{k,1}^2)=n^{-1}\E\left[\left\{\psi(\bW;\etabfhat_{-k})-\psi(\bW;\etabf^*)\right\}^2\right]\nonumber\\
&\qquad=N^{-1}O_p\left(\|\bgammahat_{-k}-\bgamma^*\|_2+\|\bdeltahat_{-k}-\bdelta^*\|_2+\|\balphahat_{-k}-\balpha^*\|_2+\|\bbetahat_{-k}-\bbeta^*\|_2\right)\nonumber\\
&\qquad=o_p(N^{-1}).\label{rate:moment2}
\end{align}
By Lemma \ref{l1},
$$\Delta_{k,1}=o_p(N^{-1/2}).$$

\vspace{0.5em}
\textbf{Step 4.} We show that, as $N,d_1,d_2\to\infty$,
$$N^{-1}\sum_{i=1}^N\psi(\bW_i;\etabf^*)-\theta_{1,1}=O_p(\sigma N^{-1/2})\;\;\mbox{and}\;\;\sigma\asymp\|\bbeta^*\|_2+1.$$
Note that
\begin{align*}
\E\left\{N^{-1}\sum_{i=1}^N\psi(\bW_i;\etabf^*)-\theta_{1,1}\right\}^2=\sigma^2/N.
\end{align*}
By Markov's inequality,
$$N^{-1}\sum_{i=1}^N\psi(\bW_i;\etabf^*)-\theta_{1,1}=O_p(\sigma N^{-1/2}).$$
For any constant $r>0$, by Minkowski inequality,
\begin{align}
&\left\|\psi(\bW;\etabf^*)-\theta_{1,1}\right\|_{\P,r}\nonumber\\
&\qquad\leq\left\|Y_{1,1}-\theta_{1,1}\right\|_{\P,r}+\left\|\left\{1-\frac{A_1}{g(\bS_1^\top\bgamma^*)}\right\}(\bS_1^\top\bbeta^*-Y_{1,1})\right\|_{\P,r}\nonumber\\
&\qquad\qquad+\left\|\frac{A_1}{g(\bS_1^\top\bgamma^*)}\left\{1-\frac{A_2}{g(\bar\bS_2^\top\bdelta^*)}\right\}(\bar\bS_2^\top\balpha^*-Y_{1,1})\right\|_{\P,r}\nonumber\\
&\qquad\leq\|\bS_1^\top\bbeta^*\|_{\P,r}+\|\varepsilon\|_{\P,r}+\|\zeta\|_{\P,r}+\lvert\theta_{1,1}\rvert+(1+c_0^{-1})\left\|\varepsilon+\zeta\right\|_{\P,r}+c_0^{-1}(1+c_0^{-1})\left\|\varepsilon\right\|_{\P,r}\nonumber\\
&\qquad\overset{(i)}{=}O\left(\|\bbeta^*\|_2+1\right),\label{bound:sigma-upper}
\end{align}
where (i) holds under Assumption \ref{cond:subG} and by $\lvert\theta_{1,1}\rvert=\lvert\E(\bS_1^\top\bbeta^*+\zeta+\varepsilon)\rvert\leq\|\bS_1^\top\bbeta^*\|_{\P,1}+\|\zeta\|_{\P,1}+\|\varepsilon\|_{\P,1}=O(\|\bbeta^*\|_2+1)$. It follows that
\begin{equation}\label{bound-upper-sigma}
\sigma=\left\|\psi(\bW;\etabf^*)-\theta_{1,1}\right\|_{\P,2}=O\left(\|\bbeta^*\|_2+1\right).
\end{equation}

In addition, we have the following representation
$$
\psi(\bW;\etabf^*)-\theta_{1,1}=r_0+r_1+r_2+r_3+r_4,
$$
where
\begin{align*}
r_0&:=\mu(\bS_1)-\theta_{1,1},\\
r_1&:=\left\{1-\frac{A_1}{\pi^*(\bS_1)}\right\}\left\{\mu^*(\bS_1)-\mu(\bS_1)\right\},\\
r_2&:=\frac{A_1\{\nu(\bar\bS_2)-\mu(\bS_1)\}}{\pi^*(\bS_1)},\\
r_3&:=\frac{A_1A_2\{Y_{1,1}-\nu(\bar\bS_2)\}}{\pi^*(\bS_1)\rho^*(\bar\bS_2)},\\
r_4&:=\frac{A_1}{\pi^*(\bS_1)}\left\{1-\frac{A_2}{\rho^*(\bar\bS_2)}\right\}\left\{\nu^*(\bar\bS_2)-\nu(\bar\bS_2)\right\}.
\end{align*}
Note that
\begin{align*}
\E(r_1\mid\bS_1)&=\left\{1-\frac{\pi(\bS_1)}{\pi^*(\bS_1)}\right\}\left\{\mu^*(\bS_1)-\mu(\bS_1)\right\}\overset{(i)}{=}0,\\
\E(r_2\mid\bS_1)&=\E\left[\frac{\nu(\bar\bS_2)-\mu(\bS_1)}{\pi^*(\bS_1)}\mid\bS_1,A_1=1\right]\pi(\bS_1)\overset{(ii)}{=}0,
\end{align*}
where (i) holds since either $\pi^*(\cdot)$ or $\mu^*(\cdot)$ is correctly specified; (ii) holds since $\E\{\nu(\bar\bS_2)\mid\bS_1,A_1=1\}=\mu(\bS_1)$. Additionally,
\begin{align*}
\E(r_3\mid\bar\bS_2,A_1=1)&\overset{(i)}{=}\frac{A_1\rho(\bar\bS_2)\{\nu(\bar\bS_2)-\nu(\bar\bS_2)\}}{\pi^*(\bS_1)\rho^*(\bar\bS_2)}=0,\\
\E(r_4\mid\bar\bS_2,A_1=1)&=\frac{A_1}{\pi^*(\bS_1)}\left\{1-\frac{\rho(\bar\bS_2)}{\rho^*(\bar\bS_2)}\right\}\left\{\nu^*(\bar\bS_2)-\nu(\bar\bS_2)\right\}\overset{(ii)}{=}0,
\end{align*}
where (i) holds since $A_2\independent Y_{1,1}\mid(\bar\bS_2,A_1=1)$; (ii) holds since either $\rho^*(\cdot)$ or $\nu^*(\cdot)$ is correctly specified. Moreover, by construction, we also have $\E(r_j\mid\bar\bS_2,A_1=0)=0$ for each $j\in\{3,4\}$. By the tower rule, 
\begin{align*}
\E(r_j\mid\bS_1)&=\E\left\{\E(r_j\mid\bar\bS_2,A_1=1)P(A_1=1\mid\bS_2)\mid\bS_1\right\}\\
&\qquad+\E\left\{(r_j\mid\bar\bS_2,A_1=0)P(A_1=0\mid\bS_2)\mid\bS_1\right\}\\
&=0.
\end{align*}
Hence, for each $j\in\{1,2,3,4\}$,
\begin{equation}\label{prod-0j}
\E(r_0r_j)=\E\{r_0\E(r_j\mid\bS_1)\}=0.
\end{equation}
Additionally, we also note that
\begin{align*}
&\E(r_1r_3\mid\bar\bS_2,A_1=1)=\left\{1-\frac{1}{\pi^*(\bS_1)}\right\}\left\{\mu^*(\bS_1)-\mu(\bS_1)\right\}\E(r_3\mid\bar\bS_2,A_1=1)=0.
\end{align*}
Similarly,
$$\E(r_2r_3\mid\bar\bS_2,A_1=1)=\frac{\{\nu(\bar\bS_2)-\mu(\bS_1)\}}{\pi^*(\bS_1)}\E(r_3\mid\bar\bS_2,A_1=1)=0.$$
Besides,
\begin{align*}
&\E(r_4r_3\mid\bar\bS_2,A_1=1)\\
&\qquad=\E\left[\frac{A_1A_2\left\{\nu^*(\bar\bS_2)-\nu(\bar\bS_2)\right\}\{Y_{1,1}-\nu(\bar\bS_2)\}}{\{\pi^*(\bS_1)\}^2\rho^*(\bar\bS_2)}\left\{1-\frac{1}{\rho^*(\bar\bS_2)}\right\}\mid\bar\bS_2,A_1=1\right]\\
&\qquad\overset{(i)}{=}\frac{A_1\rho(\bar\bS_2)\left\{\nu^*(\bar\bS_2)-\nu(\bar\bS_2)\right\}\{\nu(\bar\bS_2)-\nu(\bar\bS_2)\}}{\{\pi^*(\bS_1)\}^2\rho^*(\bar\bS_2)}\left\{1-\frac{1}{\rho^*(\bar\bS_2)}\right\}=0,
\end{align*}
where (i) holds since $A_2\independent Y_{1,1}\mid(\bar\bS_2,A_1=1)$. Note that for each $j\in\{1,2,4\}$, we also have $\E(r_jr_3\mid\bar\bS_2,A_1=0)=0$. By the tower rule,
$$\E(r_jr_3)=0,\;\;\forall j\in\{1,2,4\}.$$
Together with \eqref{prod-0j},
\begin{equation}\label{bound:lower-sig}
\sigma^2=\E(r_0^2)+\E(r_3^2)+\E(r_1+r_2+r_4)^2\geq\E(r_0^2)+\E(r_3^2).
\end{equation}
Under Assumptions \ref{cond:basic} and \ref{cond:subG},
\begin{align}
\E(r_3^2)&=\E\left[\frac{A_1A_2\{Y_{1,1}-\nu(\bar\bS_2)\}^2}{\{\pi^*(\bS_1)\rho^*(\bar\bS_2)\}^2}\right]\nonumber\\
&\geq\E\left[\frac{A_1A_2\{Y_{1,1}-\nu(\bar\bS_2)\}^2}{(1-c_0)^4}\right]\geq\frac{c_Y}{(1-c_0)^4}.\label{bound:-lower-r3}
\end{align}
Besides, by the construction of $\bbeta^*$, \eqref{eq:gradient1} holds. That is,
\begin{align}
\bzero&=\E\left(A_1\left\{\frac{1}{\pi^*(\bS_1)}-1\right\}\left[\frac{A_2\{Y_{1,1}-\nu^*(\bar\bS_2)\}}{\rho^*(\bar\bS_2)}+\nu^*(\bar\bS_2)-\bS_1^\top\bbeta^*\right]\bS_1\right)\nonumber\\
&=\bQ_1+\bQ_2+\bQ_3+\bQ_4,\label{eq:Q1-4}
\end{align}
where 
\begin{align*}
\bQ_1&:=\E\left[A_1\left\{\frac{1}{\pi^*(\bS_1)}-1\right\}\left\{1-\frac{A_2}{\rho^*(\bar\bS_2)}\right\}\{\nu^*(\bar\bS_2)-\nu(\bar\bS_2)\}\bS_1\right],\\
\bQ_2&:=\E\left[A_1\left\{\frac{1}{\pi^*(\bS_1)}-1\right\}\frac{A_2\{Y_{1,1}-\nu(\bar\bS_2)\}}{\rho^*(\bar\bS_2)}\bS_1\right],\\
\bQ_3&:=\E\left[A_1\left\{\frac{1}{\pi^*(\bS_1)}-1\right\}\{\nu(\bar\bS_2)-\mu(\bS_1)\}\bS_1\right],\\
\bQ_4&:=\E\left[A_1\left\{\frac{1}{\pi^*(\bS_1)}-1\right\}\{\mu(\bS_1)-\bS_1^\top\bbeta^*\}\bS_1\right].
\end{align*}
Observe that
\begin{align*}
\bQ_1&\overset{(i)}{=}\E\left[A_1\left\{\frac{1}{\pi^*(\bS_1)}-1\right\}\left\{1-\frac{\rho(\bar\bS_2)}{\rho^*(\bar\bS_2)}\right\}\{\nu^*(\bar\bS_2)-\nu(\bar\bS_2)\}\bS_1\right]\overset{(ii)}{=}\bzero,\\
\bQ_2&\overset{(iii)}{=}\E\left[A_1\left\{\frac{1}{\pi^*(\bS_1)}-1\right\}\frac{\rho(\bar\bS_2)\{\nu(\bar\bS_2)-\nu(\bar\bS_2)\}}{\rho^*(\bar\bS_2)}\bS_1\right]=0,\\
\bQ_3&\overset{(iv)}{=}\E\left[A_1\left\{\frac{1}{\pi^*(\bS_1)}-1\right\}\E\{\nu(\bar\bS_2)-\mu(\bS_1)\mid\bS_1,A_1=1\}\pi(\bS_1)\bS_1\right]\overset{(v)}{=}0,
\end{align*}
where (i) and (iv) hold by the tower rule; (ii) holds since either $\rho^*(\cdot)$ or $\nu^*(\cdot)$ is correctly specified; (iii) holds by the tower rule and the fact that $A_2\independent Y_{1,1}\mid(\bar\bS_2,A_1=1)$; (v) follows from \eqref{eq:rep_mu}.
Together with \eqref{eq:Q1-4}, we conclude that $\bQ_4=\bzero$. Since $\be_1^\top\bS_1=1$, it follows that
\begin{align*}
&\E\left[A_1\left\{\frac{1}{\pi^*(\bS_1)}-1\right\}\{\mu(\bS_1)-\bS_1^\top\bbeta^*\}(\bS_1^\top\bbeta^*-\theta_{1,1})\right]=\bQ_4^\top\bbeta^*-\theta_{1,1}\be_1^\top\bQ_4=0.
\end{align*}
Hence,
\begin{align*}
&\E\left[A_1\left\{\frac{1}{\pi^*(\bS_1)}-1\right\}\{\mu(\bS_1)-\theta_{1,1}\}^2\right]\\
&\qquad=\E\left[A_1\left\{\frac{1}{\pi^*(\bS_1)}-1\right\}\{\mu(\bS_1)-\bS_1^\top\bbeta^*\}^2\right]+\E\left[A_1\left\{\frac{1}{\pi^*(\bS_1)}-1\right\}(\bS_1^\top\bbeta^*-\theta_{1,1})^2\right]\\
&\qquad\geq\E\left[A_1\left\{\frac{1}{\pi^*(\bS_1)}-1\right\}(\bS_1^\top\bbeta^*-\theta_{1,1})^2\right]\\
&\qquad\overset{(i)}{=}\E\left[\pi(\bS_1)\left\{\frac{1}{\pi^*(\bS_1)}-1\right\}(\bS_1^\top\bbeta^*-\theta_{1,1})^2\right]\\
&\qquad\overset{(ii)}{\geq}c_0\left(\frac{1}{1-c_0}-1\right)\E(\bS_1^\top\bbeta^*-\theta_{1,1})^2\overset{(iii)}{\geq}\frac{c_0^2c_{\min}}{1-c_0}\|\bbeta^*\|_2^2,
\end{align*}
where (i) holds by the tower rule; (ii) holds under Assumption \ref{cond:basic}; (iii) holds under Assumption \ref{cond:subG}. Additionally, under Assumption \ref{cond:basic},
\begin{align*}
\E\left[A_1\left\{\frac{1}{\pi^*(\bS_1)}-1\right\}\{\mu(\bS_1)-\theta_{1,1}\}^2\right]\leq(c_0^{-1}-1)\E(r_0^2).
\end{align*}
Hence,
\begin{align*}
\E(r_0^2)\geq\frac{c_0^2}{(1-c_0)^2}\|\bbeta^*\|_2^2.
\end{align*}
Together with \eqref{bound:lower-sig} and \eqref{bound:-lower-r3},
$$\sigma^2\geq\frac{c_0^2}{(1-c_0)^2}\|\bbeta^*\|_2^2+\frac{c_Y}{(1-c_0)^4}.$$
Together with \eqref{bound-upper-sigma}, we conclude that
\begin{equation}\label{rate:sigma}
\sigma\asymp\|\bbeta^*\|_2+1.
\end{equation}
\end{proof}

\begin{proof}[Proof of Theorem \ref{thm:main}.]
Consider the representation \eqref{rep:thetahat}. 

\textbf{Step 1.} We first show that
$$\thetahat_{1,1}-\theta_{1,1}=N^{-1}\sum_{i=1}^N\psi(\bW_i;\etabf^*)-\theta_{1,1}+o_p(N^{-1/2}).$$
As shown in the proof of Theorem \ref{thm:rate}, we have \eqref{result:step1}, \eqref{rate:Deltak2}, and \eqref{rate:Deltak1} hold. Hence,
\begin{align*}
\thetahat_{1,1}-\theta_{1,1}&=N^{-1}\sum_{i=1}^N\psi(\bW_i;\etabf^*)-\theta_{1,1}+o_p(N^{-1/2})+O_p(r_{\bgamma}r_{\bbeta}+r_{\bdelta}r_{\balpha})\nonumber\\
&\qquad+\idf_{\rho\neq\rho^*}O_p(r_{\bgamma}r_{\balpha})+\idf_{\nu\neq\nu^*}O_p(r_{\bgamma}r_{\bdelta}+r_{\bdelta}^2)\\
&\qquad+\idf_{\mu\neq\mu^*}O_p(r_{\bgamma}^2+r_{\bgamma}r_{\bdelta}+r_{\bgamma}r_{\balpha}).
\end{align*}
Note that,
\begin{align*}
r_{\bgamma}r_{\bbeta}+r_{\bdelta}r_{\balpha}&=\frac{\sqrt{s_{\bgamma}s_{\bbeta}}\log d_1}{N}+\frac{\sqrt{s_{\bdelta}s_{\balpha}}\log d}{N}\\
&=o(N^{-1/2}),\qquad\text{when \eqref{cond:s2} is assumed};\\
\idf_{\rho\neq\rho^*}r_{\bgamma}r_{\balpha}&=\idf_{\rho\neq\rho^*}\left(\frac{\sqrt{s_{\bgamma}s_{\bbeta}}\log d_1}{N}+\frac{\sqrt{s_{\bgamma}s_{\balpha}\log d_1\log d}}{N}\right)\\
&=o(N^{-1/2}),\qquad\text{when \eqref{cond:s3} is assumed};\\
\idf_{\nu\neq\nu^*}(r_{\bgamma}r_{\bdelta}+r_{\bdelta}^2)&=\idf_{\nu\neq\nu^*}\frac{\sqrt{s_{\bdelta}\log d(s_{\bgamma}\log d_1+s_{\bdelta}\log d)}}{N}\\
&=o(N^{-1/2}),\qquad\text{when \eqref{cond:s4} is assumed};\\
\idf_{\mu\neq\mu^*}r_{\bgamma}(r_{\bgamma}+r_{\bdelta}+r_{\balpha})&=\idf_{\mu\neq\mu^*}\frac{\sqrt{s_{\bgamma}\log d_1\{s_{\bgamma}\log d_1+(s_{\bdelta}+s_{\balpha})\log d\}}}{N}\\
&=o(N^{-1/2}),\qquad\text{when \eqref{cond:s5} is assumed}.
\end{align*}
Hence, we conclude that
$$\thetahat_{1,1}-\theta_{1,1}=N^{-1}\sum_{i=1}^N\psi(\bW_i;\etabf^*)-\theta_{1,1}+o_p(N^{-1/2}).$$

\textbf{Step 2.} We show that, as $N,d_1,d_2\to\infty$,
\begin{equation}\label{normal:score}
\sigma^{-1}N^{-1/2}\sum_{i=1}^N\psi(\bW_i;\etabf^*)-\theta_{1,1}\to\mathcal N(0,1).
\end{equation}
As shown in \eqref{rate:sigma}, $\sigma\asymp\|\bbeta^*\|_2+1$. Additionally, for any constant $r>2$, by \eqref{bound:sigma-upper},
\begin{align}\label{cond:Lyapunov}
\sigma^{-r}\E\left\{\lvert\psi(\bW;\etabf^*)-\theta_{1,1}\rvert^r\right\}=O\left(\frac{(\|\bbeta^*\|_2+1)^r}{(\|\bbeta^*\|_2+1)^r}\right)=O(1),
\end{align}
By Lyapunov’s central limit theorem, \eqref{normal:score} holds.

\textbf{Step 3.} Lastly, we prove that, as $N,d_1,d_2\to\infty$,
\begin{equation}\label{normal:score2}
\sigmahat^2=\sigma^2\{1+o_p(1)\}.
\end{equation}
Note that
$$\E\left\{\psi(\bW_i;\etabf^*)-\theta_{1,1}\right\}^2=\sigma^2\asymp\|\bbeta^*\|_2^2+1.$$
By \eqref{cond:Lyapunov}, for any $r>2$,
$$\frac{\E\left\lvert\psi(\bW_i;\etabf^*)-\theta_{1,1}\right\rvert^r}{N^{r/2-1}(\|\bbeta^*\|_2+1)^r}=O(N^{1-r/2})=o(1).$$
Additionally, as shown in Theorem \ref{thm:rate}, $\thetahat_{1,1}-\theta_{1,1}=o_p(\|\bbeta^*\|_2+1)$. By \eqref{rate:moment2}, we also have $\E[\{\psi(\bW;\etabfhat_{-k})-\psi(\bW;\etabf^*)\}^2]=o_p(1)=o_p(\|\bbeta^*\|_2+1)$ for each $k\leq\K$. By Lemma D.4 of \cite{zhang2023double}, we have \eqref{normal:score2} holds.
\end{proof}

\begin{proof}[Proof of Theorem \ref{cor:correct}.]
Theorem \ref{cor:correct} follows directly from Theorem \ref{thm:main} as a special case that all the nuisance models are correctly specified.
\end{proof}

\subsection{Proofs of the results in Section \ref{sec:asymp_nuisance}}

To begin with, we first demonstrate some useful lemmas before we obtain the asymptotic results for the moment-targeted nuisance estimators. For any $\bgamma,\bbeta\in\R^{d_1}$ and $\bdelta,\balpha\in\R^d$, define
\begin{align}
&\bar\ell_1(\bgamma):=M^{-1}\sum_{i\in\mathcal I_{\bgamma}}\ell_1(\bW_i;\bgamma),\quad\bar\ell_2(\bgamma,\bdelta):=M^{-1}\sum_{i\in\mathcal I_{\bdelta}}\ell_2(\bW_i;\bgamma,\bdelta),\label{ell_2_bar}\\
&\bar\ell_3(\bgamma,\bdelta,\balpha):=M^{-1}\sum_{i\in\mathcal I_{\balpha}}\ell_3(\bW_i;\bgamma,\bdelta,\balpha),\nonumber\\
&\bar\ell_4(\bgamma,\bdelta,\balpha,\bbeta):=M^{-1}\sum_{i\in\mathcal I_{\bbeta}}\ell_4(\bW_i;\bgamma,\bdelta,\balpha,\bbeta),\nonumber
\end{align}
where the loss functions are defined as \eqref{def:alpha}, \eqref{def:l2}, \eqref{def:l3}, and \eqref{def:l4}. For any $\bgamma,\bbeta,\bDelta\in\R^{d_1}$ and $\balpha,\bdelta\in\R^d$, define
\begin{align}
\delta\bar\ell_1(\bgamma,\bDelta)&:=\bar\ell_1(\bgamma+\bDelta)-\bar\ell_1(\bgamma)-\bnabla_{\bgamma}\bar\ell_1(\bgamma)^\top\bDelta,\label{delta_ell_1_bar}\\
\delta\bar\ell_4(\bgamma,\bdelta,\balpha,\bbeta,\bDelta)&:=\bar\ell_4(\bgamma,\bdelta,\balpha,\bbeta+\bDelta)-\bar\ell_4(\bgamma,\bdelta,\balpha,\bbeta)-\bnabla_{\bbeta}\bar\ell_4(\bgamma,\bdelta,\balpha,\bbeta)^\top\bDelta\label{delta_ell_4_bar}.
\end{align}
Similarly, for any $\bgamma\in\R^{d_1}$ and $\balpha,\bdelta,\bDelta\in\R^d$, define
\begin{align}
\delta\bar\ell_2(\bgamma,\bdelta,\bDelta)&:=\bar\ell_2(\bgamma,\bdelta+\bDelta)-\bar\ell_2(\bgamma,\bdelta)-\bnabla_{\bdelta}\bar\ell_2(\bgamma,\bdelta)^\top\bDelta,\label{delta_ell_2_bar}\\
\delta\bar\ell_3(\bgamma,\bdelta,\balpha,\bDelta)&:=\bar\ell_3(\bgamma,\bdelta,\balpha+\bDelta)-\bar\ell_3(\bgamma,\bdelta,\balpha)-\bnabla_{\balpha}\bar\ell_3(\bgamma,\bdelta,\balpha)^\top\bDelta\label{delta_ell_3_bar}.
\end{align}

We demonstrate the following restricted strong convexity (RSC) conditions. Note that, the nuisance estimators are constructed based on different samples, and the probability measures in \eqref{eq:RSC1}-\eqref{eq:RSC3} are different.
\begin{lemma}\label{lemma:RSC}
Let Assumptions \ref{cond:basic} and \ref{cond:subG} hold. Define $f_{M,d_1}(\bDelta):=\kappa_1\|\bDelta\|_2^2-\kappa_2\|\bDelta\|_1^2\log d_1/M$ for any $\bDelta\in\R^{d_1}$ and $f_{M,d}(\bDelta):=\kappa_1\|\bDelta\|_2^2-\kappa_2\|\bDelta\|_1^2\log d/M$ for any $\bDelta\in\R^d$. Then, with some constants $\kappa_1,\kappa_2,c_1,c_2>0$ and note that $M\asymp N$, we have 
\begin{align}
&\P_{\S_{\bgamma}}\left(\delta\bar\ell_1(\bgamma^*,\bDelta)\geq f_{M,d_1}(\bDelta),\;\;\forall\|\bDelta\|_2\leq1\right)\geq1-c_1\exp(-c_2M).\label{eq:RSC1}
\end{align}

Further, let $\|\bgammahat-\bgamma^*\|_2\leq1$. Then
\begin{align}
&\P_{\S_{\bdelta}}\left(\delta\bar\ell_2(\bgammahat,\bdelta^*,\bDelta)\geq f_{M,d}(\bDelta),\;\;\forall\|\bDelta\|_2\leq1\right)\geq1-c_1\exp(-c_2M),\label{eq:RSC2}\\
&\P_{\S_{\bbeta}}\left(\delta\bar\ell_4(\bgammahat,\bdeltahat,\balphahat,\bbeta^*,\bDelta)\geq f_{M,d_1}(\bDelta),\;\;\forall\bDelta\in\R^{d_1}\right)\geq1-c_1\exp(-c_2M).\label{eq:RSC4}
\end{align}
In \eqref{eq:RSC2}, we only consider the randomness in $\S_{\bdelta}$, and $\bgammahat$ is treated as fixed (or conditional on). Similarly, in \eqref{eq:RSC4}, $\bgammahat$, $\bdeltahat$, and $\balphahat$ are all treated as fixed.

Moreover, let $\|\bdeltahat-\bdelta^*\|_2\leq1$. Then 
\begin{align}
&\P_{\S_{\balpha}}\left(\delta\bar\ell_3(\bgammahat,\bdeltahat,\balpha^*,\bDelta)\geq f_{M,d}(\bDelta),\;\;\forall\bDelta\in\R^d\right)\geq1-c_1\exp(-c_2M),\label{eq:RSC3}
\end{align}
where $\bgammahat$ and $\bdeltahat$ are treated as fixed.
\end{lemma}

\begin{proof}[Proof of Lemma \ref{lemma:RSC}.]
We show that, with high probability, the RSC property holds for each of the loss functions. By Taylor's theorem, with some $v_1,v_2\in(0,1)$,
\begin{align}
\delta\bar\ell_1(\bgamma^*,\bDelta)&=(2M)^{-1}\sum_{i\in\mathcal I_{\bgamma}}A_{1i}\exp\{-\bS_{1i}^\top(\bgamma^*+v_1\bDelta)\}(\bS_{1i}^\top\bDelta)^2,\nonumber\\
\delta\bar\ell_2(\bgammahat,\bdelta^*,\bDelta)&=(2M)^{-1}\sum_{i\in\mathcal I_{\bdelta}}A_{1i}A_{2i}g^{-1}(\bS_{1i}^\top\bgammahat)\exp\{-\bar\bS_{2i}^\top(\bdelta^*+v_2\bDelta)\}(\bar\bS_{2i}^\top\bDelta)^2,\label{deltaell_beta}\\
\delta\bar\ell_3(\bgammahat,\bdeltahat,\balpha^*,\bDelta)&=M^{-1}\sum_{i\in\mathcal I_{\balpha}}A_{1i}A_{2i}g^{-1}(\bS_{1i}^\top\bgammahat)\exp(-\bar\bS_{2i}^\top\bdeltahat)(\bar\bS_{2i}^\top\bDelta)^2,\label{deltaell_gamma}\\
\delta\bar\ell_4(\bgammahat,\bdeltahat,\balphahat,\bbeta^*,\bDelta)&=M^{-1}\sum_{i\in\mathcal I_{\bbeta}}A_{1i}\exp(-\bS_{1i}^\top\bgammahat)(\bS_{1i}^\top\bDelta)^2.\label{deltaell_delta}
\end{align}

\textbf{Part 1.} Let $\bU=A_1\bS_1$, $\S'=(A_{1i}\bS_{1i})_{i\in\mathcal I_{\bgamma}}$, $\phi(u)=\exp(-u)$, $v=v_1$, and $\etabf=\bgamma^*$. Under Assumption \ref{cond:basic}, $\lvert\bU^\top\etabf\rvert\leq\lvert\bS_1^\top\bgamma^*\rvert<C$ with some constant $C>0$. By Lemmas \ref{lemma:subG'} and \ref{lemma:prop2}, we have \eqref{eq:RSC1} holds. Note that, $\P_{\S_{\bgamma}}$, $\P_{\S_{\bdelta}}$, $\P_{\S_{\balpha}}$, and $\P_{\S_{\bbeta}}$ are the probability measures corresponding to disjoint (and independent) sub-samples $\S_{\bgamma}$, $\S_{\bdelta}$, $\S_{\balpha}$, and $\S_{\bbeta}$, respectively.

\vspace{0.5em}

\textbf{Part 2.} Now we treat $\bgammahat$ as fixed (or conditional on) and suppose that $\|\bgammahat-\bgamma^*\|_2\leq1$. Note that $g^{-1}(u)=1+\exp(-u)$ and $\bS=(\bS_1^\top,\bS_2^\top)^\top$. Hence,
\begin{align*}
\delta\bar\ell_2(\bgammahat,\bdelta^*,\bDelta)&=(2M)^{-1}\sum_{i\in\mathcal I_{\bdelta}}A_{1i}A_{2i}\exp\{-\bar\bS_{2i}^\top(\bdelta^*+v_2\bDelta)\}(\bar\bS_{2i}^\top\bDelta)^2\\
&\qquad+(2M)^{-1}\sum_{i\in\mathcal I_{\bdelta}}A_{1i}A_{2i}\exp\{-\bar\bS_{2i}^\top(\bdelta^*+\check{\bgamma}+v_2\bDelta)\}(\bar\bS_{2i}^\top\bDelta)^2,
\end{align*}
where $\check{\bgamma}=(\bgammahat^\top,0,\dots,0)^\top\in\R^d$. Let $\bU=A_1A_2\bS$, $\S'=(A_{1i}A_{2i}\bar\bS_{2i})_{i\in\mathcal I_{\bdelta}}$, $\phi(u)=\exp(-u)$, $v=v_2$, and $\etabf=\bdelta^*$. Note that, under Assumption \ref{cond:basic}, we have $\lvert\bU^\top\etabf\rvert\leq\lvert\bar\bS_2^\top\bdelta^*\rvert<C$ with some constant $C>0$. By Lemmas \ref{lemma:subG'} and \ref{lemma:prop2}, we have
\begin{align}
&(2M)^{-1}\sum_{i\in\mathcal I_{\bdelta}}A_{1i}A_{2i}\exp\{-\bar\bS_{2i}^\top(\bdelta^*+v_2\bDelta)\}(\bar\bS_{2i}^\top\bDelta)^2\nonumber\\
&\qquad\geq\kappa_1'\|\bDelta\|_2^2-\kappa_2'\frac{\log d}{M}\|\bDelta\|_1^2,\quad\forall\|\bDelta\|_2\leq1,\label{eq:RSC2_1}
\end{align}
with probability $\P_{\S_{\bdelta}}$ at least $1-c_1'\exp(-c_2'M)$ and some constants $\kappa_1',\kappa_2',c_1',c_2'>0$. 

Similarly, let $\bU=A_1A_2\bS$, $\S'=(A_{1i}A_{2i}\bar\bS_{2i})_{i\in\mathcal I_{\bdelta}}$, $\phi(u)=\exp(-u)$, $v=v_2$, and $\etabf=\bdelta^*+\check{\bgamma}$. On the event $\|\bgammahat-\bgamma^*\|_2\leq1$, under Assumptions \ref{cond:basic} and \ref{cond:subG}, we have $\E\{\lvert\bU^\top\etabf\rvert\}\leq\E(\lvert\bS_1^\top\bgamma^*\rvert)+\E(\lvert\bar\bS_2^\top\bdelta^*\rvert)+\E\{\lvert\bS_1^\top(\bgammahat-\bgamma^*)\rvert\}<C$ with some constant $C>0$. By Lemmas \ref{lemma:subG'} and \ref{lemma:prop2}, we have 
\begin{align}
&(2M)^{-1}\sum_{i\in\mathcal I_{\bdelta}}A_{1i}A_{2i}\exp\{-\bar\bS_{2i}^\top(\bdelta^*+\check{\bgamma}+v_2\bDelta)\}(\bar\bS_{2i}^\top\bDelta)^2\nonumber\\
&\qquad\geq\kappa_1'\|\bDelta\|_2^2-\kappa_2'\frac{\log d}{M}\|\bDelta\|_1^2,\quad\forall\|\bDelta\|_2\leq1,\label{eq:RSC2_2}
\end{align}
with probability $\P_{\S_{\bdelta}}$ at least $1-c_1'\exp(-c_2'M)$. Hence, \eqref{eq:RSC2} follows from \eqref{eq:RSC2_1} and \eqref{eq:RSC2_2}.

\vspace{0.5em}

\textbf{Part 3.} We treat both $\bgammahat$ and $\bdeltahat$ as fixed (or conditional on) and suppose that $\|\bgammahat-\bgamma^*\|_2\leq1$, $\|\bdeltahat-\bdelta^*\|_2\leq1$. Note that
\begin{align}
\delta\bar\ell_3(\bgammahat,\bdeltahat,\balpha^*,\bDelta)&=M^{-1}\sum_{i\in\mathcal I_{\balpha}}A_{1i}A_{2i}\exp(-\bar\bS_{2i}^\top\bdeltahat)(\bar\bS_{2i}^\top\bDelta)^2\nonumber\\
&\qquad+M^{-1}\sum_{i\in\mathcal I_{\balpha}}A_{1i}A_{2i}\exp\{-\bar\bS_{2i}^\top(\bdeltahat+\check{\bgamma})\}(\bar\bS_{2i}^\top\bDelta)^2.\label{eq:split}
\end{align}
Let $\bU=A_1A_2\bS$, $\S'=(A_{1i}A_{2i}\bar\bS_{2i})_{i\in\mathcal I_{\balpha}}$, $\phi(u)=\exp(-u)$, $v=0$, and $\etabf=\bdeltahat$. Here, $\E(\lvert\bU^\top\etabf\rvert)\leq\E(\lvert\bar\bS_2^\top\bdelta^*\rvert)+\E\{\lvert\bar\bS_2^\top(\bdeltahat-\bdelta^*)\rvert\}<C$ with some constant $C>0$. By Lemmas \ref{lemma:subG'} and \ref{lemma:prop2}, we have
\begin{align}
&M^{-1}\sum_{i\in\mathcal I_{\balpha}}A_{1i}A_{2i}\exp(-\bar\bS_{2i}^\top\bdeltahat)(\bar\bS_{2i}^\top\bDelta)^2\geq\kappa_1'\|\bDelta\|_2^2-\kappa_2'\frac{\log d}{M}\|\bDelta\|_1^2,\quad\forall\|\bDelta\|_2\leq1,\label{eq:RSC3_1}
\end{align}
with probability $\P_{\S_{\balpha}}$ at least $1-c_1'\exp(-c_2'M)$. 

Similarly, let $\bU=A_1A_2\bS$, $\S'=(A_{1i}A_{2i}\bar\bS_{2i})_{i\in\mathcal I_{\balpha}}$, $\phi(u)=\exp(-u)$, $v=0$, and $\etabf=\bdeltahat+\check{\bgamma}$. Then $\E(\lvert\bU^\top\etabf\rvert)\leq\E(\lvert\bS_1^\top\bgamma^*\rvert)+\E(\lvert\bar\bS_2^\top\bdelta^*\rvert)+\E\{\lvert\bS_1^\top(\bgammahat-\bgamma^*)\rvert\}+\E\{\lvert\bar\bS_2^\top(\bdeltahat-\bdelta^*)\rvert\}<C$ with some constant $C>0$. By Lemmas \ref{lemma:subG'} and \ref{lemma:prop2}, we have
\begin{align}
M^{-1}\sum_{i\in\mathcal I_{\balpha}}A_{1i}A_{2i}\exp\{-\bar\bS_{2i}^\top(\bdeltahat+\check{\bgamma})\}(\bar\bS_{2i}^\top\bDelta)^2\geq\kappa_1'\|\bDelta\|_2^2-\kappa_2'\frac{\log d}{M}\|\bDelta\|_1^2,\;\;\forall\|\bDelta\|_2\leq1,\label{eq:RSC3_2}
\end{align}
with probability $\P_{\S_{\balpha}}$ at least $1-c_1'\exp(-c_2'M)$. Note that, the function $\delta\ell_N(\bgammahat,\bdeltahat,\balpha^*,\bDelta)$ is based on a weighted squared loss, and hence the lower bounds in \eqref{eq:RSC3_1} and \eqref{eq:RSC3_2} can be extended to any $\bDelta\in\R^d$. For any $\bDelta'\in\R^d$, we let $\bDelta=\bDelta'/\|\bDelta'\|_2$. Then $\|\bDelta\|_2=1$. The lower bounds in \eqref{eq:RSC3_1} and \eqref{eq:RSC3_2} hold if we multiply the LHS and RHS by a factor $\|\bDelta'\|_2^2$. Therefore, \eqref{eq:RSC3} holds by combining the lower bounds with \eqref{eq:split} . 

\vspace{0.5em}

\textbf{Part 4.} Lastly, treat $\bgammahat$ as fixed (or conditional on) and suppose that $\|\bgammahat-\bgamma^*\|_2\leq1$. Let $\bU=A_1\bS_1$, $\S'=(A_{1i}\bS_{1i})_{i\in\mathcal I_{\bbeta}}$, $\phi(u)=\exp(-u)$, $v=0$, and $\etabf=\bgammahat$. Here, $\E\{\lvert\bU^\top\etabf\rvert\}\leq\E(\lvert\bS_1^\top\bgamma^*\rvert)+\E\{\lvert\bS_1^\top(\bgammahat-\bgamma^*)\rvert\}<C$ with some constant $C>0$. Then \eqref{eq:RSC4} holds by Lemmas \ref{lemma:subG'} and \ref{lemma:prop2}. Here, similarly as in part 3, the lower bound can be extended to any $\bDelta\in\R^d$, since $\delta\ell_N(\bgammahat,\bdeltahat,\balphahat,\bbeta^*,\bDelta)$ is also constructed based on a weighted squared loss.
\end{proof}

Additionally, we control the gradients of the loss functions evaluated at the target population parameter values. Note that the nuisance parameters \(\bgamma^*\), \(\bdelta^*\), \(\balpha^*\), and \(\bbeta^*\) are defined as the minimizers of the corresponding loss functions, \eqref{def:alpha}, \eqref{def:l2}, \eqref{def:l3}, and \eqref{def:l4}. By the KKT conditions, the gradients of the expectations of the loss functions are zero vectors, as shown in \eqref{eq:gradient_zero1}, \eqref{eq:gradient_zero2}, \eqref{eq:gradient_zero3}, and \eqref{eq:gradient_zero4}. Therefore, the gradients of the empirical averages of the loss functions, \(\bnabla_{\bgamma}\bar\ell_1(\bgamma^*)\), \(\bnabla_{\bdelta}\bar\ell_2(\bgamma^*,\bdelta^*)\), \(\bnabla_{\balpha}\bar\ell_3(\bgamma^*,\bdelta^*,\balpha^*)\), and \(\bnabla_{\bbeta}\bar\ell_4(\bgamma^*,\bdelta^*,\balpha^*,\bbeta^*)\), are averages of i.i.d. random vectors with zero means, even under model misspecification. Hence, we can apply standard union bound techniques to control the infinity norms, with the usual rates \(O_p(\sqrt{\log d/M})\) or \(O_p(\sqrt{\log d_1/M})\).

\begin{lemma}\label{lemma:gradient}
Let Assumption \ref{cond:subG} holds. Let $\sigma_{\bgamma},\sigma_{\bdelta},\sigma_{\balpha},\sigma_{\bbeta}>0$ be some constants and note that $M\asymp N$. Then, for any $t>0$,
\begin{align*}
&\P_{\S_{\bgamma}}\left(\left\|\bnabla_{\bgamma}\bar\ell_1(\bgamma^*)\right\|_\infty\leq\sigma_{\bgamma}\sqrt\frac{t+\log d_1}{M}\right)\geq1-2\exp(-t).
\end{align*}
Further, let Assumption \ref{cond:basic} holds. Then,  for any $t>0$,
\begin{align*}
\P_{\S_{\bdelta}}\left(\left\|\bnabla_{\bdelta}\bar\ell_2(\bgamma^*,\bdelta^*)\right\|_\infty\leq\sigma_{\bdelta}\sqrt\frac{t+\log d}{M}\right)&\geq1-2e^{-t},\\
\P_{\S_{\balpha}}\left(\left\|\bnabla_{\balpha}\bar\ell_3(\bgamma^*,\bdelta^*,\balpha^*)\right\|_\infty\leq\sigma_{\balpha}\left(2\sqrt\frac{t+\log d}{M}+\frac{t+\log d}{M}\right)\right)&\geq1-2e^{-t},\\
\P_{\S_{\bbeta}}\left(\left\|\bnabla_{\bbeta}\bar\ell_4(\bgamma^*,\bdelta^*,\balpha^*,\bbeta^*)\right\|_\infty\leq\sigma_{\bbeta}\left(2\sqrt\frac{t+\log d_1}{M}+\frac{t+\log d_1}{M}\right)\right)&\geq1-2e^{-t}.
\end{align*}
\end{lemma}

\begin{proof}[Proof of Lemma \ref{lemma:gradient}.]
Now, we control the gradients of the loss functions.

\textbf{Part 1.} Note that
$$\bnabla_{\bgamma}\bar\ell_1(\bgamma^*)=M^{-1}\sum_{i\in\mathcal I_{\bgamma}}\{1-A_{1i}g^{-1}(\bS_{1i}^\top\bgamma^*)\}\bS_{1i}.$$
By the construction of $\bgamma^*$, we have
\begin{align}\label{eq:gradient_zero1}
\E\left[\{1-A_1g^{-1}(\bS_1^\top\bgamma^*)\}\bS_1\right]=\bzero\in\R^{d_1}.
\end{align}
Also, for each $1\leq j\leq d_1$, $\lvert\{1-A_1g^{-1}(\bS_1^\top\bgamma^*)\}\bS_1^\top\be_j\rvert\leq(1+c_0^{-1})\lvert\bS_1^\top\be_j\rvert$ and hence, by Lemma \ref{lemma:psi2norm}, 
$$\|\{1-A_1g^{-1}(\bS_1^\top\bgamma^*)\}\bS_1^\top\be_j\|_{\psi_2}\leq(1+c_0^{-1})\|\bS_1^\top\be_j\|_{\psi_2}\leq(1+c_0^{-1})\sigma_{\bS}.$$
Let $\sigma_{\bgamma}:=\sqrt8(1+c_0^{-1})\sigma_{\bS}$. By Lemma D.2 of \cite{chakrabortty2019high}, for each $1\leq j\leq d_1$ and any $t>0$,
$$\P_{\S_{\bgamma}}\left(\left\lvert\bnabla_{\bgamma}\bar\ell_1(\bgamma^*)^\top\be_j\right\rvert>\sigma_{\bgamma}\sqrt\frac{t+\log d_1}{M}\right)\leq2\exp(-t-\log d_1).$$
It follows that,
\begin{align*}
&\P_{\S_{\bgamma}}\left(\left\|\bnabla_{\bgamma}\bar\ell_1(\bgamma^*)\right\|_\infty>\sigma_{\bgamma}\sqrt\frac{t+\log d_1}{M}\right)\\
&\qquad\leq\sum_{j=1}^{d_1}\P_{\S_{\bgamma}}\left(\left\lvert\bnabla_{\bgamma}\bar\ell_1(\bgamma^*)^\top\be_j\right\rvert>\sigma_{\bgamma}\sqrt\frac{t+\log d_1}{M}\right)\\
&\qquad\leq 2d_1\exp(-t-\log d_1)=2\exp(-t).
\end{align*}

\vspace{0.5em}

\textbf{Part 2.} Note that
$$\bnabla_{\bdelta}\bar\ell_2(\bgamma^*,\bdelta^*)=M^{-1}\sum_{i\in\mathcal I_{\bdelta}}A_{1i}g^{-1}(\bS_{1i}^\top\bgamma^*)\{1-A_{2i}g^{-1}(\bar\bS_{2i}^\top\bdelta^*)\}\bar\bS_{2i}.$$
By the construction of $\bdelta^*$, we have
\begin{align}\label{eq:gradient_zero2}
\E\left[A_1g^{-1}(\bS_1^\top\bgamma^*)\{1-A_2g^{-1}(\bar\bS_2^\top\bdelta^*)\}\bar\bS_2\right]=\bzero\in\R^d.
\end{align}
Under Assumption \ref{cond:basic}, we have $\lvert A_1g^{-1}(\bS_1^\top\bgamma^*)\{1-A_2g^{-1}(\bar\bS_2^\top\bdelta^*)\}\bar\bS_2^\top\be_j\rvert\leq c_0^{-1}(1+c_0^{-1})\lvert\bar\bS_2^\top\be_j\rvert$ for each $1\leq j\leq d$. By Lemma \ref{lemma:psi2norm}, 
$$\|A_1g^{-1}(\bS_1^\top\bgamma^*)\{1-A_2g^{-1}(\bar\bS_2^\top\bdelta^*)\}\bar\bS_2^\top\be_j\|_{\psi_2}\leq (c_0^{-2}+c_0^{-1})\|\bar\bS_2^\top\be_j\|_{\psi_2}\leq (c_0^{-2}+c_0^{-1})\sigma_{\bS}.$$
Let $\sigma_{\bdelta}:=\sqrt8(c_0^{-2}+c_0^{-1})\sigma_{\bS}$. By Lemma D.2 of \cite{chakrabortty2019high}, for each $1\leq j\leq d$ and any $t>0$,
$$\P_{\S_{\bdelta}}\left(\left\lvert\bnabla_{\bdelta}\bar\ell_2(\bgamma^*,\bdelta^*)^\top\be_j\right\rvert>\sigma_{\bdelta}\sqrt\frac{t+\log d}{M}\right)\leq2\exp(-t-\log d).$$
It follows that,
\begin{align*}
&\P_{\S_{\bdelta}}\left(\left\|\bnabla_{\bdelta}\bar\ell_2(\bgamma^*,\bdelta^*)\right\|_\infty>\sigma_{\bdelta}\sqrt\frac{t+\log d}{M}\right)\\
&\qquad\leq\sum_{j=1}^{d}\P_{\S_{\bdelta}}\left(\left\lvert\bnabla_{\bdelta}\bar\ell_2(\bgamma^*,\bdelta^*)^\top\be_j\right\rvert>\sigma_{\bdelta}\sqrt\frac{t+\log d}{M}\right)\\
&\qquad\leq 2d\exp(-t-\log d)=2\exp(-t).
\end{align*}

\vspace{0.5em}

\textbf{Part 3.} Note that
$$\bnabla_{\balpha}\bar\ell_3(\bgamma^*,\bdelta^*,\balpha^*)=-2M^{-1}\sum_{i\in\mathcal I_{\balpha}}A_{1i}A_{2i}g^{-1}(\bS_{1i}^\top\bgamma^*)\exp(-\bar\bS_{2i}^\top\bdelta^*)\varepsilon_i\bar\bS_{2i}.$$
By the construction of $\balpha^*$, we have
\begin{align}\label{eq:gradient_zero3}
\E\left\{-2A_1A_2g^{-1}(\bS_1^\top\bgamma^*)\exp(-\bar\bS_2^\top\bdelta^*)\varepsilon\bar\bS_2\right\}=\bzero\in\R^d.
\end{align}
Under Assumption \ref{cond:basic}, we have $\lvert-2A_1A_2g^{-1}(\bS_1^\top\bgamma^*)\exp(-\bar\bS_2^\top\bdelta^*)\varepsilon\bar\bS_2^\top\be_j\rvert\leq 2c_0^{-1}(c_0^{-1}-1)\lvert\varepsilon\bar\bS_2^\top\be_j\rvert$ for each $1\leq j\leq d$. By \ref{lemma:psi2norm}, 
\begin{align*}
&\|-2A_1A_2g^{-1}(\bS_1^\top\bgamma^*)\exp(-\bar\bS_2^\top\bdelta^*)\varepsilon\bar\bS_2^\top\be_j\|_{\psi_1}\\
&\qquad\leq 2c_0^{-1}(c_0^{-1}-1)\|\varepsilon\|_{\psi_2}\|\bar\bS_2^\top\be_j\|_{\psi_2}\leq 2c_0^{-1}(c_0^{-1}-1)\sigma_\varepsilon\sigma_{\bS}.
\end{align*}
Let $\sigma_{\balpha}:=2c_0^{-1}(c_0^{-1}-1)\sigma_\varepsilon\sigma_{\bS}$. By Lemma \ref{lemma:psi2norm} and Lemma D.4 of \cite{chakrabortty2019high}, for each $1\leq j\leq d$ and any $t>0$,
\begin{align*}
&\P_{\S_{\balpha}}\left(\left\lvert\bnabla_{\balpha}\bar\ell_3(\bgamma^*,\bdelta^*,\balpha^*)^\top\be_j\right\rvert>\sigma_{\balpha}\left(2\sqrt\frac{t+\log d}{M}+\frac{t+\log d}{M}\right)\right)\leq2\exp(-t-\log d).
\end{align*}
It follows that,
\begin{align*}
&\P_{\S_{\balpha}}\left(\left\|\bnabla_{\balpha}\bar\ell_3(\bgamma^*,\bdelta^*,\balpha^*)\right\|_\infty>\sigma_{\balpha}\left(2\sqrt\frac{t+\log d}{M}+\frac{t+\log d}{M}\right)\right)\\
&\qquad\leq\sum_{j=1}^{d}\P_{\S_{\balpha}}\left(\left\lvert\bnabla_{\balpha}\bar\ell_3(\bgamma^*,\bdelta^*,\balpha^*)^\top\be_j\right\rvert>\sigma_{\balpha}\left(2\sqrt\frac{t+\log d}{M}+\frac{t+\log d}{M}\right)\right)\\
&\qquad\leq 2d\exp(-t-\log d)=2\exp(-t).
\end{align*}

\vspace{0.5em}

\textbf{Part 4.} Note that
$$\bnabla_{\bbeta}\bar\ell_4(\bgamma^*,\bdelta^*,\balpha^*,\bbeta^*)=-2M^{-1}\sum_{i\in\mathcal I_{\bbeta}}A_{1i}\exp(-\bS_{1i}^\top\bgamma^*)\left\{\zeta_{i}+A_{2i}g^{-1}(\bar\bS_{2i}^\top\bdelta^*)\varepsilon_i\right\}\bS_{1i}.$$
By the construction of $\bbeta^*$, we have
\begin{align}\label{eq:gradient_zero4}
\E\left[-2A_1\exp(-\bS_1^\top\bgamma^*)\left\{\zeta+A_2g^{-1}(\bar\bS_2^\top\bdelta^*)\varepsilon\right\}\bS_1\right]=\bzero\in\R^{d_1}.\end{align}
Under Assumption \ref{cond:basic}, we have $\lvert-2A_1\exp(-\bS_1^\top\bgamma^*)\left\{\zeta+A_2g^{-1}(\bar\bS_2^\top\bdelta^*)\varepsilon\right\}\bS_1^\top\be_j\rvert\leq2(c_0^{-1}-1)(\lvert\zeta\rvert+c_0^{-1}\lvert\varepsilon\rvert)\lvert\bS_1^\top\be_j\rvert$ for each $1\leq j\leq d$. By Lemma \ref{lemma:psi2norm}, 
\begin{align*}
&\|-2A_1\exp(-\bS_1^\top\bgamma^*)\left\{\zeta+A_2g^{-1}(\bar\bS_2^\top\bdelta^*)\varepsilon\right\}\bS_1^\top\be_j\|_{\psi_1}\\
&\qquad\leq2(c_0^{-1}-1)(\|\zeta\|_{\psi_2}+c_0^{-1}\|\varepsilon\|_{\psi_2})\|\bS_1^\top\be_j\|_{\psi_2}\leq2(c_0^{-1}-1)(\sigma_\zeta+c_0^{-1}\sigma_\varepsilon)\sigma_{\bS}.
\end{align*}
Let $\sigma_{\bbeta}:=2(c_0^{-1}-1)(\sigma_\zeta+c_0^{-1}\sigma_\varepsilon)\sigma_{\bS}$. By Lemma \ref{lemma:psi2norm} and Lemma D.4 of \cite{chakrabortty2019high}, for each $1\leq j\leq d_1$ and any $t>0$,
\begin{align*}
&\P_{\S_{\bbeta}}\left(\left\lvert\bnabla_{\bbeta}\bar\ell_4(\bgamma^*,\bdelta^*,\balpha^*,\bbeta^*)^\top\be_j\right\rvert>\sigma_{\bbeta}\left(2\sqrt\frac{t+\log d_1}{M}+\frac{t+\log d_1}{M}\right)\right)\\
&\qquad\leq2\exp(-t-\log d_1).
\end{align*}
It follows that,
\begin{align*}
&\P_{\S_{\bbeta}}\left(\left\|\bnabla_{\bbeta}\bar\ell_4(\bgamma^*,\bdelta^*,\balpha^*,\bbeta^*)\right\|_\infty>\sigma_{\bbeta}\left(2\sqrt\frac{t+\log d_1}{M}+\frac{t+\log d_1}{M}\right)\right)\\
&\qquad\leq\sum_{j=1}^{d_1}\P_{\S_{\bbeta}}\left(\left\lvert\bnabla_{\bbeta}\bar\ell_4(\bgamma^*,\bdelta^*,\balpha^*,\bbeta^*)^\top\be_j\right\rvert>\sigma_{\bbeta}\left(2\sqrt\frac{t+\log d_1}{M}+\frac{t+\log d_1}{M}\right)\right)\\
&\qquad\leq 2d_1\exp(-t-\log d_1)=2\exp(-t).
\end{align*}
\end{proof}

Notably, the nuisance estimates \(\bgammahat\), \(\bdeltahat\), \(\balphahat\), and \(\bbetahat\) are constructed sequentially, with each later estimate being defined based on the previous ones. As a result, the \(\ell_\infty\) bounds developed in Lemma \ref{lemma:gradient} above cannot be directly applied, since the controlled gradients are evaluated at the target values \(\bgamma^*\), \(\bdelta^*\), \(\balpha^*\), and \(\bbeta^*\), which do not incorporate the imputation errors from earlier steps. While the construction of \(\bgammahat\) does not rely on other estimates and \(\balphahat\) and \(\bbetahat\) are both constructed using (weighted) square loss, the analysis of \(\bdeltahat\) presents the greatest challenge due to its dependence on \(\bgammahat\) and the associated non-square loss \eqref{def:l2}. Therefore, we first present several additional lemmas to characterize the behavior of the nuisance estimate \(\bdeltahat\) and demonstrate how it is influenced by the initial estimate \(\bgammahat\).

For any \(\bDelta \in \R^d\), define 
$$\mathcal F(\bDelta) := \delta\bar\ell_2(\bgammahat, \bdelta^*, \bDelta) + \lambda_{\bdelta} \|\bdelta^* + \bDelta\|_1 + \bnabla_{\bdelta} \bar\ell_2(\bgammahat, \bdelta^*)^\top \bDelta - \lambda_{\bdelta} \|\bdelta^*\|_1,$$ 
where \(\bar\ell_2(\bgamma, \bdelta)\) and \(\delta\bar\ell_2(\bgamma, \bdelta, \bDelta)\) are defined in \eqref{ell_2_bar} and \eqref{delta_ell_2_bar}, respectively. Let \(\bDelta_{\bdelta} = \bdeltahat - \bdelta^*\). By construction, we ensure that \(\mathcal F(\bDelta_{\bdelta}) \leq 0\). Additionally, the RSC property established in Lemma \ref{lemma:RSC} provides a lower bound for \(\delta\bar\ell_2(\bgammahat, \bdelta^*, \bDelta_{\bdelta})\) with high probability, as long as \(\|\bDelta_{\bdelta}\|_2 \leq 1\), a condition that will be guaranteed in Lemma \ref{lemma:beta4} below.

The primary challenge in the remaining analysis lies in controlling the term \(\bnabla_{\bdelta} \bar\ell_2(\bgammahat, \bdelta^*)^\top \bDelta_{\bdelta}\), where the gradient \(\bnabla_{\bdelta} \bar\ell_2(\bgammahat, \bdelta^*)\) is evaluated at the target value \(\bdelta^*\) and the estimated value \(\bgammahat\). Unlike standard techniques used for \(\ell_1\)-regularized estimators, model misspecification with \(\rho(\cdot) \neq \rho^*(\cdot)\) prevents us from guaranteeing a sufficiently small upper bound for the infinity norm \(\|\bnabla_{\bdelta} \bar\ell_2(\bgammahat, \bdelta^*)\|_\infty\), as the mean \(\E_{\S_{\bdelta}} \{\bnabla_{\bdelta} \bar\ell_2(\bgammahat, \bdelta^*)\}\) is not necessarily zero. Therefore, instead of directly applying the inequality
\[\lvert\bnabla_{\bdelta} \bar\ell_2(\bgammahat, \bdelta^*)^\top \bDelta_{\bdelta}\rvert \leq \|\bnabla_{\bdelta} \bar\ell_2(\bgammahat, \bdelta^*)\|_\infty \|\bDelta_{\bdelta}\|_1,\]
we consider a different decomposition:
\begin{align*}
&\lvert\bnabla_{\bdelta} \bar\ell_2(\bgammahat, \bdelta^*)^\top \bDelta_{\bdelta}\rvert = \lvert\bnabla_{\bdelta} \bar\ell_2(\bgamma^*, \bdelta^*)^\top \bDelta_{\bdelta} + R_1(\bDelta_{\bdelta})\rvert \\
&\qquad \leq \|\bnabla_{\bdelta} \bar\ell_2(\bgamma^*, \bdelta^*)\|_\infty \|\bDelta_{\bdelta}\|_1 + \lvert R_1(\bDelta_{\bdelta})\rvert,
\end{align*}
where it is ensured that \(\E \{\bnabla_{\bdelta} \bar\ell_2(\bgamma^*, \bdelta^*)\} = \bzero\), even under model misspecification, and we have \(\|\bnabla_{\bdelta} \bar\ell_2(\bgamma^*, \bdelta^*)\|_\infty = O_p(\sqrt{\log d / N})\) as shown in Lemma \ref{lemma:gradient}. 

In the following lemma, we show that the remainder term, defined as 
$$R_1(\bDelta) := \{\bnabla_{\bdelta} \bar\ell_2(\bgammahat, \bdelta^*) - \bnabla_{\bdelta} \bar\ell_2(\bgamma^*, \bdelta^*)\}^\top \bDelta,$$
can be controlled through the imputation error from estimating \(\bgamma^*\), as well as the \(\ell_1\)- and \(\ell_2\)-norms of \(\bDelta_{\bdelta}\).

\begin{lemma}\label{lemma:beta3}
Let Assumptions \ref{cond:basic} and \ref{cond:subG} hold, $s_{\bgamma}=o(N/\log d_1)$, $s_{\bdelta}=o(N/\log d)$, and consider some $\lambda_{\bgamma}\asymp\sqrt{\log d_1/N}$. 
For any $0<t<\kappa_1^2M/(16^2\sigma_{\bdelta}^2s_{\bdelta})$, let $\lambda_{\bdelta}=2\sigma_{\bdelta}\sqrt{(t+\log d)/M}$.
Define
\begin{align}
\mathcal A_1:=&\{\|\bnabla_{\bdelta}\bar\ell_2(\bgamma^*,\bdelta^*)\|_\infty\leq\lambda_{\bdelta}/2\},\label{def:A1}\\
\mathcal A_2:=&\left\{\lvert R_1(\bDelta)\rvert\leq c\sqrt\frac{s_{\bgamma}\log d_1}{N}\left(\frac{\|\bDelta\|_1}{\sqrt{N/\log d}}+\|\bDelta\|_2\right),\;\;\forall\bDelta\in\R^d\right\},\label{def:A2}\\
\mathcal A_3:=&\left\{\delta\bar\ell_2(\bgammahat,\bdelta^*,\bDelta)\geq\kappa_1\|\bDelta\|_2^2-\kappa_2\frac{\log d}{M}\|\bDelta\|_1^2,\;\;\forall\bDelta\in\R^d:\|\bDelta\|_2\leq1\right\}\label{def:A3},
\end{align}
where
$R_1(\bDelta):=\left\{\bnabla_{\bdelta}\bar\ell_2(\bgammahat,\bdelta^*)-\bnabla_{\bdelta}\bar\ell_2(\bgamma^*,\bdelta^*)\right\}^\top\bDelta$ and $c>0$ is some constant. Then, $\P_{\S_{\bgamma}\cup\S_{\bdelta}}(\mathcal A_1\cap\mathcal A_2)\geq1-t-2\exp(-t)$, as long as $N>N_1$ with some large enough constant $N_1>0$. 

Define $\bar s_{\bdelta}:=s_{\bgamma}\log d_1/\log d+s_{\bdelta}$, $\Ctil(s,k):=\{\bDelta\in\R^d:\|\bDelta\|_1\leq k\sqrt{s}\|\bDelta\|_2\}$, and $\Ktil(s,k,1):=\Ctil(s,k)\cap\{\bDelta\in\R^d:\|\bDelta\|_2=1\}$ for any $s,k>0$. Then, on the event $\mathcal A_1\cap\mathcal A_2\cap\mathcal A_3$, for all $\bDelta\in\Ktil(\bar s_{\bdelta},k_0,1)$, we have $\mathcal F(\bDelta)>0$.
\end{lemma}

\begin{proof}[Proof of Lemma \ref{lemma:beta3}.]
Note that 
\begin{align}
\mathcal F(\bDelta):=&\delta\bar\ell_2(\bgammahat,\bdelta^*,\bDelta)+\lambda_{\bdelta}\|\bdelta^*+\bDelta\|_1+\bnabla_{\bdelta}\bar\ell_2(\bgammahat,\bdelta^*)^\top\bDelta-\lambda_{\bdelta}\|\bdelta^*\|_1\nonumber\\
=&\delta\bar\ell_2(\bgammahat,\bdelta^*,\bDelta)+\lambda_{\bdelta}\|\bdelta^*+\bDelta\|_1+\bnabla_{\bdelta}\bar\ell_2(\bgamma^*,\bdelta^*)^\top\bDelta+	R_1(\bDelta)-\lambda_{\bdelta}\|\bdelta^*\|_1,\label{eq:beta1}
\end{align}
where
\begin{align*}
	R_1(\bDelta):&=\left\{\bnabla_{\bdelta}\bar\ell_2(\bgammahat,\bdelta^*)-\bnabla_{\bdelta}\bar\ell_2(\bgamma^*,\bdelta^*)\right\}^\top\bDelta\\
	&=M^{-1}\sum_{i\in\mathcal I_{\bdelta}}A_{1i}\left\{g^{-1}(\bS_{1i}^\top\bgammahat)-g^{-1}(\bS_{1i}^\top\bgamma^*)\right\}\left\{1-A_{2i}g^{-1}(\bar\bS_{2i}^\top\bdelta^*)\right\}\bar\bS_{2i}^\top\bDelta.
\end{align*}
Let $\lambda_{\bdelta}=2\sigma_{\bdelta}\sqrt{(t+\log d)/M}$ with some $t>0$. By Lemma \ref{lemma:gradient}, we have $\P_{\S_{\bdelta}}(\mathcal A_1)\geq1-2\exp(-t)$.		
On the event $\mathcal A_1$, we have $\lvert\bnabla_{\bdelta}\bar\ell_2(\bgamma^*,\bdelta^*)^\top\bDelta\rvert\leq\lambda_{\bdelta}\|\bDelta\|_1/2$. Note that $\|\bdelta^*\|_1=\|\bdelta_{S_{\bdelta}}^*\|_1\leq\|\bdelta_{S_{\bdelta}}^*+\bDelta_{S_{\bdelta}}\|_1+\|\bDelta_{S_{\bdelta}}\|_1$, $\|\bDelta\|_1=\|\bDelta_{S_{\bdelta}}\|_1+\|\bDelta_{S_{\bdelta}^c}\|_1$, and $\|\bdelta^*+\bDelta\|_1=\|\bdelta_{S_{\bdelta}}^*+\bDelta_{S_{\bdelta}}\|_1+\|\bDelta_{S_{\bdelta}^c}\|_1$. Recall the equation \eqref{eq:beta1}. It follows that
\begin{align*}
2\mathcal F(\bDelta)\geq 2\delta\bar\ell_2(\bgammahat,\bdelta^*,\bDelta)+\lambda_{\bdelta}\|\bDelta_{S_{\bdelta}^c}\|_1-3\lambda_{\bdelta}\|\bDelta_{S_{\bdelta}}\|_1-2\lvert R_1(\bDelta)\rvert.
\end{align*}
Hence,
\begin{align}\label{eq:beta5}
2\mathcal F(\bDelta)\geq 2\delta\bar\ell_2(\bgammahat,\bdelta^*,\bDelta)+\lambda_{\bdelta}\|\bDelta\|_1-4\lambda_{\bdelta}\|\bDelta_{S_{\bdelta}}\|_1-2\lvert R_1(\bDelta)\lvert.
\end{align}
Under the overlap condition in Assumption \ref{cond:basic} and since $\lvert A_1\rvert\leq 1$,
\begin{align*}
&\lvert R_1(\bDelta)\rvert\leq(1+c_0^{-1})M^{-1}\sum_{i\in\mathcal I_{\bdelta}}A_{1i}\left\{g^{-1}(\bS_{1i}^\top\bgammahat)-g^{-1}(\bS_{1i}^\top\bgamma^*)\right\}\bar\bS_{2i}^\top\bDelta\\
&\;\;\overset{(i)}{\leq}(1+c_0^{-1})\sqrt{M^{-1}\sum_{i\in\mathcal I_{\bdelta}}\left\{g^{-1}(\bS_{1i}^\top\bgammahat)-g^{-1}(\bS_{1i}^\top\bgamma^*)\right\}^2}\sqrt{M^{-1}\sum_{i\in\mathcal I_{\bdelta}}(\bar\bS_{2i}^\top\bDelta)^2},
\end{align*}
where (i) holds by the Cauchy–Schwarz inequality.	
It follows that
\begin{align*}
&\sup_{\bDelta\in\R^{d}/\{\bzero\}}\frac{\lvert R_1(\bDelta)\rvert}{\|\bDelta\|_1/\sqrt{N/\log d}+\|\bDelta\|_2}\\
&\qquad\leq(1+c_0^{-1})\sqrt{M^{-1}\sum_{i\in\mathcal I_{\bdelta}}\left\{g^{-1}(\bS_{1i}^\top\bgammahat)-g^{-1}(\bS_{1i}^\top\bgamma^*)\right\}^2}\\
&\qquad\qquad\cdot\sqrt{\sup_{\bDelta\in\R^{d}/\{\bzero\}}\frac{M^{-1}\sum_{i\in\mathcal I_{\bdelta}}(\bar\bS_{2i}^\top\bDelta)^2}{\|\bDelta\|_1^2N^{-1}\log d+\|\bDelta\|_2^2}},
\end{align*}
since $(\|\bDelta\|_1/\sqrt{N/\log d}+\|\bDelta\|_2)^2\geq\|\bDelta\|_1^2N^{-1}\log d+\|\bDelta\|_2^2$.
Note that
\begin{align*}
&\E_{\S_{\bdelta}}\left[M^{-1}\sum_{i\in\mathcal I_{\bdelta}}\left\{g^{-1}(\bS_{1i}^\top\bgammahat)-g^{-1}(\bS_{1i}^\top\bgamma^*)\right\}^2\right]\\
&\qquad=\E\left[\left\{g^{-1}(\bS_1^\top\bgammahat)-g^{-1}(\bS_1^\top\bgamma^*)\right\}^2\right]\overset{(i)}{=}O_p\left(\frac{s_{\bgamma}\log d_1}{N}\right),
\end{align*}
where (i) holds by Lemma \ref{lemma:preliminary_alpha}. By Lemma \ref{l1},
$$M^{-1}\sum_{i\in\mathcal I_{\bdelta}}\left\{g^{-1}(\bS_{1i}^\top\bgammahat)-g^{-1}(\bS_{1i}^\top\bgamma^*)\right\}^2=O_p\left(\frac{s_{\bgamma}\log d_1}{N}\right).$$
Besides, by Lemma \ref{lemma:beta1}, we also have
$$\sup_{\bDelta\in\R^{d}/\{\bzero\}}\frac{M^{-1}\sum_{i\in\mathcal I_{\bdelta}}(\bar\bS_{2i}^\top\bDelta)^2}{\|\bDelta\|_1^2N^{-1}\log d+\|\bDelta\|_2^2}=O_p(1).$$
Hence,
\begin{align*}
	\sup_{\bDelta\in\R^{d}/\{\bzero\}}\frac{\lvert R_1(\bDelta)\rvert}{\|\bDelta\|_1/\sqrt{N/\log d}+\|\bDelta\|_2}
	&= O_p\left(\sqrt\frac{s_{\bgamma}\log d_1}{N}\right).
\end{align*}
That is, with any $t>0$, when $N$ is large enough, there exists some constant $c>0$ such that $\P_{\S_{\bgamma}\cup\S_{\bdelta}}(\mathcal A_2)\geq1-t$.
Hence, 
$$\P_{\S_{\bgamma}\cup\S_{\bdelta}}(\mathcal A_1\cap\mathcal A_2)\geq1-t-2\exp(-t). $$
Recall the definitions \eqref{def:A1} and \eqref{def:A2}. Now, conditional on $\mathcal A_1\cap\mathcal A_2$, we have
\begin{align*}
2\mathcal F(\bDelta)&\geq 2\delta\bar\ell_2(\bgammahat,\bdelta^*,\bDelta)+\lambda_{\bdelta}\|\bDelta\|_1-4\lambda_{\bdelta}\|\bDelta_{S_{\bdelta}}\|_1\\
&\qquad-2c\sqrt\frac{s_{\bgamma}\log d_1}{N}\left(\frac{\|\bDelta\|_1}{\sqrt{N/\log d}}+\|\bDelta\|_2\right).
\end{align*}
Since $\lambda_{\bdelta}=2\sigma_{\bdelta}\sqrt{(t+\log d)/M}\geq2\sigma_{\bdelta}\sqrt{\log d/M}$, $M\asymp N$, and $s_{\bgamma}=o(N/\log d_1)$, we have $\sqrt{s_{\bgamma}\log d_1\log d/N^2}=o(\lambda_{\bdelta})$. Hence, with some $N_0>0$, when $N>N_0$, we have $4c\sqrt{s_{\bgamma}\log d_1\log d/N^2}\leq\lambda_{\bdelta}$. 
It follows that
\begin{align*}
4\mathcal F(\bDelta)\geq 4\delta\bar\ell_2(\bgammahat,\bdelta^*,\bDelta)+\lambda_{\bdelta}\|\bDelta\|_1-8\lambda_{\bdelta}\|\bDelta_{S_{\bdelta}}\|_1-4c\sqrt\frac{s_{\bgamma}\log d_1}{N}\|\bDelta\|_2.
\end{align*}
Note that $\|\bDelta_{S_{\bdelta}}\|_1\leq\sqrt{s_{\bdelta}}\|\bDelta_{S_{\bdelta}}\|_2\leq\sqrt{s_{\bdelta}}\|\bDelta\|_2$. 
Hence,
\begin{align}\label{eq:beta3}
4\mathcal F(\bDelta)&\geq4\delta\bar\ell_2(\bgammahat,\bdelta^*,\bDelta)+\lambda_{\bdelta}\|\bDelta\|_1-\left(8\lambda_{\bdelta}\sqrt{s_{\bdelta}}+4c\sqrt\frac{s_{\bgamma}\log d_1}{N}\right)\|\bDelta\|_2.
\end{align}
For any $\bDelta\in\Ktil(\bar s_{\bdelta},k_0,1)$, we have
\begin{align*}
4\mathcal F(\bDelta)\geq 4\delta\bar\ell_2(\bgammahat,\bdelta^*,\bDelta)+\lambda_{\bdelta}\|\bDelta\|_1-\left(8\lambda_{\bdelta}\sqrt{s_{\bdelta}}+4c\sqrt\frac{s_{\bgamma}\log d_1}{N}\right),
\end{align*}
on the event $\mathcal A_1\cap\mathcal A_2$ and when $N>N_0$. Here, on the event $\mathcal A_3$, we have
\begin{align*}
&\delta\bar\ell_2(\bgammahat,\bdelta^*,\bDelta)\geq\kappa_1\|\bDelta\|_2^2-\kappa_2\frac{\log d}{M}\|\bDelta\|_1^2\overset{(i)}{\geq}\kappa_1-\kappa_2k_0^2\frac{\bar s_{\bdelta}\log d}{M},
\end{align*}
where (i) holds since $\bDelta\in\Ktil(\bar s_{\bdelta},k_0,1)$. Therefore, conditional on the event $\mathcal A_1\cap\mathcal A_2\cap\mathcal A_3$, when $N>N_1$ with some constant $N_1>0$,
\begin{align*}
\mathcal F(\bDelta)&\geq\kappa_1-\kappa_2k_0^2\frac{\bar s_{\bdelta}\log d}{M}-2\lambda_{\bdelta}\sqrt{s_{\bdelta}}-\frac{c}{2}\sqrt\frac{s_{\bgamma}\log d_1}{N}\geq\kappa_1/2,
\end{align*}
since as $N\to\infty$, we have $\bar s_{\bdelta}\log d/M=s_{\bgamma}\log d_1/M+s_{\bdelta}\log d/M=o(1)$, $\sqrt{s_{\bgamma}\log d_1/N}=o(1)$, and $2\lambda_{\bdelta}\sqrt{s_{\bdelta}}=4\sigma_{\bdelta}\sqrt{s_{\bdelta}(t+\log d)/M}\leq4\sigma_{\bdelta}\sqrt{s_{\bdelta}t/M}+4\sigma_{\bdelta}\sqrt{s_{\bdelta}\log d/M}\leq\kappa_1/4+o(1)$ when $t<\kappa_1^2M/(16^2\sigma_{\bdelta}^2s_{\bdelta})$.
\end{proof}

Due to the presence of the \(\ell_2\)-norm when controlling the remainder term \(R_1(\bDelta)\), as seen in \eqref{def:A2}, it is no longer suitable to confine our analysis to the usual cone set, typically defined as \(\C(S,k)=\{\|\bDelta\in\R^d : \|\bDelta_{S^c}\|_1 \leq k \|\bDelta_S\|_1\}\). Instead, we focus on a different cone set \(\Ctil(s,k) = \{\bDelta \in \R^d : \|\bDelta\|_1 \leq k \sqrt{s} \|\bDelta\|_2\}\), which incorporates both \(\ell_1\)- and \(\ell_2\)-norms. The following lemma demonstrates that the error term \(\bDelta_{\bdelta} := \bdeltahat - \bdelta^*\) lies within the cone set \(\Ctil(s,k)\) for some appropriately chosen values of \(s\) and \(k\).

\begin{lemma}\label{lemma:beta2}
Let Assumptions \ref{cond:basic} and \ref{cond:subG} hold. Let $s_{\bgamma}=o(N/\log d_1)$ and consider some $\lambda_{\bgamma}\asymp\sqrt{\log d_1/N}$. For any $t>0$, let $\lambda_{\bdelta}=2\sigma_{\bdelta}\sqrt{(t+\log d)/M}$. Events $\mathcal A_1$ and $\mathcal A_2$ are defined in \eqref{def:A1} and \eqref{def:A2}.
Then, on the event $\mathcal A_1\cap\mathcal A_2$, when $N>N_0$,
\begin{align*}
4\delta\bar\ell_2(\bgammahat,\bdelta^*,\bDelta_{\bdelta})+\lambda_{\bdelta}\|\bDelta_{\bdelta}\|_1/2\leq&\left(8\lambda_{\bdelta}\sqrt{s_{\bdelta}}+4c\sqrt\frac{s_{\bgamma}\log d_1}{N}\right)\|\bDelta_{\bdelta}\|_2,\\
\|\bDelta_{\bdelta}\|_1\leq&k_0\sqrt{\bar s_{\bdelta}}\|\bDelta_{\bdelta}\|_2,
\end{align*}
where $N_0,k_0,c>0$ are some constants and $\bar s_{\bdelta}:=s_{\bgamma}\log d_1/\log d+s_{\bdelta}$. That is, $\bDelta_{\bdelta}\in\Ctil(\bar s_{\bdelta},k_0)$.
\end{lemma}

\begin{proof}[Proof of Lemma \ref{lemma:beta2}.]
Based on the construction of $\bdeltahat$, we have
\begin{align*}
\bar\ell_2(\bgammahat,\bdeltahat)+\lambda_{\bdelta}\|\bdeltahat\|_1\leq\bar\ell_2(\bgammahat,\bdelta^*)+\lambda_{\bdelta}\|\bdelta^*\|_1.
\end{align*}
By definition \eqref{delta_ell_2_bar}, we have $\delta\bar\ell_2(\bgammahat,\bdelta^*,\bDelta_{\bdelta})=\bar\ell_2(\bgammahat,\bdeltahat)-\bar\ell_2(\bgammahat,\bdelta^*)-\bnabla_{\bdelta}\bar\ell_2(\bgammahat,\bdelta^*)^\top\bDelta_{\bdelta}$. 
It follows that
\begin{align}
&\mathcal F(\bDelta_{\bdelta})
=\delta\bar\ell_2(\bgammahat,\bdelta^*,\bDelta_{\bdelta})+\lambda_{\bdelta}\|\bdeltahat\|_1+\bnabla_{\bdelta}\bar\ell_2(\bgammahat,\bdelta^*)^\top\bDelta_{\bdelta}-\lambda_{\bdelta}\|\bdelta^*\|_1\label{eq:basic-delta}\\
&\qquad=\delta\bar\ell_2(\bgammahat,\bdelta^*,\bDelta_{\bdelta})+\lambda_{\bdelta}\|\bdelta^*+\bDelta_{\bdelta}\|_1+\bnabla_{\bdelta}\bar\ell_2(\bgamma^*,\bdelta^*)^\top\bDelta_{\bdelta}+R_1(\bDelta_{\bdelta})-\lambda_{\bdelta}\|\bdelta^*\|_1\nonumber\\
&\qquad\leq0,\nonumber
\end{align}
where
\begin{align*}
&R_1(\bDelta_{\bdelta})=\left\{\bnabla_{\bdelta}\bar\ell_2(\bgammahat,\bdelta^*)-\bnabla_{\bdelta}\bar\ell_2(\bgamma^*,\bdelta^*)\right\}^\top\bDelta_{\bdelta}\\
&\qquad=M^{-1}\sum_{i\in\mathcal I_{\bdelta}}A_{1i}\left\{g^{-1}(\bS_{1i}^\top\bgammahat)-g^{-1}(\bS_{1i}^\top\bgamma^*)\right\}\left\{1-A_{2i}g^{-1}(\bar\bS_{2i}^\top\bdelta^*)\right\}\bar\bS_{2i}^\top\bDelta_{\bdelta}.
\end{align*}
Repeat the same procedure in the proof of Lemma \ref{lemma:beta3} for obtaining \eqref{eq:beta5} and \eqref{eq:beta3}. Then, conditional on $\mathcal A_1$, we have
\begin{align}
0\geq2\mathcal F(\bDelta_{\bdelta})\geq 2\delta\bar\ell_2(\bgammahat,\bdelta^*,\bDelta_{\bdelta})+\lambda_{\bdelta}\|\bDelta_{\bdelta}\|_1-4\lambda_{\bdelta}\|\bDelta_{\bdelta,S_{\bdelta}}\|_1-2\lvert R_1(\bDelta_{\bdelta})\rvert\label{eq:beta2}.
\end{align}
Conditional on $\mathcal A_1\cap\mathcal A_2$, we further have
\begin{align*}
0\geq 4\mathcal F(\bDelta_{\bdelta})&\geq4\delta\bar\ell_2(\bgammahat,\bdelta^*,\bDelta_{\bdelta})+\left(\lambda_{\bdelta}-4c\sqrt{\frac{s_{\bgamma}\log d_1}{N}\frac{\log d}{N}}\right)\|\bDelta_{\bdelta}\|_1\\
&\qquad-\left(8\lambda_{\bdelta}\sqrt{s_{\bdelta}}+4c\sqrt\frac{s_{\bgamma}\log d_1}{N}\right)\|\bDelta_{\bdelta}\|_2.
\end{align*}
Since $s_{\bgamma}=o(N/\log d_1)$ and $\lambda_{\bdelta}=2\sigma_{\bdelta}\sqrt{(t+\log d)/M}\geq2\sigma_{\bdelta}\sqrt{\log d/M}\asymp\sqrt{\log d/N}$, we have $4c\sqrt{\frac{s_{\bgamma}\log d_1}{N}\frac{\log d}{N}}\leq\lambda_{\bdelta}/2$ as long as $N>N_0$ with some large enough $N_0>0$. Therefore,
\begin{align*}
4\delta\bar\ell_2(\bgammahat,\bdelta^*,\bDelta_{\bdelta})+\lambda_{\bdelta}\|\bDelta_{\bdelta}\|_1/2\leq&\left(8\lambda_{\bdelta}\sqrt{s_{\bdelta}}+4c\sqrt\frac{s_{\bgamma}\log d_1}{N}\right)\|\bDelta_{\bdelta}\|_2.
\end{align*}
Recall the equation \eqref{deltaell_beta}. We have $\delta\bar\ell_2(\bgammahat,\bdelta^*,\bDelta_{\bdelta})\geq0$. Since $\lambda_{\bdelta}\geq2\sigma_{\bdelta}\sqrt{\log d/M}\asymp\sqrt{\log d/N}$, there exists some constant $k_0>0$, such that
$$\|\bDelta_{\bdelta}\|_1\leq k_0\sqrt{\frac{s_{\bgamma}\log d_1}{\log d}+s_{\bdelta}}\|\bDelta_{\bdelta}\|_2=k_0\sqrt{\bar s_{\bdelta}}\|\bDelta_{\bdelta}\|_2,$$
on $\mathcal A_1\cap\mathcal A_2$ and when $N>N_0$ with some $N_0>0$.
\end{proof}

While the RSC property \eqref{eq:RSC4} developed in Lemma \ref{lemma:RSC} applies only to the set where \(\|\bDelta\|_2 \leq 1\), the following lemma demonstrates that \(\bDelta_{\bdelta}\) satisfies this condition, based on the results from Lemma \ref{lemma:beta3}.

\begin{lemma}\label{lemma:beta4}
Let the assumptions in Lemma \ref{lemma:beta3} hold and also that $\bDelta_{\bdelta}\in\Ctil(\bar s_{\bdelta},k_0)$. Then, on the event $\mathcal A_1\cap\mathcal A_2\cap\mathcal A_3$, we have $\|\bDelta_{\bdelta}\|_2\leq1$.
\end{lemma}

\begin{proof}[Proof of Lemma \ref{lemma:beta4}.]
We prove by contradiction. Suppose that $\|\bDelta_{\bdelta}\|_2>1$. Let $\bDeltatil=\bDelta_{\bdelta}/\|\bDelta_{\bdelta}\|_2$. Then $\|\bDeltatil\|_2=1$. When $\bDelta_{\bdelta}\in\Ctil(\bar s_{\bdelta},k_0)$, we have 
$$\|\bDeltatil\|_1=\|\bDelta_{\bdelta}\|_1/\|\bDelta_{\bdelta}\|_2\leq k_0\sqrt{\bar s_{\bdelta}}=k_0\sqrt{\bar s_{\bdelta}}\|\bDeltatil\|_2.$$
That is, $\bDeltatil\in\Ctil(\bar s_{\bdelta},k_0)$, and hence $\bDeltatil\in\Ktil(\bar s_{\bdelta},k_0,1)$. Let $u=\|\bDelta_{\bdelta}\|_2^{-1}$. Then $0<u<1$. Note that $\mathcal F(\cdot)$ is a convex function. Hence, when $N>N_1$,
\begin{align*}
\mathcal F(\bDeltatil)=\mathcal F(u\bDelta_{\bdelta}+(1-u)\bzero)\leq u\mathcal F(\bDelta_{\bdelta})+(1-u)\mathcal F(\bzero)\overset{(i)}{=}u\mathcal F(\bDelta_{\bdelta})\overset{(ii)}\leq0,
\end{align*}
where (i) holds since $\mathcal F(\bzero)=0$ by construction of $\mathcal F(\cdot)$; (ii) holds by the construction of $\bdeltahat$. However, by Lemma \ref{lemma:beta3}, $\mathcal F(\bDeltatil)>0$. Thus, we conclude that $\|\bDelta_{\bdelta}\|_2\leq1$.
\end{proof}

We are now ready to present the proofs for the asymptotic results of the nuisance estimates under potential model misspecification.

\begin{proof}[Proof of Theorem \ref{thm:nuisance}.]
We prove the consistency rates of the nuisance parameter estimators when the models are possibly misspecified.

(a) By Lemmas \ref{lemma:RSC} and \ref{lemma:gradient}, as well as Corollary 9.20 of \cite{wainwright2019high}, we have
$$\|\bgammahat-\bgamma^*\|_2=O_p\left(\sqrt\frac{s_{\bgamma}\log d_1}{M}\right),\quad\|\bgammahat-\bgamma^*\|_1=O_p\left(s_{\bgamma}\sqrt\frac{\log d_1}{M}\right).$$

\vspace{0.5em}

(b) By Lemma \ref{lemma:beta3}, $\P_{\S_{\bgamma}\cup\S_{\bdelta}}(\mathcal A_1\cap\mathcal A_2)\geq1-t-2\exp(-t)$, where $\mathcal A_1$ and $\mathcal A_2$ are defined in \eqref{def:A1} and \eqref{def:A2}, respectively. By Lemma \ref{lemma:beta2}, conditional on $\mathcal A_1\cap\mathcal A_2$, we have $\bDelta_{\bdelta}=\bdeltahat-\bdelta^*\in\Ctil(\bar s_{\bdelta},k_0)=\{\bDelta\in\R^d:\|\bDelta\|_1\leq k_0\sqrt{\bar s_{\bdelta}}\|\bDelta\|_2\}$, where $\bar s_{\bdelta}=s_{\bgamma}\log d_1/\log d+s_{\bdelta}$ and $k_0>0$ is a constant. Additionally, by Lemma \ref{lemma:beta4}, we also have $\|\bDelta_{\bdelta}\|_2\leq1$. By (a), we have $\P_{\S_{\bgamma}}(\{\|\bgammahat-\bgamma^*\|_2\leq1\})=1-o(1)$. Then, by \eqref{eq:RSC2} in Lemma \ref{lemma:RSC}, $\P_{\S_{\bgamma}\cup\S_{\bdelta}}(\mathcal A_3)\geq1-o(1)-c_1\exp(-c_2M)=1-o(1)$, where $\mathcal A_3$ is defined in \eqref{def:A3}. Now, also condition on $\mathcal A_3$. Then we have, for large enough $N$,
\begin{align*}
&\left(2\lambda_{\bdelta}\sqrt{s_{\bdelta}}+c\sqrt\frac{s_{\bgamma}\log d_1}{N}\right)\|\bDelta_{\bdelta}\|_2\overset{(i)}{\geq}\delta\bar\ell_2(\bgammahat,\bdelta^*,\bDelta_{\bdelta})+\frac{\lambda_{\bdelta}}{8}\|\bDelta_{\bdelta}\|_1\\
&\qquad\overset{(ii)}{\geq}\kappa_1\|\bDelta_{\bdelta}\|_2^2-\kappa_2\frac{\log d}{M}\|\bDelta_{\bdelta}\|_1^2+\frac{\lambda_{\bdelta}}{8}\|\bDelta_{\bdelta}\|_1\\
&\qquad\overset{(iii)}{\geq}\left(\kappa_1-\kappa_2k_0^2\frac{\bar s_{\bdelta}\log d}{M}\right)\|\bDelta_{\bdelta}\|_2^2\overset{(iv)}{\geq}\frac{\kappa_1}{2}\|\bDelta_{\bdelta}\|_2^2,
\end{align*}
where (i) holds by Lemma \ref{lemma:beta2}; (ii) holds by the construction of $\mathcal A_3$ and also that $\|\bDelta_{\bdelta}\|_2\leq1$; (iii) holds since $\bDelta_{\bdelta}\in\Ctil(\bar s_{\bdelta},k_0)$ and $\lambda_{\bdelta}\|\bDelta_{\bdelta}\|_1/8\geq0$; (iv) holds for large enough $N$, since $\bar s_{\bdelta}\log d/M=s_{\bgamma}\log d_1/M+s_{\bdelta}\log d/M=o(1)$. Therefore, conditional on $\mathcal A_1\cap\mathcal A_2\cap\mathcal A_3$,
$$\|\bDelta_{\bdelta}\|_2\leq\frac{4\lambda_{\bdelta}\sqrt{s_{\bdelta}}}{\kappa_1}+\frac{2c}{\kappa_1}\sqrt\frac{s_{\bgamma}\log d_1}{N}=O\left(\sqrt\frac{s_{\bgamma}\log d_1+s_{\bdelta}\log d}{N}\right),$$
with some $\lambda_{\bdelta}=2\sigma_{\bdelta}\sqrt{(t+\log d)/M}\asymp\sqrt{\log d/N}$. Since $\bDelta_{\bdelta}\in\Ctil(\bar s_{\bdelta},k_0)$, it follows that
$$\|\bDelta_{\bdelta}\|_1\leq k_0\sqrt{\bar s_{\bdelta}}\|\bDelta_{\bdelta}\|_2=O\left(s_{\bgamma}\sqrt\frac{(\log d_1)^2}{N\log d}+s_{\bdelta}\sqrt\frac{\log d}{N}\right).$$
Therefore, we conclude that
\begin{align*}
\|\bdeltahat-\bdelta^*\|_2=&O_p\left(\sqrt\frac{s_{\bgamma}\log d_1+s_{\bdelta}\log d}{N}\right),\\
\|\bdeltahat-\bdelta^*\|_1=&O_p\left(s_{\bgamma}\sqrt\frac{(\log d_1)^2}{N\log d}+s_{\bdelta}\sqrt\frac{\log d}{N}\right).
\end{align*}

\vspace{0.5em}

(c) For any $t>0$, let $\lambda_{\balpha}=2\sigma_{\balpha}\{2\sqrt{(t+\log d)/M}+(t+\log d)/M\}$. Choose some $\lambda_{\bgamma}\asymp\sqrt{\log d_1/N}$ and $\lambda_{\bdelta}\asymp\sqrt{\log d/N}$. Define
\begin{align}
\mathcal A_4:=&\{\|\bnabla_{\balpha}\bar\ell_3(\bgamma^*,\bdelta^*,\balpha^*)\|_\infty \leq\lambda_{\balpha}/2\},\label{def:A4}\\
\mathcal A_5:=&\left\{\delta\bar\ell_3(\bgammahat,\bdeltahat,\balpha^*,\bDelta)\geq\kappa_1\|\bDelta\|_2^2-\kappa_2\frac{\log d}{M}\|\bDelta\|_1^2,\;\;\forall\bDelta\in\R^d\right\}.\label{def:A5}
\end{align}
By Lemma \ref{lemma:gradient}, we have $\P_{\S_{\balpha}}(\mathcal A_4)\geq1-2\exp(-t)$. Let $\bDelta=\balphahat-\balpha^*$.
Similar to the proof of Lemma \ref{lemma:beta2} that leads to \eqref{eq:beta2}, we have on the event $\mathcal A_4$,
\begin{align}\label{eq:gamma1}
2\delta\bar\ell_3(\bgammahat,\bdeltahat,\balpha^*,\bDelta)+\lambda_{\balpha}\|\bDelta\|_1\leq4\lambda_{\balpha}\|\bDelta_{S_{\balpha}}\|_1+2\lvert R_2\rvert,
\end{align}
where
\begin{align*}
R_2&=\left\{\bnabla_{\balpha}\bar\ell_3(\bgammahat,\bdeltahat, \balpha^*)-\bnabla_{\balpha}\bar\ell_3(\bgamma^*,\bdelta^*,\balpha^*)\right\}^\top\bDelta\\
&=2M^{-1}\sum_{i\in\mathcal I_{\balpha}}A_{1i}A_{2i}\left\{\frac{\exp(-\bar\bS_{2i}^\top\bdeltahat)}{g(\bS_{1i}^\top\bgammahat)}-\frac{\exp(-\bar\bS_{2i}^\top\bdelta^*)}{g(\bS_{1i}^\top\bgamma^*)}\right\}\varepsilon_i\bar\bS_{2i}^\top\bDelta.
\end{align*}
By the fact that $2ab\leq a^2/2+2b^2$, 
\begin{align*}	
\lvert R_2\rvert&\leq \frac{1}{2} \delta\bar\ell_3(\bgammahat,\bdeltahat,\balpha^*,\bDelta)+2R_3,
\end{align*}
where 
\begin{align*}
R_3=M^{-1}\sum_{i\in\mathcal I_{\balpha}}\left(\frac{\exp(-\bar\bS_{2i}^\top\bdeltahat)}{g(\bS_{1i}^\top\bgammahat)}-\frac{\exp(-\bar\bS_{2i}^\top\bdelta^*)}{g(\bS_{1i}^\top\bgamma^*)}\right)^2\frac{g(\bS_{1i}^\top\bgammahat)}{\exp(-\bar\bS_{2i}^\top\bdeltahat)}\varepsilon_i^2.
\end{align*}
By (a) and (b) of Theorem \ref{thm:nuisance}, we have $\P_{\S_{\bgamma}\cup\S_{\bdelta}}(\{\|\bgammahat-\bgamma^*\|_2\leq1,\|\bdeltahat-\bdelta^*\|_2\leq1\})=1-o(1)$.
Note that
\begin{align*}
\E_{\S_{\balpha}}[R_3]
&=\E\left[\left(\frac{\exp(-\bar\bS_2^\top\bdeltahat)}{g(\bS_1^\top\bgammahat)}-\frac{\exp(-\bar\bS_2^\top\bdelta^*)}{g(\bS_1^\top\bgamma^*)}\right)^2\frac{g(\bS_1^\top\bgammahat)}{\exp(-\bar\bS_2^\top\bdeltahat)}\varepsilon^2\right]\\
&\leq\left\|\frac{\exp(-\bar\bS_2^\top\bdeltahat)}{g(\bS_1^\top\bgammahat)}-\frac{\exp(-\bar\bS_2^\top\bdelta^*)}{g(\bS_1^\top\bgamma^*)}\right\|_{\P,6}^2\left\|\frac{g(\bS_{1i}^\top\bgammahat)}{\exp(-\bar\bS_{2i}^\top\bdeltahat)}\right\|_{\P,3}\|\varepsilon\|_{\P,6}^2\\
&\overset{(i)}{=}O_p\left(\frac{s_{\bgamma}\log d_1+s_{\bdelta}\log d}{N}\right).
\end{align*}
where (i) holds by Lemma \ref{lemma:psi2norm}, as well as the fact that
\begin{align}
&\left\|\frac{\exp(-\bar\bS_2^\top\bdeltahat)}{g(\bS_1^\top\bgammahat)}-\frac{\exp(-\bar\bS_2^\top\bdelta^*)}{g(\bS_1^\top\bgamma^*)}\right\|_{\P,6}\nonumber\\
&\qquad\leq\left\|g^{-1}(\bS_1^\top\bgamma^*)\left\{\exp(-\bar\bS_2^\top\bdeltahat)-\exp(-\bar\bS_2^\top\bdelta^*)\right\}\right\|_{\P,6}\nonumber\\
&\qquad\qquad+\left\|\exp(-\bar\bS_2^\top\bdeltahat)\left\{g^{-1}(-\bS_1^\top\bgammahat)-g^{-1}(-\bS_1^\top\bgamma^*)\right\}\right\|_{\P,6}\nonumber\\
&\qquad=O_p\left(\sqrt\frac{s_{\bgamma}\log d_1+s_{\bdelta}\log d}{N}\right)\label{rate:alpha+beta}
\end{align}
using Minkowski inequality, (generalized) H\"older's inequality, and Lemmas \ref{lemma:preliminary_alpha} and \ref{lemma:preliminary_beta}.
Hence,
\begin{align}\label{eq:gamma2}
R_3=O_p\left(\frac{s_{\bgamma}\log d_1+s_{\bdelta}\log d}{N}\right).
\end{align}
Because of the inequality \eqref{eq:gamma1}, we have
\begin{align*}
\delta\bar\ell_3(\bgammahat,\bdeltahat,\balpha^*,\bDelta)+\lambda_{\balpha}\|\bDelta\|_1\leq4\lambda_{\balpha}\|\bDelta_{S_{\balpha}}\|_1+4R_3.
\end{align*}
Note that $\|\bDelta_{S_{\balpha}}\|_1\leq\sqrt{s_{\balpha}}\|\bDelta_{S_{\balpha}}\|_2\leq\sqrt{s_{\balpha}}\|\bDelta\|_2$. Hence,
\begin{align*}
\delta\bar\ell_3(\bgammahat,\bdeltahat,\balpha^*,\bDelta)+\lambda_{\balpha}\|\bDelta\|_1\leq4\lambda_{\balpha}\sqrt{s_{\balpha}}\|\bDelta\|_2+4R_3.
\end{align*}
Recall the equation \eqref{deltaell_gamma}. We have $\delta\bar\ell_3(\bgammahat,\bdeltahat,\balpha^*,\bDelta)\geq0$. Then
\begin{align}\label{eq:gamma3}
\|\bDelta\|_1\leq4\sqrt{s_{\balpha}}\|\bDelta\|_2+\frac{4R_3}{\lambda_{\balpha}}.
\end{align}
Then, by Lemma \ref{lemma:RSC}, $\P_{\S_{\bgamma}\cup\S_{\bdelta}\cup\S_{\balpha}}(\mathcal A_5)\geq1-o(1)-c_1\exp(-c_2M)=1-o(1)$, where $\mathcal A_5$ is defined in \eqref{def:A5}. 
Now, conditional on $\mathcal A_4\cap\mathcal A_5$, for large enough $N$,
\begin{align*} 
4\lambda_{\balpha}\sqrt{s_{\balpha}}\|\bDelta\|_2+4R_3
&\overset{(i)}{\geq}\delta\bar\ell_3(\bgammahat,\bdeltahat,\balpha^*,\bDelta)
\overset{(ii)}{\geq} \kappa_1\|\bDelta\|_2^2-\kappa_2\frac{\log d}{M}\|\bDelta\|_1^2\\
&\overset{(iii)}{\geq} \kappa_1\|\bDelta\|_2^2-2\kappa_2\frac{\log d}{M}\left(16s_{\balpha}\|\bDelta\|_2^2+\frac{16R_3^2}{\lambda_{\balpha}^2}\right)\\
&\overset{(iv)}{\geq} \frac{\kappa_1}{2}\|\bDelta\|_2^2-32\kappa_2R_3^2\frac{ \log d }{M\lambda_{\balpha}^2},
\end{align*}
where (i) holds by $\|\bDelta\|_1 \geq0$; (ii) holds by the construction of $\mathcal A_5$; (iii) holds by \eqref{eq:gamma3} and the fact that $(a+b)^2\leq2a^2+2b^2$; (iv) holds for large enough $N$, since $s_{\balpha}\log d/M =o(1)$. Hence, on the event $\mathcal A_4\cap\mathcal A_5$, for large enough $N$,
\begin{align*} 
\kappa_1\|\bDelta\|_2^2-8\lambda_{\balpha}\sqrt{s_{\balpha}}\|\bDelta\|_2-64\kappa_2R_3^2\frac{ \log d }{M\lambda_{\balpha}^2}-8R_3\leq0.
\end{align*}
Choose some $\lambda_{\balpha}=2\sigma_{\balpha}\{2\sqrt{(t+\log d)/M}+(t+\log d)/M\}\asymp\sqrt{\log d/N}$. It follows from Lemma \ref{lemma:sol} that
\begin{align*} 
\|\bDelta\|_2
&\leq \frac{8\lambda_{\balpha}\sqrt{s_{\balpha}}}{\kappa_1}+\sqrt{64R_3^2\frac{ \kappa_2\log d }{\kappa_1 M\lambda_{\balpha}^2}+\frac{8R_3}{\kappa_1}}\\
&\overset{(i)}{=}O_p\left(\sqrt{\frac{s_{\balpha}\log d}{N}}+\frac{s_{\bgamma}\log d_1+s_{\bdelta}\log d}{N}+\sqrt\frac{s_{\bgamma}\log d_1+s_{\bdelta}\log d}{N} \right) \\
&=O_p\left(\sqrt\frac{s_{\bgamma}\log d_1+s_{\bdelta}\log d+s_{\balpha}\log d}{N}\right),
\end{align*}
where (i) holds by $\lambda_{\balpha}\sqrt{s_{\balpha}}\asymp\sqrt{s_{\balpha}\log d/N}$ and \eqref{eq:gamma2}.
Recall the inequality \eqref{eq:gamma3}. We have 
\begin{align*}
\|\bDelta\|_1=O_p\left(s_{\bgamma}\sqrt\frac{(\log d_1)^2}{N\log d}+s_{\bdelta}\sqrt\frac{\log d}{N}+
s_{\balpha}\sqrt\frac{\log d}{N}\right).
\end{align*}

\vspace{0.5em}

(d) For any $t>0$, let $\lambda_{\bbeta}=2\sigma_{\bbeta}\{2\sqrt{(t+\log d_1)/M}+(t+\log d_1)/M\}$. Choose some $\lambda_{\bgamma}\asymp\sqrt{\log d_1/N}$, $\lambda_{\bdelta}\asymp\sqrt{\log d/N}$, and $\lambda_{\balpha}\asymp\sqrt{\log d/N}$. Define
\begin{align}
\mathcal A_6:=&\left\{\left\|\bnabla_{\bbeta}\bar\ell_4(\bgamma^*,\bdelta^*,\balpha^*,\bbeta^*)\right\|_\infty \leq\lambda_{\bbeta}/2\right\},\label{def:A6}\\
\mathcal A_7:=&\left\{\delta\bar\ell_4(\bgammahat,\bdeltahat,\balphahat,\bbeta^*,\bDelta)\geq \kappa_1\|\bDelta\|_2^2-\kappa_2\frac{\log d_1}{M}\|\bDelta\|_1^2,\;\;\forall\bDelta\in\R^{d_1}\right\}.\label{def:A7}
\end{align}

By Lemma \ref{lemma:gradient}, we have $\P_{\S_{\balpha}}(\mathcal A_6)\geq1-2\exp(-t)$. Let $\bDelta=\bbetahat-\bbeta^*$.
Similar to the proof of Lemma \ref{lemma:beta2} for obtaining \eqref{eq:beta2}, we have, on the event $\mathcal A_6$,
\begin{align}\label{eq:delta1}
2\delta\bar\ell_4(\bgammahat,\bdeltahat,\balphahat,\bbeta^*,\bDelta)+\lambda_{\bbeta}\|\bDelta\|_1\leq4\lambda_{\bbeta}\|\bDelta_{S_{\bbeta}}\|_1+2\lvert R_4\rvert.
\end{align}
where
\begin{align*}
R_4&=\left\{\bnabla_{\bbeta}\bar\ell_4(\bgammahat,\bdeltahat, \balphahat,\bbeta^*)-\bnabla_{\bbeta}\bar\ell_4(\bgamma^*,\bdelta^*,\balpha^*,\bbeta^*)\right\}^\top\bDelta\\
&=2M^{-1}\sum_{i\in\mathcal I_{\bbeta}}A_{1i}\biggl\{\exp(-\bS_{1i}^\top\bgammahat)\left(\bar\bS_{2i}^\top\balphahat-\bS_{1i}^\top\bbeta^*+\frac{A_{2i}(Y_i-\bar\bS_{2i}^\top\balphahat)}{g(\bar\bS_{2i}^\top\bdeltahat)}\right)\\
&\qquad-\exp(-\bS_{1i}^\top\bgamma^*)\left(\bar\bS_{2i}^\top\balpha^*-\bS_{1i}^\top\bbeta^*+\frac{A_{2i}(Y_i-\bar\bS_{2i}^\top\balpha^*)}{g(\bar\bS_{2i}^\top\bdelta^*)}\right)\biggl\}\bS_{1i}^\top\bDelta.
\end{align*}
By the fact that $2ab\leq a^2/2+2b^2$, 
\begin{align*}
\lvert R_4\rvert&\leq \frac{1}{2} \delta\bar\ell_4(\bgammahat,\bdeltahat,\balphahat,\bbeta^*,\bDelta)+2R_5,
\end{align*}
where 
\begin{align*}
R_5&=M^{-1}\sum_{i\in\mathcal I_{\bbeta}}\frac{1}{\exp(-\bS_{1i}^\top\bgammahat)}\biggl\{\exp(-\bS_{1i}^\top\bgammahat)\left(\bar\bS_{2i}^\top\balphahat-\bS_{1i}^\top\bbeta^*+\frac{A_{2i}(Y_i-\bar\bS_{2i}^\top\balphahat)}{g(\bar\bS_{2i}^\top\bdeltahat)}\right)\\
&\qquad-\exp(-\bS_{1i}^\top\bgamma^*)\left(\bar\bS_{2i}^\top\balpha^*-\bS_{1i}^\top\bbeta^*+\frac{A_{2i}(Y_i-\bar\bS_{2i}^\top\balpha^*)}{g(\bar\bS_{2i}^\top\bdelta^*)}\right)\biggl\}^2.
\end{align*}
Note that
\begin{align*}
\E_{\S_{\bbeta}}[R_5]
=\E\left[\frac{1}{\exp(-\bS_1^\top\bgammahat)}(Q_1+Q_2+Q_3)^2 \right],
\end{align*}
where 
\begin{align*}
Q_1&=\exp(-\bS_1^\top\bgammahat)\left( 1-\frac{A_2}{g(\bar\bS_2^\top\bdeltahat)}\right)\bar\bS_2^\top(\balphahat-\balpha^*),\\
Q_2&=\left\{\exp(-\bS_1^\top\bgammahat)-\exp(-\bS_1^\top\bgamma^*)\right\}\zeta,\\
Q_3&=B\left\{\frac{\exp(-\bar\bS_2^\top\bgammahat)}{g(\bS_1^\top\bdeltahat)}-\frac{\exp(-\bar\bS_2^\top\bgamma^*)}{g(\bS_1^\top\bdelta^*)}\right\}\varepsilon.
\end{align*}
By (a) and (b) of Theorem \ref{thm:nuisance}, we have $\P_{\S_{\bgamma}\cup\S_{\bdelta}}(\{\|\bgammahat-\bgamma^*\|_2\leq1,\|\bdeltahat-\bdelta^*\|_2\leq1\})=1-o(1)$. Then by Hölder's inequality,
\begin{align*}
\E_{\S_{\bbeta}}[R_5]
&\leq\left\|\frac{1}{\exp(-\bS_1^\top\bgammahat)}\right\|_{\P,2}\left(\|Q_1\|_{\P,4}+\|Q_2\|_{\P,4}+\|Q_3\|_{\P,4}\right)^2
\end{align*}
and
\begin{align*}
\|Q_1\|_{\P,4}&\leq\left\|\exp(-\bS_1^\top\bgammahat)\right\|_{\P,12}\left\|\left( 1-\frac{A_2}{g(\bar\bS_2^\top\bdeltahat)}\right)\right\|_{\P,12}\left\|\bar\bS_2^\top(\balphahat-\balpha^*)\right\|_{\P,12}\\
&\overset{(i)}{=}O_p\left(\sqrt\frac{s_{\bgamma}\log d_1+s_{\bdelta}\log d+s_{\balpha}\log d}{N}\right),\\
\|Q_2\|_{\P,4}&\leq\left\|\left\{\exp(-\bS_1^\top\bgammahat)-\exp(-\bS_1^\top\bgamma^*)\right\}\right\|_{\P,8}\|\zeta\|_{\P,8}\overset{(ii)}{=}O_p\left(\sqrt\frac{s_{\bgamma}\log d_1}{N}\right),\\
\|Q_3\|_{\P,4}&\leq\left\|\frac{\exp(-\bar\bS_2^\top\bgammahat)}{g(\bS_1^\top\bdeltahat)}-\frac{\exp(-\bar\bS_2^\top\bgamma^*)}{g(\bS_1^\top\bdelta^*)}\right\|_{\P,8}\|\varepsilon\|_{\P,8}\\
&\overset{(iii)}{=}O_p\left(\sqrt\frac{s_{\bgamma}\log d_1+s_{\bdelta}\log d}{N}\right),
\end{align*}
where (i) and (ii) hold by Lemmas \ref{lemma:preliminary_alpha}, \ref{lemma:preliminary_beta}, \ref{lemma:preliminary_gamma} and Lemma \ref{lemma:psi2norm}; (iii) holds analogously as in \eqref{rate:alpha+beta}.
Hence,
\begin{align}
R_5=O_p\left(\frac{s_{\bgamma}\log d_1+s_{\bdelta}\log d+s_{\balpha}\log d}{N}\right).
\end{align}
Recall the inequality \eqref{eq:delta1}. We have
\begin{align*}
\delta\bar\ell_4(\bgammahat,\bdeltahat,\balphahat,\bbeta^*,\bDelta)+\lambda_{\bbeta}\|\bDelta\|_1\leq4\lambda_{\bbeta}\|\bDelta_{S_{\bbeta}}\|_1+4R_5.
\end{align*}
Note that $\|\bDelta_{S_{\bbeta}}\|_1\leq\sqrt{s_{\bbeta}}\|\bDelta_{S_{\bbeta}}\|_2\leq\sqrt{s_{\bbeta}}\|\bDelta\|_2$. Hence,
\begin{align*}
\delta\bar\ell_4(\bgammahat,\bdeltahat,\balphahat,\bbeta^*,\bDelta)+\lambda_{\bbeta}\|\bDelta\|_1\leq4\lambda_{\bbeta}\sqrt{s_{\bbeta}}\|\bDelta_{S_{\bbeta}}\|_1+4\lvert R_5\rvert.
\end{align*}
Recall the equation \eqref{deltaell_delta}. We have $\delta\bar\ell_4(\bgammahat,\bdeltahat,\balphahat,\bbeta^*,\bDelta)\geq0$. Then
\begin{align}\label{eq:delta2}
\|\bDelta\|_1\leq4\sqrt{s_{\bbeta}}\|\bDelta\|_2+\frac{4R_5}{\lambda_{\bbeta}}.
\end{align}
By Lemma \ref{lemma:RSC}, $\P_{\S_{\bgamma}\cup\S_{\bdelta}\cup\S_{\bbeta}}(\mathcal A_7)\geq1-o(1)-c_1\exp(-c_2M)=1-o(1)$, where $\mathcal A_7$ is defined in \eqref{def:A7}.
The remaining parts of the proof can be shown analogously as (c) of Theorem \ref{thm:nuisance}.
Now, conditional on $\mathcal A_6\cap\mathcal A_7$, for large enough $N$,
\begin{align*} 
\|\bDelta\|_2
&\leq \frac{8\lambda_{\bbeta}\sqrt{s_{\bbeta}}}{\kappa_1}+\sqrt{64R_3^2\frac{ \kappa_2\log d_1 }{\kappa_1 M\lambda_{\bbeta}^2}+\frac{8R_3}{\kappa_1}}\\
&=O_p\left(\sqrt\frac{s_{\bgamma}\log d_1+s_{\bdelta}\log d+s_{\balpha}\log d+s_{\bbeta}\log d_1}{N}\right).
\end{align*}
with some $\lambda_{\bbeta}=2\sigma_{\bbeta}\{2\sqrt{(t+\log d_1)/M}+(t+\log d_1)/M\}\asymp\sqrt{\log d_1/M}$. Recall the inequality \eqref{eq:delta2}. We have 
\begin{align*}
\|\bDelta\|_1=O_p\left(s_{\bgamma}\sqrt\frac{\log d_1}{N}+s_{\bdelta}\sqrt\frac{(\log d)^2}{N\log d_1}+
s_{\balpha}\sqrt\frac{(\log d)^2}{N\log d_1}+s_{\bbeta}\sqrt\frac{\log d_1}{N}\right).
\end{align*}
\end{proof}

\subsection{Proofs of the results in Section \ref{sec:asymp_nuisance'}}\label{sec:proof_asymp_nuisance'}

Assuming correctly specified models, we control the gradients in Lemma \ref{lemma:gradient'} below (approximately) by the usual rate \(O_p(\sqrt{\log d/N})\) or \(O_p(\sqrt{\log d_1/N})\). Unlike Lemma \ref{lemma:gradient}, we now control the gradients involving the estimated nuisance parameters. For instance, in part (a) of Lemma \ref{lemma:gradient'} below, we can control \(\|\bnabla_{\bdelta}\bar\ell_2(\bgammahat,\bdelta^*)\|_\infty\), and the estimation error of \(\bgammahat\) is negligible as long as \(s_{\bgamma} = O_p(N/(\log d_1 \log d))\).

\begin{lemma}\label{lemma:gradient'}

\textbf{(a)} Let $\rho(\cdot)=\rho^*(\cdot)$. Let the assumptions in part (a) of Theorem \ref{thm:nuisance} hold. Then, as $N,d_1,d_2\to\infty$,
$$\left\|\bnabla_{\bdelta}\bar\ell_2(\bgammahat,\bdelta^*)\right\|_\infty=O_p\left(\left(1+\sqrt\frac{s_{\bgamma}\log d_1\log d}{N}\right)\sqrt\frac{\log d}{N}\right).$$

\textbf{(b)} Let $\nu(\cdot)=\nu^*(\cdot)$. Let the assumptions in part (b) of Theorem \ref{thm:nuisance} hold. Then, as $N,d_1,d_2\to\infty$,
$$\left\|\bnabla_{\balpha}\bar\ell_3(\bgammahat,\bdeltahat,\balpha^*)\right\|_\infty=O_p\left(\left(1+\sqrt\frac{(s_{\bgamma}\log d_1+s_{\bdelta}\log d)\log d}{N}\right)\sqrt\frac{\log d}{N}\right).$$

\textbf{(c)} Let $\nu(\cdot)=\nu^*(\cdot)$ and $\mu(\cdot)=\mu^*(\cdot)$. Let the assumptions in part (c) of Theorem \ref{thm:nuisance} hold. Then, as $N,d_1,d_2\to\infty$,
$$\left\|\bnabla_{\bbeta}\bar\ell_4(\bgammahat,\bdeltahat,\balpha^*,\bbeta^*)\right\|_\infty=O_p\left(\left(1+\sqrt\frac{(s_{\bgamma}\log d_1+s_{\bdelta}\log d)\log d_1}{N}\right)\sqrt\frac{\log d_1}{N}\right).$$

\textbf{(d)} Let $\rho(\cdot)=\rho^*(\cdot)$ and $\mu(\cdot)=\mu^*(\cdot)$. Let the assumptions in part (c) of Theorem \ref{thm:nuisance} hold. Then, as $N,d_1,d_2\to\infty$,
\begin{align*}
&\left\|\bnabla_{\bbeta}\bar\ell_4(\bgammahat,\bdelta^*,\balphahat,\bbeta^*)\right\|_\infty\\
&\qquad=O_p\left(\left(1+\sqrt\frac{(s_{\bgamma}\log d_1+s_{\bdelta}\log d+s_{\balpha}\log d)\log d_1}{N}\right)\sqrt\frac{\log d_1}{N}\right).
\end{align*}

\textbf{(e)} Let $\rho(\cdot)=\rho^*(\cdot)$, $\nu(\cdot)=\nu^*(\cdot)$, and $\mu(\cdot)=\mu^*(\cdot)$. Let the assumptions in part (c) of Theorem \ref{thm:nuisance} hold. Then, as $N,d_1,d_2\to\infty$,
\begin{align*}
\left\|\bnabla_{\bbeta}\bar\ell_4(\bgammahat,\bdelta^*,\balpha^*,\bbeta^*)\right\|_\infty=O_p\left(\left(1+\sqrt\frac{s_{\bgamma}(\log d_1)^2}{N}\right)\sqrt\frac{\log d_1}{N}\right).
\end{align*}
\end{lemma}

\begin{proof}[Proof of Lemma \ref{lemma:gradient'}.]
By Lemma \ref{lemma:gradient}, we have
\begin{align}
\left\|\bnabla_{\bdelta}\bar\ell_2(\bgamma^*,\bdelta^*)\right\|_\infty=&O_p\left(\sqrt\frac{\log d}{N}\right),\label{gradient_beta}\\
\left\|\bnabla_{\balpha}\bar\ell_3(\bgamma^*,\bdelta^*,\balpha^*)\right\|_\infty=&O_p\left(\sqrt\frac{\log d}{N}\right),\label{gradient_gamma}\\
\left\|\bnabla_{\bbeta}\bar\ell_4(\bgamma^*,\bdelta^*,\balpha^*,\bbeta^*)\right\|_\infty=&O_p\left(\sqrt\frac{\log d_1}{N}\right)\label{gradient_delta}.
\end{align}

(a) Let $\rho(\cdot)=\rho^*(\cdot)$. Note that
\begin{align*}
\bnabla_{\bdelta}\bar\ell_2(\bgammahat,\bdelta^*)-\bnabla_{\bdelta}\bar\ell_2(\bgamma^*,\bdelta^*)=M^{-1}\sum_{i\in\mathcal I_{\bdelta}}\bW_{\bdelta,i},
\end{align*}
where
\begin{align*}
\bW_{\bdelta,i}:=A_{1i}\{g^{-1}(\bS_{1i}^\top\bgammahat)-g^{-1}(\bS_{1i}^\top\bgamma^*)\}\{1-A_{2i}g^{-1}(\bar\bS_{2i}^\top\bdelta^*)\}\bar\bS_{2i}.
\end{align*}
Let $\bW_{\bdelta}$ be an indepenent copy of $\bW_{\bdelta,i}$. Then, by the tower rule,
$$\E(\bW_{\bdelta})=\bzero\in\R^d.$$
By Lemma \ref{Nemirovski}, we have
\begin{align*}
&\E_{\S_{\bdelta}}\left(\left\|M^{-1}\sum_{i\in\mathcal I_{\bdelta}}\bW_{\bdelta,i}\right\|_\infty^2\right)\leq M^{-1}(2e\log d-e)\E(\|\bW_{\bdelta}\|_\infty^2)\\
&\qquad\overset{(i)}{\leq}(1+c_0^{-1})M^{-1}(2e\log d-e)\E\left\{\left\lvert g^{-1}(\bS_1^\top\bgammahat)-g^{-1}(\bS_1^\top\bgamma^*)\right\rvert^2\|\bar\bS_2\|_\infty^2\right\}\\
&\qquad\overset{(ii)}{\leq}(1+c_0^{-1})M^{-1}(2e\log d-e)\left\|g^{-1}(\bS_1^\top\bgammahat)-g^{-1}(\bS_1^\top\bgamma^*)\right\|_{\P,4}^2\left\|\|\bar\bS_2\|_\infty\right\|_{\P,4}^2\\
&\qquad\overset{(iii)}{=}O_p\left(\frac{s_{\bgamma}\log d_1(\log d)^2}{N^2}\right),
\end{align*}
where (i) holds since $\lvert A_1\{1-A_2g^{-1}(\bar\bS_2\bdelta^*)\}\rvert\leq1+c_0^{-1}$ almost surely under Assumption \ref{cond:basic}; (ii) holds by H\"older's inequality; (iii) holds by Lemma \ref{lemma:preliminary_alpha} and Lemma \ref{lemma:psi2norm}. By Lemma \ref{l1},
$$\left\|\bnabla_{\bdelta}\bar\ell_2(\bgammahat,\bdelta^*)-\bnabla_{\bdelta}\bar\ell_2(\bgamma^*,\bdelta^*)\right\|_\infty=\left\|M^{-1}\sum_{i\in\mathcal I_{\bdelta}}\bW_{\bdelta,i}\right\|_\infty=O_p\left(\frac{\sqrt{s_{\bgamma}\log d_1}\log d}{N}\right).$$
Together with \eqref{gradient_beta}, we have
$$\left\|\bnabla_{\bdelta}\bar\ell_2(\bgammahat,\bdelta^*)\right\|_\infty=O_p\left(\left(1+\sqrt\frac{s_{\bgamma}\log d_1\log d}{N}\right)\sqrt\frac{\log d}{N}\right).$$
The remaining parts of the proof can be shown analogously as in (a).

\vspace{0.5em}

(b) Let $\nu(\cdot)=\nu^*(\cdot)$. Note that
\begin{align*}
\bnabla_{\balpha}\bar\ell_3(\bgammahat,\bdeltahat,\balpha^*)-\bnabla_{\balpha}\bar\ell_3(\bgamma^*,\bdelta^*,\balpha^*)=M^{-1}\sum_{i\in\mathcal I_{\balpha}}\bW_{\balpha,i},
\end{align*}
where
\begin{align*}
\bW_{\balpha,i}:=-2A_{1i}A_{2i}\{g^{-1}(\bS_{1i}^\top\bgammahat)\exp(-\bar\bS_{2i}^\top\bdeltahat)-g^{-1}(\bS_{1i}^\top\bgamma^*)\exp(-\bar\bS_{2i}^\top\bdelta^*)\}\varepsilon_i\bar\bS_{2i}.
\end{align*}
Let $\bW_{\balpha}$ be an indepenent copy of $\bW_{\balpha,i}$. Then, by the tower rule,
$$\E(\bW_{\balpha})=\bzero\in\R^d.$$
By Lemma \ref{Nemirovski}, we have
\begin{align*}
&\E_{\S_{\balpha}}\left(\left\|M^{-1}\sum_{i\in\mathcal I_{\balpha}}\bW_{\balpha,i}\right\|_\infty^2\right)\leq M^{-1}(2e\log d-e)\E(\|\bW_{\balpha}\|_\infty^2)\\
&\qquad\overset{(i)}{\leq}2M^{-1}(2e\log d-e)\E\left\{\left\lvert\frac{\exp(-\bar\bS_2^\top\bdeltahat)}{g(\bS_1^\top\bgammahat)}-\frac{\exp(-\bar\bS_2^\top\bdelta^*)}{g(\bS_1^\top\bgamma^*)}\right\rvert^2\varepsilon^2\|\bar\bS_2\|_\infty^2\right\}\\
&\qquad\overset{(ii)}{\leq}2M^{-1}(2e\log d-e)\left\|\frac{\exp(-\bar\bS_2^\top\bdeltahat)}{g(\bS_1^\top\bgammahat)}-\frac{\exp(-\bar\bS_2^\top\bdelta^*)}{g(\bS_1^\top\bgamma^*)}\right\|_{\P,6}^2\left\|\varepsilon\|_{\P,6}^2\|\|\bar\bS_2\|_\infty\right\|_{\P,6}^2\\
&\qquad\overset{(iii)}{=}O_p\left(\frac{(s_{\bgamma}\log d_1+s_{\bdelta}\log d)(\log d)^2}{N^2}\right),
\end{align*}
where (i) holds since $\lvert A_1A_2\rvert\leq1$; (ii) holds by H\"older's inequality; (iii) holds by Lemma \ref{lemma:psi2norm} and \eqref{rate:alpha+beta}. By Lemma \ref{l1},
$$\left\|\bnabla_{\balpha}\bar\ell_3(\bgammahat,\bdeltahat,\balpha^*)-\bnabla_{\bdelta}\bar\ell_3(\bgamma^*,\bdelta^*,\balpha^*)\right\|_\infty=O_p\left(\frac{\sqrt{s_{\bgamma}\log d_1}\log d}{N}\right).$$
Together with \eqref{gradient_gamma}, we have
$$\left\|\bnabla_{\balpha}\bar\ell_3(\bgammahat,\bdeltahat,\balpha^*)\right\|_\infty=O_p\left(\left(1+\sqrt\frac{(s_{\bgamma}\log d_1+s_{\bdelta}\log d)\log d}{N}\right)\sqrt\frac{\log d}{N}\right).$$

\vspace{0.5em}

(c) Let $\nu(\cdot)=\nu^*(\cdot)$ and $\mu(\cdot)=\mu^*(\cdot)$. Note that
\begin{align*}
\bnabla_{\bbeta}\bar\ell_4(\bgammahat,\bdeltahat,\balpha^*,\bbeta^*)-\bnabla_{\bbeta}\bar\ell_4(\bgamma^*,\bdelta^*,\balpha^*,\bbeta^*)=M^{-1}\sum_{i\in\mathcal I_{\bbeta}}(\bW_{\bbeta,1,i}+\bW_{\bbeta,2,i}),
\end{align*}
where
\begin{align*}
\bW_{\bbeta,1,i}:=&-2A_{1i}\{\exp(-\bS_{1i}^\top\bgammahat)-\exp(-\bS_{1i}^\top\bgamma^*)\}\zeta_{i}\bS_{1i},\\
\bW_{\bbeta,2,i}:=&-2A_{1i}A_{2i}\{\exp(-\bS_{1i}^\top\bgammahat)g^{-1}(\bar\bS_{2i}^\top\bdeltahat)-\exp(-\bS_{1i}^\top\bgamma^*)g^{-1}(\bar\bS_{2i}^\top\bdelta^*)\}\varepsilon_i\bS_{1i}.
\end{align*}
Let $\bW_{\bbeta,1}$ and $\bW_{\bbeta,2}$ be indepenent copies of $\bW_{\bbeta,1,i}$ and $\bW_{\bbeta,1,i}$, respectively. Then, by the tower rule,
$$\E(\bW_{\bbeta,1})=\E(\bW_{\bbeta,2})=\bzero\in\R^{d_1}.$$
By Lemma \ref{Nemirovski}, we have
\begin{align*}
&\E_{\S_{\bbeta}}\left(\left\|M^{-1}\sum_{i\in\mathcal I_{\bbeta}}\bW_{\bbeta,1,i}\right\|_\infty^2\right)\leq M^{-1}(2e\log d_1-e)\E(\|\bW_{\bbeta,1}\|_\infty^2)\\
&\qquad\overset{(i)}{\leq}4M^{-1}(2e\log d_1-e)\E\left\{\left\lvert\exp(-\bS_1^\top\bgammahat)-\exp(\bS_1^\top\bgamma^*)\right\rvert^2\zeta^2\|\bS_1\|_\infty^2\right\}\\
&\qquad\overset{(ii)}{\leq}^{-1}(2e\log d_1-e)\left\|\exp(-\bS_1^\top\bgammahat)-\exp(\bS_1^\top\bgamma^*)\right\|_{\P,6}^2\left\|\zeta\|_{\P,6}^2\|\|\bS_1\|_\infty\right\|_{\P,6}^2\\
&\qquad\overset{(iii)}{=}O_p\left(\frac{s_{\bgamma}\log d_1(\log d_1)^2}{N^2}\right),
\end{align*}
where (i) holds since $|A_1|\leq1$; (ii) holds by H\"older's inequality; (iii) holds by Lemma \ref{lemma:preliminary_alpha} and Lemma \ref{lemma:psi2norm}. Similarly, by Lemma \ref{Nemirovski}, we also have
\begin{align*}
&\E_{\S_{\bbeta}}\left(\left\|M^{-1}\sum_{i\in\mathcal I_{\bbeta}}\bW_{\bbeta,2,i}\right\|_\infty^2\right)\leq M^{-1}(2e\log d_1-e)\E(\|\bW_{\bbeta,2}\|_\infty^2)\\
&\qquad\overset{(i)}{\leq}4M^{-1}(2e\log d_1-e)\E\left\{\left\lvert\frac{\exp(-\bS_1^\top\bgammahat)}{g(\bar\bS_2^\top\bdeltahat)}-\frac{\exp(-\bS_1^\top\bgamma^*)}{g(\bar\bS_2^\top\bdelta^*)}\right\rvert^2\varepsilon^2\|\bS_1\|_\infty^2\right\}\\
&\qquad\overset{(ii)}{\leq}4M^{-1}(2e\log d_1-e)\left\|\frac{\exp(-\bS_1^\top\bgammahat)}{g(\bar\bS_2^\top\bdeltahat)}-\frac{\exp(-\bS_1^\top\bgamma^*)}{g(\bar\bS_2^\top\bdelta^*)}\right\|_{\P,6}^2\left\|\varepsilon\|_{\P,6}^2\|\|\bS_1\|_\infty\right\|_{\P,6}^2\\
&\qquad\overset{(iii)}{=}O_p\left(\frac{(s_{\bgamma}\log d_1+s_{\bdelta}\log d)(\log d_1)^2}{N^2}\right),
\end{align*}
where (i) holds since $\lvert A_1A_2\rvert\leq1$; (ii) holds by H\"older's inequality; (iii) holds by Lemma \ref{lemma:psi2norm} and, analogously as in \eqref{rate:alpha+beta},
\begin{align*}
\left\|\frac{\exp(-\bS_1^\top\bgammahat)}{g(\bar\bS_2^\top\bdeltahat)}-\frac{\exp(-\bS_1^\top\bgamma^*)}{g(\bar\bS_2^\top\bdelta^*)}\right\|_{\P,6}=O_p\left(\sqrt\frac{s_{\bgamma}\log d_1+s_{\bdelta}\log d}{N}\right).
\end{align*}
Hence, it follows that
\begin{align*}
&\E_{\S_{\bbeta}}\left\{\left\|\bnabla_{\bbeta}\bar\ell_4(\bgammahat,\bdeltahat,\balpha^*,\bbeta^*)-\bnabla_{\balpha}\bar\ell_4(\bgamma^*,\bdelta^*,\balpha^*,\bbeta^*)\right\|_\infty^2\right\}\\
&\qquad=O_p\left(\frac{(s_{\bgamma}\log d_1+s_{\bdelta}\log d)(\log d_1)^2}{N^2}\right).
\end{align*}
By Lemma \ref{l1},
$$\left\|\bnabla_{\bbeta}\bar\ell_4(\bgammahat,\bdeltahat,\balpha^*,\bbeta^*)-\bnabla_{\bbeta}\bar\ell_4(\bgamma^*,\bdelta^*,\balpha^*,\bbeta^*)\right\|_\infty=O_p\left(\frac{\sqrt{s_{\bgamma}\log d_1+s_{\bdelta}\log d}\log d_1}{N}\right).$$
Together with \eqref{gradient_delta}, we have
$$\left\|\bnabla_{\bbeta}\bar\ell_4(\bgammahat,\bdeltahat,\balpha^*,\bbeta^*)\right\|_\infty=O_p\left(\left(1+\sqrt\frac{(s_{\bgamma}\log d_1+s_{\bdelta}\log d)\log d_1}{N}\right)\sqrt\frac{\log d_1}{N}\right).$$

\vspace{0.5em}

(d) Let $\rho(\cdot)=\rho^*(\cdot)$ and $\mu(\cdot)=\mu^*(\cdot)$. Note that
\begin{align*}
\bnabla_{\bbeta}\bar\ell_4(\bgammahat,\bdelta^*,\balphahat,\bbeta^*)-\bnabla_{\bbeta}\bar\ell_4(\bgamma^*,\bdelta^*,\balpha^*,\bbeta^*)=M^{-1}\sum_{i\in\mathcal I_{\bbeta}}(\bW_{\bbeta,3,i}+\bW_{\bbeta,4,i}),
\end{align*}
where
\begin{align}
\bW_{\bbeta,3,i}&:=-2A_{1i}\exp(-\bS_{1i}^\top\bgammahat)\left\{1-\frac{A_{2i}}{g(\bar\bS_{2i}^\top\bdelta^*)}\right\}\bar\bS_{2i}^\top(\balphahat-\balpha^*)\bS_{1i},\nonumber\\
\bW_{\bbeta,4,i}&:=-2A_{1i}\{\exp(-\bS_{1i}^\top\bgammahat)-\exp(-\bS_{1i}^\top\bgamma^*)\}\left\{\frac{A_{2i}}{g(\bar\bS_{2i}^\top\bdelta^*)}\varepsilon_i+\zeta_{i}\right\}\bS_{1i}.\label{def:Wdelta4}
\end{align}

Let $\bW_{\bbeta,3}$ and $\bW_{\bbeta,4}$ be independent copies of $\bW_{\bbeta,3,i}$ and $\bW_{\bbeta,4,i}$, respectively. Then, by the tower rule,
$$\E(\bW_{\bbeta,3})=\E(\bW_{\bbeta,4})=\bzero\in\R^{d_1}.$$
By Lemma \ref{Nemirovski}, we have
\begin{align*}
&\E_{\S_{\bbeta}}\left(\left\|M^{-1}\sum_{i\in\mathcal I_{\bbeta}}\bW_{\bbeta,3,i}\right\|_\infty^2\right)\leq M^{-1}(2e\log d_1-e)\E(\|\bW_{\bbeta,3}\|_\infty^2)\\
&\qquad\overset{(i)}{\leq}(2e\log d_1-e)(1+c_0^{-1})^2\E\left\{\exp(\bS_1^\top\bgammahat)\left\{\bar\bS_2^\top(\balphahat-\balpha^*)\right\}^2\|\bS_1\|_\infty^2\right\}\\
&\qquad\overset{(ii)}{\leq}(2e\log d_1-e)(1+c_0^{-1})^2\left\|\exp(\bS_1^\top\bgammahat)\right\|_{\P,3}^2\left\|\bar\bS_2^\top(\balphahat-\balpha^*)\right\|_{\P,6}^2\left\|\|\bS_1\|_\infty\right\|_{\P,6}^2\\
&\qquad\overset{(iii)}{=}O_p\left(\frac{(s_{\bgamma}\log d_1+s_{\bdelta}\log d+s_{\balpha}\log d)(\log d_1)^2}{N^2}\right),
\end{align*}
where (i) holds since $\lvert A_1\{1-A_2/g(\bar\bS_2^\top\bdelta^*)\}\rvert\leq(1+c_0^{-1})$ under Assumption \ref{cond:basic}; (ii) holds by H\"older's inequality; (iii) holds by Lemmas \ref{lemma:preliminary_alpha}, \ref{lemma:preliminary_gamma}, and Lemma \ref{lemma:psi2norm}. Similarly, we also have
\begin{align}
&\E_{\S_{\bbeta}}\left(\left\|M^{-1}\sum_{i\in\mathcal I_{\bbeta}}\bW_{\bbeta,4,i}\right\|_\infty^2\right)\leq M^{-1}(2e\log d_1-e)\E(\|\bW_{\bbeta,4}\|_\infty^2)\nonumber\\
&\qquad\leq4M^{-1}(2e\log d_1-e)\E\left[\left\{\exp(\bS_1^\top\bgammahat)-\exp(\bS_1^\top\bgamma^*)\right\}^2\left(c_0^{-1}\lvert\varepsilon\rvert+\lvert\zeta\rvert\right)^2\|\bS_1\|_\infty^2\right]\nonumber\\
&\qquad\leq8M^{-1}(2e\log d_1-e)\left\|\exp(\bS_1^\top\bgammahat)-\exp(\bS_1^\top\bgamma^*)\right\|_{\P,6}^2\left(c_0^{-2}\|\varepsilon\|_{\P,6}^2+\|\zeta\|_{\P,6}^2\right)\left\|\|\bS_1\|_\infty\right\|_{\P,6}^2\nonumber\\
&\qquad\overset{(i)}{=}O_p\left(\frac{s_{\bgamma}\log d_1(\log d_1)^2}{N^2}\right),\label{rate:Wdelta4}
\end{align}
where (i) holds by Lemmas \ref{lemma:preliminary_alpha} and Lemma \ref{lemma:psi2norm}. Hence, it follows that
\begin{align*}
&\E_{\S_{\bbeta}}\left\{\left\|\bnabla_{\bbeta}\bar\ell_4(\bgammahat,\bdelta^*,\balphahat,\bbeta^*)-\bnabla_{\bbeta}\bar\ell_4(\bgamma^*,\bdelta^*,\balpha^*,\bbeta^*)\right\|_\infty^2\right\}\\
&\qquad=O_p\left(\frac{(s_{\bgamma}\log d_1+s_{\bdelta}\log d+s_{\balpha}\log d)(\log d_1)^2}{N^2}\right).
\end{align*}
By Lemma \ref{l1},
\begin{align*}
&\left\|\bnabla_{\bbeta}\bar\ell_4(\bgammahat,\bdelta^*,\balphahat,\bbeta^*)-\bnabla_{\bbeta}\bar\ell_4(\bgamma^*,\bdelta^*,\balpha^*,\bbeta^*)\right\|_\infty\\
&\qquad=O_p\left(\frac{\sqrt{s_{\bgamma}\log d_1+s_{\bdelta}\log d+s_{\balpha}\log d}\log d_1}{N}\right).
\end{align*}
Together with \eqref{gradient_delta}, we have
\begin{align*}
&\left\|\bnabla_{\bbeta}\bar\ell_4(\bgammahat,\bdelta^*,\balphahat,\bbeta^*)\right\|_\infty\\
&\qquad=O_p\left(\left(1+\sqrt\frac{(s_{\bgamma}\log d_1+s_{\bdelta}\log d+s_{\balpha}\log d)\log d_1}{N}\right)\sqrt\frac{\log d_1}{N}\right).
\end{align*}

\vspace{0.5em}

(e) Let $\rho(\cdot)=\rho^*(\cdot)$, $\nu(\cdot)=\nu^*(\cdot)$, and $\mu(\cdot)=\mu^*(\cdot)$. Note that
\begin{align*}
\bnabla_{\bbeta}\bar\ell_4(\bgammahat,\bdelta^*,\balpha^*,\bbeta^*)-\bnabla_{\bbeta}\bar\ell_4(\bgamma^*,\bdelta^*,\balpha^*,\bbeta^*)=M^{-1}\sum_{i\in\mathcal I_{\bbeta}}\bW_{\bbeta,4,i},
\end{align*}
where $\bW_{\bbeta,4,i}$ is defined in \eqref{def:Wdelta4}. By \eqref{rate:Wdelta4} and Lemma \ref{l1},
$$\left\|\bnabla_{\bbeta}\bar\ell_4(\bgammahat,\bdelta^*,\balpha^*,\bbeta^*)-\bnabla_{\bbeta}\bar\ell_4(\bgamma^*,\bdelta^*,\balpha^*,\bbeta^*)\right\|_\infty=O_p\left(\frac{\sqrt{s_{\bgamma}\log d_1}\log d_1}{N}\right).$$
Together with \eqref{gradient_delta}, we have
\begin{align*}
\left\|\bnabla_{\bbeta}\bar\ell_4(\bgammahat,\bdelta^*,\balphahat,\bbeta^*)\right\|_\infty=O_p\left(\left(1+\sqrt\frac{s_{\bgamma}(\log d_1)^2}{N}\right)\sqrt\frac{\log d_1}{N}\right).
\end{align*}

\end{proof}

\begin{proof}[Proof of Theorem \ref{thm:nuisance'}.]
We show the consistency rate of the nuisance estimators under correctly specified models.

(a) Let $\rho(\cdot)=\rho^*(\cdot)$. Then, by Lemma \ref{lemma:gradient'}, when $s_{\bgamma}=O(N/(\log d_1\log d))$,
$$\left\|\bnabla_{\bdelta}\bar\ell_2(\bgammahat,\bdelta^*)\right\|_\infty=O_p\left(\sqrt\frac{\log d}{N}\right).$$
By Lemma \ref{lemma:RSC}, we have \eqref{eq:RSC2} when $\|\bgammahat-\bgamma^*\|_2\leq1$. In addition, by Lemma \ref{lemma:preliminary_alpha}, we also have $\P_{\S_{\bgamma}}(\|\bgammahat-\bgamma^*\|_2\leq1)=1-o(1)$. By Corollary 9.20 of \cite{wainwright2019high}, we have
\begin{align*}
\|\bdeltahat-\bdelta^*\|_2=O_p\left(\sqrt\frac{s_{\bdelta}\log d}{N}\right),\quad\|\bdeltahat-\bdelta^*\|_1=O_p\left(s_{\bdelta}\sqrt\frac{\log d}{N}\right).
\end{align*}

\vspace{0.5em}

(b) Let $\nu(\cdot)=\nu^*(\cdot)$. Then, by Lemma \ref{lemma:gradient'}, when $s_{\bgamma}=O(N/(\log d_1\log d))$ and $s_{\bdelta}=O(N/(\log d)^2)$,
$$\left\|\bnabla_{\balpha}\bar\ell_3(\bgammahat,\bdeltahat,\balpha^*)\right\|_\infty=O_p\left(\sqrt\frac{\log d}{N}\right).$$
By Lemma \ref{lemma:RSC}, we have \eqref{eq:RSC3} when $\|\bgammahat-\bgamma^*\|_2\leq1$ and $\|\bdeltahat-\bdelta^*\|_2\leq1$. In addition, by Lemmas \ref{lemma:preliminary_alpha} and \ref{lemma:preliminary_beta}, we also have $\P_{\S_{\bgamma}\cup\S_{\bdelta}}(\|\bgammahat-\bgamma^*\|_2\leq1\cap\|\bdeltahat-\bdelta^*\|_2\leq1)=1-o(1)$. By Corollary 9.20 of \cite{wainwright2019high}, we have
\begin{align*}
\|\balphahat-\balpha^*\|_2=O_p\left(\sqrt\frac{s_{\balpha}\log d}{N}\right),\quad\|\balphahat-\balpha^*\|_1=O_p\left(s_{\balpha}\sqrt\frac{\log d}{N}\right).
\end{align*}

\vspace{0.5em}

(c) Let $\nu(\cdot)=\nu^*(\cdot)$ and $\mu(\bS_1)=\bS_1^\top\bbeta$. Then, by Lemma \ref{lemma:gradient'}, when $s_{\bgamma}=O(N/(\log d_1)^2)$ and $s_{\bdelta}=O(N/(\log d_1\log d))$,
$$\left\|\bnabla_{\bbeta}\bar\ell_4(\bgammahat,\bdeltahat,\balpha^*,\bbeta^*)\right\|_\infty=O_p\left(\sqrt\frac{\log d_1}{N}\right).$$
That is, for any $t>0$, there exists some $\lambda_3\asymp\sqrt{\log d_1/N}$ such that 
$$\mathcal E_3:=\{\|\bnabla_{\bbeta}\bar\ell_4(\bgammahat,\bdeltahat,\balpha^*,\bbeta^*)\|_\infty\leq\lambda_3\}$$
holds with probability at least $1-t$. Condition on the event $\mathcal E_3$, and choose some $\lambda_{\bbeta}>2\lambda_3$. By the construction of $\bbeta$, we have
$$\bar\ell_4(\bgammahat,\bdeltahat,\balphahat,\bbetahat)+\lambda_{\bbeta}\|\bbetahat\|_1\leq\bar\ell_4(\bgammahat,\bdeltahat,\balphahat,\bbeta^*)+\lambda_{\bbeta}\|\bbeta^*\|_1.$$
Let $\bDelta=\bbetahat-\bbeta^*$ and $S=\{j\in\{1,\dots,d_1\}:\bbeta_j^*\neq0\}$. Note that, 
$$\delta\bar\ell_4(\bgammahat,\bdeltahat,\balphahat,\bbeta^*,\bDelta)=\bar\ell_4(\bgammahat,\bdeltahat,\balphahat,\bbetahat)-\bar\ell_4(\bgammahat,\bdeltahat,\balphahat,\bbeta^*)-\bnabla_{\bbeta}\bar\ell_4(\bgammahat,\bdeltahat,\balphahat,\bbeta^*)^\top\bDelta.$$
Hence,
\begin{align*}
&\delta\bar\ell_4(\bgammahat,\bdeltahat,\balphahat,\bbeta^*,\bDelta)+\lambda_{\bbeta}\|\bbetahat\|_1\leq-\bnabla_{\bbeta}\bar\ell_4(\bgammahat,\bdeltahat,\balphahat,\bbeta^*)^\top\bDelta+\lambda_{\bbeta}\|\bbeta^*\|_1\\
&\quad\leq\left\|\bnabla_{\bbeta}\bar\ell_4(\bgammahat,\bdeltahat,\balpha^*,\bbeta^*)\right\|_\infty\|\bDelta\|_1+\lambda_{\bbeta}\|\bbeta^*\|_1+\lvert R_6\rvert\leq\lambda_{\bbeta}\|\bDelta\|_1/2+\lambda_{\bbeta}\|\bbeta^*\|_1+\lvert R_6\rvert,
\end{align*}
where
$$R_6:=\left\{\bnabla_{\bbeta}\bar\ell_4(\bgammahat,\bdeltahat,\balphahat,\bbeta^*)-\bnabla_{\bbeta}\bar\ell_4(\bgammahat,\bdeltahat,\balpha^*,\bbeta^*)\right\}^\top\bDelta.$$
Note that $\|\bbeta^*\|_1=\|\bbeta_S^*\|_1\leq\|\bbetahat_S\|_1+\|\bDelta_S\|_1$, $\|\bbetahat\|_1=\|\bbetahat_S\|_1+\|\bbetahat_{S^c}\|_1=\|\bbetahat_S\|_1+\|\bDelta_{S^c}\|_1$, and $\|\bDelta\|_1=\|\bDelta_S\|_1+\|\bDelta_{S^c}\|_1$. Hence, we have
$$2\delta\bar\ell_4(\bgammahat,\bdeltahat,\balphahat,\bbeta^*,\bDelta)+\lambda_{\bbeta}\|\bDelta_{S^c}\|_1\leq3\lambda_{\bbeta}\|\bDelta_S\|_1+2\lvert R_6\rvert.$$
Observe that
\begin{align*}
\lvert R_6\rvert&=\left\lvert2M^{-1}\sum_{i\in\mathcal I_{\bbeta}}A_{1i}\exp(-\bS_{1i}^\top\bgammahat)\left\{1-\frac{A_{2i}}{g(\bar\bS_{2i}^\top\bdeltahat)}\right\}\bar\bS_{2i}^\top(\balphahat-\balpha^*)\bS_{1i}^\top\bDelta\right\rvert\\
&\leq\frac{\delta\bar\ell_4(\bgammahat,\bdeltahat,\balphahat,\bbeta^*,\bDelta)}{2}+R_7,
\end{align*}
where $\delta\bar\ell_4(\bgammahat,\bdeltahat,\balphahat,\bbeta^*,\bDelta)=M^{-1}\sum_{i\in\mathcal I_{\bbeta}}A_{1i}\exp(-\bS_{1i}^\top\bgammahat)(\bS_{1i}^\top\bDelta)^2$ and
$$R_7:=2M^{-1}\sum_{i\in\mathcal I_{\bbeta}}\exp(-\bS_{1i}^\top\bgammahat)\left\{1-\frac{A_{2i}}{g(\bar\bS_{2i}^\top\bdeltahat)}\right\}^2\left\{\bar\bS_{2i}^\top(\balphahat-\balpha^*)\right\}^2.$$
It follows that
\begin{align}\label{eq:basic_delta}
\delta\bar\ell_4(\bgammahat,\bdeltahat,\balphahat,\bbeta^*,\bDelta)+\lambda_{\bbeta}\|\bDelta_{S^c}\|_1\leq3\lambda_{\bbeta}\|\bDelta_S\|_1+2R_7.
\end{align}
Condition on $\|\bgammahat-\bgamma^*\|_2\leq1$, where by Lemma \ref{lemma:preliminary_alpha}, $\|\bgammahat-\bgamma^*\|_2\leq1$ holds with probability $1-o(1)$. Also, condition on the event that 
\begin{align}\label{eq:basic_delta2}
\delta\bar\ell_4(\bgammahat,\bdeltahat,\balphahat,\bbeta^*,\bDelta)\geq\kappa_1\|\bDelta\|_2^2-\kappa_2\frac{\log d_1}{M}\|\bDelta\|_1^2,
\end{align}
which, by Lemma \ref{lemma:RSC}, holds with probability $1-o(1)$. Since $\delta\bar\ell_4(\bgammahat,\bdeltahat,\balphahat,\bbeta^*,\bDelta)\geq0$, by \eqref{eq:basic_delta}, we have
\begin{align}\label{eq:basic_delta4}
\|\bDelta\|_1\leq4\|\bDelta_S\|_1+2\lambda_{\bbeta}^{-1}R_7.
\end{align}
Note that $\|\bDelta_S\|_1\leq\sqrt{s_{\bbeta}}\|\bDelta_S\|_2\leq\sqrt{s_{\bbeta}}\|\bDelta\|_2$. Hence,
\begin{align*}
&3\lambda_{\bbeta}\sqrt{s_{\bbeta}}\|\bDelta\|_2+2R_7\geq3\lambda_{\bbeta}\|\bDelta_S\|_1+2R_7\overset{(i)}{\geq}\delta\bar\ell_4(\bgammahat,\bdeltahat,\balphahat,\bbeta^*,\bDelta)\\
&\qquad\overset{(ii)}{\geq}\kappa_1\|\bDelta\|_2^2-\kappa_2\frac{\log d_1}{M}\|\bDelta\|_1^2\overset{(iii)}{\geq}\kappa_1\|\bDelta\|_2^2-\kappa_2\frac{\log d_1}{M}\left(4\|\bDelta_S\|_1+2\lambda_{\bbeta}^{-1}R_7\right)^2\\
&\qquad\geq\kappa_1\|\bDelta\|_2^2-4\kappa_2\frac{\log d_1}{M}\left(4\|\bDelta_S\|_1^2+\lambda_{\bbeta}^{-2}R_7^2\right)\\
&\qquad\geq\kappa_1\|\bDelta\|_2^2-4\kappa_2\frac{\log d_1}{M}\left(4s_{\bbeta}\|\bDelta\|_2^2+\lambda_{\bbeta}^{-2}R_7^2\right)\geq\frac{\kappa_1}{2}\|\bDelta\|_2^2-\frac{4\kappa_2\log d_1}{\lambda_{\bbeta}^2M}R_7^2,
\end{align*}
when $M>32\kappa_2s_{\bbeta}\log d_1/\kappa_1$. Here, (i) follows from \eqref{eq:basic_delta} and the fact that $\lambda_{\bbeta}\|\bDelta_{S^c}\|_1\geq0$; (ii) holds under the event that \eqref{eq:basic_delta2} occurs; (iii) follows from \eqref{eq:basic_delta4}. By Lemma \ref{lemma:sol},
\begin{align}\label{eq:basic_delta3}
\|\bDelta\|_2\leq\frac{6\lambda_{\bbeta}\sqrt{s_{\bbeta}}}{\kappa_1}+\sqrt{\frac{8\kappa_2R_7^2\log d_1}{\kappa_1\lambda_{\bbeta}^2M}+\frac{4R_7}{\kappa_1}}.
\end{align}
Now, we upper bound the term $R_7$. Observe that
\begin{align*}
\E_{\S_{\bbeta}}(R_7)&=2\E\left[\exp(-\bS_1^\top\bgammahat)\left\{1-\frac{A_2}{g(\bar\bS_2^\top\bdeltahat)}\right\}^2\left\{\bar\bS_2^\top(\balphahat-\balpha^*)\right\}^2\right]\\
&\overset{(i)}{\leq}2\left\|\exp(-\bS_1^\top\bgammahat)\right\|_{\P,3}\left\|1-\frac{A_2}{g(\bar\bS_2^\top\bdeltahat)}\right\|_{\P,6}^2\left\|\bar\bS_2^\top(\balphahat-\balpha^*)\right\|_{\P,6}^2\\
&\overset{(ii)}{\leq}2\left\|\exp(-\bS_1^\top\bgammahat)\right\|_{\P,3}\left\{1+\left\|g^{-1}(\bar\bS_2^\top\bdeltahat)\right\|_{\P,6}\right\}^2\left\|\bar\bS_2^\top(\balphahat-\balpha^*)\right\|_{\P,6}^2
\\
&\overset{(iii)}{=}O_p\left(\frac{s_{\balpha}\log d}{N}\right),
\end{align*}
where (i) holds by H\"older's inequality; (ii) holds by Minkowski inequality; (iii) holds by Lemmas \ref{lemma:preliminary_alpha}, \ref{lemma:preliminary_beta}, and \ref{lemma:preliminary_gamma}. By Lemma \ref{l1},
$$R_7=O_p\left(\frac{s_{\balpha}\log d}{N}\right).$$
By \eqref{eq:basic_delta3} and since $\lambda_{\bbeta}\asymp\sqrt{\log d_1/N}$, we have
$$\|\bDelta\|_2=O_p\left(\sqrt\frac{s_{\bbeta}\log d_1}{N}+\frac{s_{\balpha}\log d}{N}+\sqrt\frac{s_{\balpha}\log d}{N}\right)=O_p\left(\sqrt\frac{s_{\balpha}\log d+s_{\bbeta}\log d_1}{N}\right).$$
By \eqref{eq:basic_delta4},
$$\|\bDelta\|_1\leq4\sqrt{s_{\bbeta}}\|\bDelta\|_2+2\lambda_{\bbeta}^{-1}R_7=O_p\left(s_{\bbeta}\sqrt\frac{\log d_1}{N}+s_{\gamma}\sqrt\frac{(\log d)^2}{N\log d_1}\right).$$

\vspace{0.5em}

(d) Let $\rho(\cdot)=\rho^*(\cdot)$ and $\mu(\cdot)=\mu^*(\cdot)$. Then, by Lemma \ref{lemma:gradient'}, when $s_{\bgamma}=o(N/(\log d_1)^2)$, $s_{\bdelta}=o(N/(\log d_1\log d))$, and $s_{\balpha}=o(N/(\log d_1\log d))$,
$$\left\|\bnabla_{\bbeta}\bar\ell_4(\bgammahat,\bdelta^*,\balphahat,\bbeta^*)\right\|_\infty=O_p\left(\sqrt\frac{\log d_1}{N}\right).$$
That is, for any $t>0$, there exists some $\lambda_4\asymp\sqrt{\log d_1/N}$ such that 
$$\mathcal E_4:=\{\|\bnabla_{\bbeta}\bar\ell_4(\bgammahat,\bdelta^*,\balphahat,\bbeta^*)\|_\infty\leq\lambda_4\}$$
holds with probability at least $1-t$. Condition on the event $\mathcal E_4$, and choose some $\lambda_{\bbeta}>2\lambda_4$. Similarly as in part (c), we obtain
$$2\delta\bar\ell_4(\bgammahat,\bdeltahat,\balphahat,\bbeta^*,\bDelta)+\lambda_{\bbeta}\|\bDelta_{S^c}\|_1\leq3\lambda_{\bbeta}\|\bDelta_S\|_1+2\lvert R_8\rvert,$$
where
\begin{align*}
\lvert R_8\rvert&=\left\lvert\left\{\bnabla_{\bbeta}\bar\ell_4(\bgammahat,\bdeltahat,\balphahat,\bbeta^*)-\bnabla_{\bbeta}\bar\ell_4(\bgammahat,\bdelta^*,\balphahat,\bbeta^*)\right\}^\top\bDelta\right\rvert\\
&=\left\lvert2M^{-1}\sum_{i\in\mathcal I_{\bbeta}}A_{1i}A_{2i}\exp(-\bS_{1i}^\top\bgammahat)\left\{g^{-1}(\bar\bS_{2i}^\top\bdeltahat)-g^{-1}(\bar\bS_{2i}^\top\bdelta^*)\right\}\varepsilonhat_i\bS_{1i}^\top\bDelta\right\rvert\\
&\leq\frac{\delta\bar\ell_4(\bgammahat,\bdeltahat,\balphahat,\bbeta^*,\bDelta)}{2}+R_9.
\end{align*}
Here, $\varepsilonhat_i:=Y_i(1,1)-\bar\bS_{2i}^\top\balphahat$,
\begin{align*}
\delta\bar\ell_4(\bgammahat,\bdeltahat,\balphahat,\bbeta^*,\bDelta)=&M^{-1}\sum_{i\in\mathcal I_{\bbeta}}A_{1i}\exp(-\bS_{1i}^\top\bgammahat)(\bS_{1i}^\top\bDelta)^2,\\
R_9:=&2M^{-1}\sum_{i\in\mathcal I_{\bbeta}}A_{2i}\exp(-\bS_{1i}^\top\bgammahat)\left\{g^{-1}(\bar\bS_{2i}^\top\bdeltahat)-g^{-1}(\bar\bS_{2i}^\top\bdelta^*)\right\}^2\varepsilonhat_i^2.
\end{align*}
Observe that
\begin{align*}
\E_{\S_{\bbeta}}(R_9)&=2\E\left[A_2\exp(-\bS_1^\top\bgammahat)\left\{g^{-1}(\bar\bS_2^\top\bdeltahat)-g^{-1}(\bar\bS_2^\top\bdelta^*)\right\}^2\varepsilonhat^2\right]\\
&\leq2\left\|\exp(-\bS_1^\top\bgammahat)\right\|_{\P,3}\left\|g^{-1}(\bar\bS_2^\top\bdeltahat)-g^{-1}(\bar\bS_2^\top\bdelta^*)\right\|_{\P,6}^2\left\|\varepsilonhat\right\|_{\P,6}^2\\
&\overset{(i)}{=}O_p\left(\frac{s_{\bdelta}\log d}{N}\right),
\end{align*}
where (i) holds by Lemmas \ref{lemma:preliminary_alpha}, \ref{lemma:preliminary_beta}, and \ref{lemma:preliminary_gamma}. By Lemma \ref{l1},
$$R_9=O_p\left(\frac{s_{\bdelta}\log d}{N}\right).$$
Repeat the same procedure as in part (c), we have
\begin{align*}
\|\bDelta\|_2&\leq\frac{6\lambda_{\bbeta}\sqrt{s_{\bbeta}}}{\kappa_1}+\sqrt{\frac{8\kappa_2R_9^2\log d_1}{\kappa_1\lambda_{\bbeta}^2M}+\frac{4R_9}{\kappa_1}},\\
\|\bDelta\|_1&\leq4\sqrt{s_{\bbeta}}\|\bDelta\|_2+2\lambda_{\bbeta}^{-1}R_9,
\end{align*}
with probability at least $1-t-o(1)$. Hence,
\begin{align*}
\|\bDelta\|_2&=O_p\left(\sqrt\frac{s_{\bdelta}\log d+s_{\bbeta}\log d_1}{N}\right),\\
\|\bDelta\|_1&=O_p\left(s_{\bdelta}\sqrt\frac{(\log d)^2}{N\log d_1}+s_{\bbeta}\sqrt\frac{\log d_1}{N}\right).
\end{align*}

\vspace{0.5em}

(e) Let $\rho(\cdot)=\rho^*(\cdot)$, $\nu(\cdot)=\nu^*(\cdot)$, and $\mu(\cdot)=\mu^*(\cdot)$. Then, by Lemma \ref{lemma:gradient'}, when $s_{\bgamma}=o(N/(\log d_1)^2)$, $s_{\bdelta}=o(N/(\log d_1\log d))$, and $s_{\balpha}=o(N/(\log d_1\log d))$,
\begin{align*}
\left\|\bnabla_{\bbeta}\bar\ell_4(\bgammahat,\bdeltahat,\balpha^*,\bbeta^*)\right\|_\infty&=O_p\left(\sqrt\frac{\log d_1}{N}\right),\\
\left\|\bnabla_{\bbeta}\bar\ell_4(\bgammahat,\bdelta^*,\balphahat,\bbeta^*)\right\|_\infty&=O_p\left(\sqrt\frac{\log d_1}{N}\right),\\
\left\|\bnabla_{\bbeta}\bar\ell_4(\bgammahat,\bdelta^*,\balpha^*,\bbeta^*)\right\|_\infty&=O_p\left(\sqrt\frac{\log d_1}{N}\right).
\end{align*}
Define 
$$\ba:=\bnabla_{\bbeta}\bar\ell_4(\bgammahat,\bdeltahat,\balpha^*,\bbeta^*)+\bnabla_{\bbeta}\bar\ell_4(\bgammahat,\bdelta^*,\balphahat,\bbeta^*)-\bnabla_{\bbeta}\bar\ell_4(\bgammahat,\bdelta^*,\balpha^*,\bbeta^*).$$ 
Then $\|\ba\|_\infty=O_p(\sqrt\frac{\log d_1}{N})$. Hence, for any $t>0$, there exists some $\lambda_5\asymp\sqrt{\log d_1/N}$ such that $\mathcal E_5:=\{\|\ba\|_\infty\leq\lambda_5\}$ holds with probability at least $1-t$. Condition on the event $\mathcal E_5$, and choose some $\lambda_{\bbeta}>2\lambda_5$. Similarly as in parts (c) and (d), we obtain
$$2\delta\bar\ell_4(\bgammahat,\bdeltahat,\balphahat,\bbeta^*,\bDelta)+\lambda_{\bbeta}\|\bDelta_{S^c}\|_1\leq3\lambda_{\bbeta}\|\bDelta_S\|_1+2\lvert R_{10}\rvert,$$
where
\begin{align*}
R_{10}&=\left\{\bnabla_{\bbeta}\bar\ell_4(\bgammahat,\bdeltahat,\balphahat,\bbeta^*)-\ba\right\}^\top\bDelta\\
&=\left\{\bnabla_{\bbeta}\bar\ell_4(\bgammahat,\bdeltahat,\balphahat,\bbeta^*)-\bnabla_{\bbeta}\bar\ell_4(\bgammahat,\bdeltahat,\balpha^*,\bbeta^*)\right\}^\top\bDelta\\
&\qquad-\left\{\bnabla_{\bbeta}\bar\ell_4(\bgammahat,\bdelta^*,\balphahat,\bbeta^*)-\bnabla_{\bbeta}\bar\ell_4(\bgammahat,\bdelta^*,\balpha^*,\bbeta^*)\right\}^\top\bDelta\\
&=2M^{-1}\sum_{i\in\mathcal I_{\bbeta}}A_{1i}A_{2i}\exp(-\bS_{1i}^\top\bgammahat)\left\{g^{-1}(\bar\bS_{2i}^\top\bdeltahat)-g^{-1}(\bar\bS_{2i}^\top\bdelta^*)\right\}\cdot\bar\bS_{2i}^\top(\balphahat-\balpha^*)\bS_{1i}^\top\bDelta.
\end{align*}
By Young's inequality for products,
\begin{align*}
\lvert R_{10}\rvert\leq\frac{\delta\bar\ell_4(\bgammahat,\bdeltahat,\balphahat,\bbeta^*,\bDelta)}{2}+R_{11},
\end{align*}
where $\delta\bar\ell_4(\bgammahat,\bdeltahat,\balphahat,\bbeta^*,\bDelta)=M^{-1}\sum_{i\in\mathcal I_{\bbeta}}A_{1i}\exp(-\bS_{1i}^\top\bgammahat)(\bS_{1i}^\top\bDelta)^2$ and
\begin{align*}
R_{11}:=2M^{-1}\sum_{i\in\mathcal I_{\bbeta}}A_{2i}\exp(-\bS_{1i}^\top\bgammahat)\left\{g^{-1}(\bar\bS_{2i}^\top\bdeltahat)-g^{-1}(\bar\bS_{2i}^\top\bdelta^*)\right\}^2\left\{\bar\bS_{2i}^\top(\balphahat-\balpha^*)\right\}^2.
\end{align*}
Observe that
\begin{align*}
\E_{\S_{\bbeta}}(R_{11})&=2\E\left[A_2\exp(-\bS_1^\top\bgammahat)\left\{g^{-1}(\bar\bS_2^\top\bdeltahat)-g^{-1}(\bar\bS_2^\top\bdelta^*)\right\}^2\left\{\bar\bS_2^\top(\balphahat-\balpha^*)\right\}^2\right]\\
&\leq2\left\|\exp(-\bS_1^\top\bgammahat)\right\|_{\P,3}\left\|g^{-1}(\bar\bS_2^\top\bdeltahat)-g^{-1}(\bar\bS_2^\top\bdelta^*)\right\|_{\P,6}^2\left\|\bar\bS_2^\top(\balphahat-\balpha^*)\right\|_{\P,6}^2\\
&\overset{(i)}{=}O_p\left(\frac{s_{\bdelta}s_{\balpha}(\log d)^2}{N^2}\right),
\end{align*}
where (i) holds by Lemmas \ref{lemma:preliminary_alpha}, \ref{lemma:preliminary_beta}, and \ref{lemma:preliminary_gamma}. By Lemma \ref{l1},
$$R_{11}=O_p\left(\frac{s_{\bdelta}s_{\balpha}(\log d)^2}{N^2}\right).$$
Repeat the same procedure as in parts (c) and (d), we have
\begin{align*}
\|\bDelta\|_2&\leq\frac{6\lambda_{\bbeta}\sqrt{s_{\bbeta}}}{\kappa_1}+\sqrt{\frac{8\kappa_2R_{11}^2\log d_1}{\kappa_1\lambda_{\bbeta}^2M}+\frac{4R_{11}}{\kappa_1}},\\
\|\bDelta\|_1&\leq4\sqrt{s_{\bbeta}}\|\bDelta\|_2+2\lambda_{\bbeta}^{-1}R_{11},
\end{align*}
with probability at least $1-t-o(1)$. Hence,
\begin{align*}
\|\bDelta\|_2&=O_p\left(\frac{\sqrt{s_{\bdelta}s_{\balpha}}\log d}{N}+\sqrt\frac{s_{\bbeta}\log d_1}{N}\right),\\
\|\bDelta\|_1&=O_p\left(\frac{s_{\bdelta}s_{\balpha}\log d}{N}\sqrt\frac{(\log d)^2}{N\log d_1}+s_{\bbeta}\sqrt\frac{\log d_1}{N}\right).
\end{align*}

\end{proof}

\section{Proofs of the auxiliary Lemmas}\label{sec:proof_lemmas}
\begin{proof}[Proof of Lemma \ref{lemma:beta1}.]
We prove the lemma by considering two cases separately.

 (a) If $d\leq m$. Choose $S=\{1,\dots,d\}$. Since $\bX$ is a sub-Gaussian vector, we have
\begin{align}\label{B.7_1}
\sup_{\|\bbeta\|_2=1}\E\{(\bX^\top\bbeta)^2\}=O(1).
\end{align}
For any $\bDelta\in\R^d$, by triangle inequality, we have
\begin{align*}
&m^{-1}\sum_{i=1}^m(\bX_i^\top\bDelta)^2\leq\|\bDelta\|_2^2\sup_{\|\bbeta\|_2=1}m^{-1}\sum_{i=1}^m(\bX_i^\top\bbeta)^2\\
&\qquad\leq\|\bDelta\|_2^2\left[\sup_{\|\bbeta\|_2=1}\E\{(\bX^\top\bbeta)^2\}+\sup_{\|\bbeta\|_2=1}\left\lvert m^{-1}\sum_{i=1}^m(\bX_i^\top\bbeta)^2-\E\{(\bX^\top\bbeta)^2\}\right\rvert\right]
\end{align*}
It follows that
\begin{align*}
&\sup_{\bDelta\in\R^{d}/\{\bzero\}}\frac{m^{-1}\sum_{i=1}^m(\bX_i^\top\bDelta)^2}{\|\bDelta\|_2^2}\\
&\qquad\leq\left[\sup_{\|\bbeta\|_2=1}\E\{(\bX^\top\bbeta)^2\}+\sup_{\|\bbeta\|_2=1}\left\lvert m^{-1}\sum_{i=1}^m(\bX_i^\top\bbeta)^2-\E\{(\bX^\top\bbeta)^2\}\right\rvert\right]
\end{align*}
By Theorem 6.5 of \cite{wainwright2019high} and \eqref{B.7_1}, we have, as $m,d\to\infty$, 
\begin{align*}
\sup_{\bDelta\in\R^{d}/\{\bzero\}}\frac{m^{-1}\sum_{i=1}^m(\bX_i^\top\bDelta)^2}{\|\bDelta\|_2^2}
=O_p\left(1+\sqrt\frac{d}{m}+\frac{d}{m}\right)
\overset{(i)}{=}O_p(1),
\end{align*}
where (i) holds since $d\leq m$. Hence,
\begin{align*} 
\sup_{\bDelta\in\R^{d}/\{\bzero\}}\frac{m^{-1}\sum_{i=1}^m(\bX_i^\top\bDelta)^2}{\|\bDelta\|_1^2m^{-1}\log d+\|\bDelta\|_2^2}
\leq\sup_{\bDelta\in\R^{d}/\{\bzero\}}\frac{m^{-1}\sum_{i=1}^m(\bX_i^\top\bDelta)^2}{\|\bDelta\|_2^2}
=O_p(1).
\end{align*}
\vspace{0.5em}

 (b) If $d>m$. Choose a set $S\subseteq\{1,\dots,d\}$ satisfying $s:=|S|\asymp m/\log d<d$. For any $\bDelta\in\R^d$, define $\bDeltatil=(\bDeltatil_S^\top,\bDeltatil_{S^c}^\top)^\top\in\R^d$ such that 
$$\bDeltatil_S=s^{-1}\|\bDelta\|_1(1,\dots,1)^\top\in\R^s,\quad\bDeltatil_{S^c}=\bDelta_{S^c}\in\R^{d-s}.$$
Then 
$$\|\bDeltatil_{S^c}\|_1=\|\bDelta_{S^c}\|_1\leq\|\bDelta\|_1=\|\bDeltatil_S\|_1.$$
Hence, $\bDeltatil\in\C(S,3):=\{\bDelta\in\R^d:\|\bDelta_{S^c}\|_1\leq3\|\bDelta_S\|_1\}$. In addition, since $(\bDeltatil-\bDelta)_{S^c}=\bzero\in\R^{d-s}$, we also have $\bDeltatil-\bDelta\in\C(S,3)$. Therefore, by the fact that $(a+b)^2\leq2a^2+2b^2$, we have
\begin{align*}
&m^{-1}\sum_{i=1}^m(\bX_i^\top\bDelta)^2\leq2m^{-1}\sum_{i=1}^m(\bX_i^\top\bDeltatil)^2+2m^{-1}\sum_{i=1}^m\left\{\bX_i^\top(\bDeltatil-\bDelta)\right\}^2\\
&\qquad\leq2\left(\|\bDeltatil\|_2^2+\|\bDeltatil-\bDelta\|_2^2\right)\sup_{\bbeta\in\C(S,3)\cap\|\bbeta\|_2=1}m^{-1}\sum_{i=1}^m(\bX_i^\top\bbeta)^2.
\end{align*}
Now, we observe that
\begin{align*}
\|\bDeltatil\|_2^2&=\|\bDeltatil_S\|_2^2+\|\bDeltatil_{S^c}\|_2^2=s^{-1}\|\bDelta\|_1^2+\|\bDelta_{S^c}\|_2^2,\\
\|\bDeltatil-\bDelta\|_2^2&=\|\bDeltatil_S-\bDelta_S\|_2^2\leq2\|\bDeltatil_S\|_2^2+2\|\bDelta_S\|_2^2=2s^{-1}\|\bDelta\|_1^2+2\|\bDelta_S\|_2^2.
\end{align*}
Hence, we have for all $\bDelta\in\R^d$,
\begin{align*}
m^{-1}\sum_{i=1}^m(\bX_i^\top\bDelta)^2
\leq2\left(3s^{-1}\|\bDelta\|_1^2+2\|\bDelta\|_2^2\right)\sup_{\bbeta\in\C(S,3)\cap\|\bbeta\|_2=1}m^{-1}\sum_{i=1}^m(\bX_i^\top\bbeta)^2,
\end{align*}
since
$\|\bDeltatil\|_2^2+\|\bDeltatil-\bDelta\|_2^2\leq3s^{-1}\|\bDelta\|_1^2+2\|\bDelta\|_2^2.$
It follows that
\begin{align*}
\sup_{\bDelta\in\R^{d}/\{\bzero\}}\frac{m^{-1}\sum_{i=1}^m(\bX_i^\top\bDelta)^2}{6s^{-1}\|\bDelta\|_1^2+4\|\bDelta\|_2^2}
\leq\sup_{\bbeta\in\C(S,3)\cap\|\bbeta\|_2=1}m^{-1}\sum_{i=1}^m(\bX_i^\top\bbeta)^2,
\end{align*}
By Lemma \ref{lemma:beta0} and \eqref{B.7_1}, as $m,d\to\infty$, 
\begin{align*}
\sup_{\bDelta\in\R^{d}/\{\bzero\}}\frac{m^{-1}\sum_{i=1}^m(\bX_i^\top\bDelta)^2}{s^{-1}\|\bDelta\|_1^2+\|\bDelta\|_2^2}
=O_p\left(1+\sqrt\frac{s\log d}{m}+\frac{s\log d}{m}\right)=O_p(1),
\end{align*}
since $s\asymp m/\log d$.
\end{proof}

\begin{proof}[Proof of Lemma \ref{lemma:sol}.]
\begin{align*}
x\leq\frac{b+\sqrt{b^2+4ac}}{2a}\leq\frac{b+\sqrt{b^2}+\sqrt{4ac}}{2a}=\frac{b}{a}+\sqrt\frac{c}{a}.
\end{align*}
\end{proof}

\begin{proof}[Proof of Lemma \ref{lemma:preliminary_alpha}.]
Let $\mathcal X$ the support of $\bS_1$. Under Assumption \ref{cond:basic}, for all $\bS_1\in\mathcal X$, there exists some constant $c>0$ such that
$$\exp(\bS_1^\top\bgamma^*)\leq c,\quad\exp(-\bS_1^\top\bgamma^*)<g^{-1}(\bS_1^\top\bgamma^*)\leq c.$$
By Theorem \ref{thm:nuisance},
\begin{align*}
\|\bgammahat-\bgamma^*\|_2=O_p\left(\sqrt\frac{s_{\bgamma}\log d_1}{N}\right).
\end{align*}
Since $\bS_1$ is a sub-Gaussian random vector under Assumption \ref{cond:subG}, by Theorem 2.6 of \cite{wainwright2019high},
$$\left\|\bS_1^\top(\bgammahat-\bgamma^*)\right\|_{\P,r}=O\left(\|\bgammahat-\bgamma^*\|_2\right)=O_p\left(\sqrt\frac{s_{\bgamma}\log d_1}{N}\right).$$
Additionally, note that $s_{\bgamma}=o(N/\log d_1)$. It follows that
$$\P_{\S_{\bgamma}}(\|\bgammahat-\bgamma^*\|_2\leq1)=1-o(1).$$
For any $\bgamma\in\{w\bgamma^*+(1-w)\bgammahat:w\in[0,1]\}$, we have
\begin{align}
&\left\|g^{-1}(\bS_1^\top\bgamma)-g^{-1}(\bS_1^\top\bgamma^*)\right\|_{\P,r}=\left\|\exp(-\bS_1^\top\bgamma^*)\left[\exp\left\{-\bS_1^\top(\bgamma-\bgamma^*)\right\}-1\right]\right\|_{\P,r}\nonumber\\
&\qquad\leq c\left\|\exp\left\{-\bS_1^\top(\bgamma-\bgamma^*)\right\}-1\right\|_{\P,r}\label{eq:diff1}
\end{align}
By Taylor's Thorem, for any $\bS_1\in\mathcal X$, with some $v\in(0,1)$,
\begin{align}
&\left\lvert\exp\left\{-\bS_1^\top(\bgamma-\bgamma^*)\right\}-1\right\rvert=\exp\left\{-v\bS_1^\top(\bgamma-\bgamma^*)\right\}\left\lvert\bS_1^\top(\bgamma-\bgamma^*)\right\rvert\nonumber\\
&\qquad\leq\left[1+\exp\left\{-\bS_1^\top(\bgamma-\bgamma^*)\right\}\right]\left\lvert\bS_1^\top(\bgamma-\bgamma^*)\right\rvert.\label{eq:diff1'}
\end{align}
Condition on the event $\|\bgammahat-\bgamma^*\|_2\leq1$. Note that $\bgamma-\bgamma^*=(1-w)(\bgammahat-\bgamma^*)$ and $1-w\in[0,1]$, we have
\begin{align}
&\left\|g^{-1}(\bS_1^\top\bgamma)-g^{-1}(\bS_1^\top\bgamma^*)\right\|_{\P,r}\overset{(i)}{\leq}c\left\|\left[1+\exp\left\{-\bS_1^\top(\bgamma-\bgamma^*)\right\}\right]\bS_1^\top(\bgamma-\bgamma^*)\right\|_{\P,r}\nonumber\\
&\qquad\overset{(ii)}{\leq}c\left\|\bS_1^\top(\bgamma-\bgamma^*)\right\|_{\P,r}+c\left\|\exp\left\{-\bS_1^\top(\bgamma-\bgamma^*)\right\}\right\|_{\P,2r}\left\|\bS_1^\top(\bgamma-\bgamma^*)\right\|_{\P,2r}\nonumber\\
&\qquad\overset{(iii)}{=}O\left(\|\bgammahat-\bgamma^*\|_2\right),\label{rate:invps_alpha'}
\end{align}
where (i) holds by \eqref{eq:diff1} and \eqref{eq:diff1'}; (ii) holds by Minkowski inequality and H\"older's inequality; (iii) holds by Theorem 2.6 of \cite{wainwright2019high} under Assumption \ref{cond:subG}. It follows that,
$$\left\|g^{-1}(\bS_1^\top\bgamma)\right\|_{\P,r}\leq\left\|g^{-1}(\bS_1^\top\bgamma^*)\right\|_{\P,r}+O\left(\|\bgammahat-\bgamma^*\|_2\right)\leq C,$$
with some constant $C>0$, since $\|\bgammahat-\bgamma^*\|_2\leq1$. Therefore, we conclude that $\P_{\S_{\bgamma}}(\mathcal E_1)=1-o(1)$. Moreover, by the fact that $\exp(-u)=g^{-1}(u)-1<g^{-1}(u)$ and $\|X\|_{\P,r'}\leq\|X\|_{\P,12}$ for any $X\in\R$ and $1\leq r'\leq12$, we have
$$\left\|g^{-1}(\bS_1^\top\bgamma)\right\|_{\P,r'}\leq C,\quad\left\|\exp(-\bS_1^\top\bgamma)\right\|_{\P,r'}\leq C.$$
Moreoever, we have \eqref{rate:invps_alpha}, since $\bgammahat\in\{w\bgamma^*+(1-w)\bgammahat:w\in[0,1]\}$, $\P_{\S_{\bgamma}}(\mathcal E_1)=1-o(1)$, and \eqref{rate:invps_alpha'} holds. Besides, note that
\begin{align*}
&\left\|\exp(\bS_1^\top\bgamma)-\exp(\bS_1^\top\bgamma^*)\right\|_{\P,r'}\leq c\left\|\exp\{\bS_1^\top(\bgamma-\bgamma^*)\}-1\right\|_{\P,r'}\\
&\qquad\leq c\left[\left\|\exp\{\bS_1^\top(\bgamma-\bgamma^*)\}\right\|_{\P,r'}+1\right]=O(1),
\end{align*}
since $\bS_1$ is sub-Gaussian and $\|\bgamma-\bgamma^*\|_2\leq1$. Therefore,
\begin{align*}
&\left\|\exp(\bS_1^\top\bgamma)\right\|_{\P,r'}\leq\left\|\exp(\bS_1^\top\bgamma^*)\right\|_{\P,r'}+\left\|\exp(\bS_1^\top\bgamma)-\exp(\bS_1^\top\bgamma^*)\right\|_{\P,r'}\\
&\qquad\leq c+O(1)=O(1).
\end{align*}
\end{proof}

\begin{proof}[Proof of Lemma \ref{lemma:preliminary_beta}.]
Let $\mathcal S$ the support of $\bar\bS_2$. Under Assumption \ref{cond:basic}, there exists some constant $c>0$, such that, for all $\bar\bS_2\in\mathcal S$,
$$\exp(\bar\bS_2^\top\bdelta^*)\leq c,\quad\exp(-\bar\bS_2^\top\bdelta^*)<g^{-1}(\bar\bS_2^\top\bdelta^*)\leq c.$$

 (a) By Theorem \ref{thm:nuisance},
$$\|\bdeltahat-\bdelta^*\|_2=O_p\left(\sqrt\frac{s_{\bgamma}\log d_1+s_{\bdelta}\log d}{N}\right)=o_p(1).$$
By Assumption \ref{cond:subG} and Theorem 2.6 of \cite{wainwright2019high},
$$\left\|\bar\bS_2^\top(\bdeltahat-\bdelta^*)\right\|_{\P,r}=O\left(\|\bdeltahat-\bdelta^*\|_2\right)=O_p\left(\sqrt\frac{s_{\bgamma}\log d_1+s_{\bdelta}\log d}{N}\right).$$

\vspace{0.5em}

 (b) By Theorem \ref{thm:nuisance'},
\begin{align*}
\|\bdeltahat-\bdelta^*\|_2=O_p\left(\sqrt\frac{s_{\bdelta}\log d}{N}\right)=o_p(1).
\end{align*}
Similarly, by Assumption \ref{cond:subG} and Theorem 2.6 of \cite{wainwright2019high},
$$\left\|\bar\bS_2^\top(\bdeltahat-\bdelta^*)\right\|_{\P,r}=O\left(\|\bdeltahat-\bdelta^*\|_2\right)=O_p\left(\sqrt\frac{s_{\bdelta}\log d}{N}\right).$$
The remaining proof is an analog of the proof of Lemma \ref{lemma:preliminary_alpha}. 
\end{proof}

\begin{proof}[Proof of Lemma \ref{lemma:preliminary_gamma}.]
The upper bounds for $\|\bar\bS_2^\top(\balphahat-\balpha^*)\|_{\P,r}$ follow directly from Theorems \ref{thm:nuisance}, \ref{thm:nuisance'}, Theorem 2.6 of \cite{wainwright2019high}, and the sub-Gaussianity of $\bar\bS_2$ assumed in Assumption \ref{cond:subG}. Let either (a) or (b) holds. Then we have $\|\bar\bS_2^\top(\balphahat-\balpha^*)\|_{\P,r}=o_p(1)$. Note that, $\balphatil-\balpha^*=(1-v_1)(\balphahat-\balpha^*)$. Therefore,
\begin{align*}
\|\varepsilontil\|_{\P,r}&\leq\|\varepsilon\|_{\P,r}+\|\bar\bS_2^\top(\balphatil-\balpha^*)\|_{\P,r}=\|\varepsilon\|_{\P,r}+(1-v_1)\|\bar\bS_2^\top(\balphahat-\balpha^*)\|_{\P,r}\\
&=O(1)+o_p(1)=O_p(1).
\end{align*}
\end{proof}

\begin{proof}[Proof of Lemma \ref{lemma:preliminary_delta}.]
The upper bounds for $\left\|\bS_1^\top(\bbetahat-\bbeta^*)\right\|_{\P,r}$ follow directly by Theorems \ref{thm:nuisance}, \ref{thm:nuisance'}, Theorem 2.6 of \cite{wainwright2019high}, and the sub-Gaussianity of $\bS_1$ assumed in Assumption \ref{cond:subG}. Let either (a) or (b) of Lemma \ref{lemma:preliminary_delta} holds, and let either (a) or (b) of \ref{lemma:preliminary_gamma} holds. Then we have $\|\bS_1^\top(\bbetahat-\bbeta^*)\|_{\P,r}=o_p(1)$ and $\|\bar\bS_2^\top(\balphahat-\balpha^*)\|_{\P,r}=o_p(1)$. Note that, $\balphatil-\balpha^*=(1-v_1)(\balphahat-\balpha^*)$ and $\bbetatil-\bbeta^*=(1-v_2)(\bbetahat-\bbeta^*)$. Therefore,
\begin{align*}
\|\zetatil\|_{\P,r}&\leq\|\zeta\|_{\P,r}+\|\bS_1^\top(\bbetatil-\bbeta^*)\|_{\P,r}+\|\bar\bS_2^\top(\balphatil-\balpha^*)\|_{\P,r}\\
&=\|\zeta\|_{\P,r}+(1-v_1)\|\bS_1^\top(\bbetahat-\bbeta^*)\|_{\P,r}+(1-v_2)\|\bar\bS_2^\top(\balphahat-\balpha^*)\|_{\P,r}\\
&=O(1)+o_p(1)=O_p(1).
\end{align*}
\end{proof}

\end{document}